\documentclass[fleqn,11pt]{article}

\usepackage{amsfonts,amsbsy,amsmath,amssymb,enumerate,amsfonts, color,
array,graphicx,latexsym,booktabs, amsthm}
\usepackage{bm}
\usepackage{booktabs}
\usepackage{amsthm}
\usepackage{caption}
\usepackage{array}
\usepackage{fancyhdr}
\usepackage{tabularx}
\usepackage{multicol}
\usepackage{longtable}
\usepackage{threeparttable}
\usepackage{titlesec}
\usepackage{titletoc}
\usepackage{hyperref}
\usepackage{booktabs}
\usepackage{subfig}
\usepackage{array}
\usepackage{fancyhdr}
\usepackage{tabularx}
\usepackage{lscape}
\usepackage{pdflscape} 
\usepackage{supertabular}
\usepackage{multirow}
\usepackage{indentfirst}
\usepackage{makecell}
\usepackage{graphicx}
\usepackage{float}
\usepackage{graphicx}
\usepackage{caption}
\usepackage{subfig}
\usepackage{setspace}
\usepackage{adjustbox}
\newtheorem{thm}{Theorem}[section]
\newtheorem{assumption}{Assumption}[section]

\newtheorem{lemma}{Lemma}[section]

\renewcommand{\thesection}{\arabic{section}}

\newcommand{\rar}{\rightarrow}
\newcommand{\E}{\mathbb{E}}
\renewcommand{\P}{\mathbb{P}}
\newcommand{\mcalK}{\mathcal{K}}
\newcommand{\reals}{\mathbb{R}}
\newcommand{\mcalXbar}{\overline{\mathcal{X}}_1}
\newcommand{\eps}{\varepsilon}
\newcommand{\indep}{\perp\!\!\!\perp}

\usepackage{hyperref}

\RequirePackage[round,authoryear,longnamesfirst,sort]{natbib}

\hypersetup{citecolor=true,colorlinks=true,allcolors=blue}
\allowdisplaybreaks

\usepackage{xr}
\begin{document}
\title{Estimation of Conditional Average Treatment Effects with
High-Dimensional Data\thanks{We are grateful to Songnian Chen, Edward Kennedy, Whitney Newey, Ryo Okui, Yoon-Jae Whang and participants of California Econometrics Conference (Irvine), Cemmap-WISE joint workshop, AMES 2019, for valuable comments.}}

\author{Qingliang Fan\thanks{Department of Economics, The Chinese University of Hong Kong. E-mail: michaelqfan@gmail.com. Fan acknowledges financial support from National Natural Science Foundation of China Grant No. 71671149.} \and
Yu-Chin Hsu\thanks{Institute of Economics, Academia Sinica, Taiwan. Department of
Finance, National Central University and Department of Economics, University Chengchi University. E-mail: ychsu@econ.sinica.edu.tw. Yu-Chin Hsu gratefully acknowledges research support from the Ministry of Science and Technology of Taiwan(107-2410-H-001-034-MY3) and the Career Development Award of Academia Sinica, Taiwan.}
\and Robert P. Lieli\thanks{Department of Economics, Central European University,
Budapest.  E-mail: lielir@ceu.edu.} \and
Yichong Zhang\thanks{School of Economics, Singapore Management University. E-mail: yczhang@smu.edu.sg. Zhang acknowledges financial support from the Singapore Ministry of Education Tier 2 grant under grant MOE2018-T2-2-169 and the Lee Kong Chian Fellowship.}}

\maketitle

\begin{abstract}
Given the unconfoundedness assumption, we propose new nonparametric estimators for the reduced dimensional conditional average treatment effect (CATE) function. In the first stage, the nuisance functions necessary for identifying CATE are estimated by machine learning methods, allowing the number of covariates to be comparable to or larger than the sample size.  
The second stage consists of a low-dimensional local linear regression, reducing CATE to a function of the covariate(s) of interest. We consider two variants of the estimator depending on whether the nuisance functions are estimated over the full sample or over a hold-out sample. Building on Belloni at al.\ (2017) and Chernozhukov et al.\ (2018), we derive functional limit theory for the estimators and provide an easy-to-implement procedure for uniform inference based on the multiplier bootstrap. The empirical application revisits the effect of maternal smoking on a baby's birth weight as a function of the mother's age.

\noindent \textbf{Keywords:} Heterogeneous treatment effects; High-dimensional data; Uniform confidence band.

\noindent \textbf{JEL codes:} C14, C21, C55
\end{abstract}

\section{Introduction}

In settings with individual-level treatment effect heterogeneity, the unconfoundedness assumption theoretically permits identification and consistent estimation of the conditional average treatment effect (CATE) for all possible values of the set of covariates $X$ used in adjusting for selection bias. One way to think about these covariates is that they are ex-ante predictors of an individual's potential outcomes with and without treatment, and hence are highly correlated with the treatment participation decision as well. Unconfoundedness states that the econometrician observes all relevant predictors  so that conditional on $X$, the treatment takeup decision is no longer statistically related to the potential outcomes.\footnote{This condition was formalized by \cite{RR83}; since then, unconfoundedness (or ``selection on observables'' or ``conditional independence'') has become one of the standard paradigms for modeling selection effects. See, e.g., \cite{IW09} for further discussion.} Nevertheless, in many situations the individual deciding on treatment participation is likely to have access to private signals about their potential outcomes. Relying on the unconfoundedness assumption amounts to hoping that a set of publicly observed characteristics can still proxy for the information content of these signals. Therefore, the unconfoundedness assumption is more plausible in applications in which $X$ is a rich, detailed set of covariates, i.e., the dimension of $X$ is high.

While CATE as a function of $X$ provides a detailed characterization of treatment effect heterogeneity across observable subpopulations, this information is very hard to analyze and convey if $X$ is high-dimensional. Of course, one could examine slices of this function along some component(s) $X_1$ of $X$ while holding the other components $X_{-1}$ of $X$ constant. Nevertheless, how CATE varies as a function of $X_1$ will generally depend on the level at which $X_{-1}$ is held constant, requiring the examination of (infinitely) many different slices. For this reason, instead of holding the variables in $X_{-1}$ constant, \cite{AHL15} (henceforth AHL) suggest integrating them out with respect to the conditional distribution of $X_{-1}$ given $X_1$ or, in practice, a smoothed estimate of this distribution. This gives rise to a reduced dimensional CATE function that is easier to present and interpret.\footnote{If all covariates $X$ are integrated out, one obtains an estimator of ATE as in \cite{H98} or \cite{HIR03}.}

In this paper we propose two-step estimators of the reduced dimensional CATE function where in the first step the required high-dimensional nuisance regressions are conducted by machine learning methods designed specifically to handle such problems, while the second integration step is implemented by a traditional local linear nonparametric regression.\footnote{This step assumes that $X_1$ is a continuous variable, which is the technically challenging and interesting case.} We derive the statistical properties of two variants of the estimator. In the first case, the first step (nuisance function estimation) and the second step (local linear regression) are both implemented over the full sample of available observations. In the second case, the available sample is split into parts, and the first step is implemented in one subsample while the second step is done in the complement sample. The roles of the subsamples are then rotated and the results are averaged. This is the ``cross-fitting'' approach to machine-learning-aided causal inference advocated by \cite{CC18}. The first approach is used by \cite{BCFH13} in estimating unconditional treatment effects.

In proposing and studying these estimators, we contribute to two recent strands of the econometrics literature. First, we advance the currently available flexible methods for the estimation of reduced dimensional CATE functions due to AHL and \cite{LOW16} (henceforth LOW). Second, we make technical contributions to the recent literature that employs machine learning methods in tackling the prediction component of causal inference problems; see, e.g., \cite{BCH14b,BCH14a}, \cite{BCFH13}, \cite{CC17}, \cite{CC18}. Taking a broader perspective, our paper is also related to a large statistics literature on regular estimation and the use of orthogonal (doubly robust) moment conditions.

Regarding the first set of papers, AHL use an inverse probability weighted conditional moment of the data to identify CATE. They consider both kernel-based and parametric estimation of the propensity score in the first step, and derive the asymptotic distribution of the estimated CATE function evaluated at a fixed point $x_1$ in the support of $X_1$. LOW advance these results in two respects: their estimator is based on a Neyman-orthogonal moment condition and they also provide a method for uniform inference about the CATE function as a whole (rather than point by point). While LOW only use parametric models to estimate the nuisance functions involved in the moment condition, orthogonality lends their estimator a ``double robustness'' property: either the model for the propensity score or the models for the conditional means of the potential outcomes are allowed to be misspecified (but not both).

The CATE estimator proposed here is based on the same orthogonal moment condition as in LOW, but the required nuisance functions are estimated by machine learning methods, which allow for data-driven flexible functional forms as well as a (very) high-dimensional set of covariates. Neyman orthogonality is crucial in ensuring that the proposed CATE estimators are robust to the regularization bias inherent in the first stage, making post-selection inference possible. As the asymptotic theory is derived from high-level assumptions, there are a number of applicable first-stage estimation methods in practice, such as a random forest or $\ell_1$-penalized lasso or post-lasso. In this paper we use lasso estimation as the leading example.

In light of the discussion of the unconfoundedness assumption above, replacing the parametric estimators in LOW with machine learning methods greatly enhances the applicability and empirical relevance of flexible CATE estimation. At the same time, the asymptotic theory remains tractable: we provide methods for pointwise as well as uniform inference about the CATE function under both the full sample and sample-splitting implementation schemes. The uniform methods utilize the multiplier bootstrap, while pointwise inference can be based either on the bootstrap or the analytic results.

Turning to the literature on machine learning in treatment effect estimation, we build primarily on  \cite{BCFH13} for the full-sample method and \cite{CC17,CC18} for the split-sample method, while providing the necessary extension of the theory to account for the use of local linear regression in second step. In these papers the parameter of interest is identified by the restriction that the unconditional expectation of a ``score function'' evaluated at the true parameter value (and the true nuisance functions) is zero. By contrast, the identifying restriction in our case is that the \emph{conditional} expectation of the same score function is zero. Hence, our estimation procedure does not simply consist of substituting in the estimated nuisance functions and setting the sample average score to zero; instead, the score function will enter a local linear regression with kernel weight $\mathcal{K}((X_{1i}-x_1)/h)/h$ on each observation, 
where $h$ denotes a smoothing parameter (bandwidth).

The key high-level assumptions we employ in deriving our asymptotic results involve bounding the $L_\infty$ norm of the difference between the true and estimated nuisance functions, and the $L_2$ norm of the same difference multiplied by the kernel. The rates at which these error bounds are required to converge to zero are closely linked to the rate at which the bandwidth sequence converges to zero. From a purely technical standpoint, incorporating the bandwidth conditions into the high-level norm bounds in the full-sample as well as the cross-fitting case is a central contribution of the paper. Similarly to AHL and LOW, the resulting convergence rate of the CATE estimators is $\sqrt{Nh^d}$, where $N$ is the sample size and $d=dim(X_1)$. 

In addition to the error bounds, the full-sample estimator also requires controlling the complexity (entropy) of the function space in which the nuisance functions take values. In the case of lasso estimation, this can be accomplished by restricting how fast the number of covariates and the sparsity indices associated with the nuisance functions are allowed to  increase with the sample size. These conditions are more stringent than in the case of estimating ATE.

There are several papers in the broader statistics literature that have considered estimation problems related to ours \citep{R04, L13, Luedtke2016a, Luedtke2016b, NW17, L19}.\footnote{We thank Edward Kennedy and an anonymous referee for these references.} \cite{NW17}, in particular, estimate the full-dimensional CATE function in a data-rich environment using penalized regression, and establish the quasi-oracle error bounds for their estimator. While we also use a high-dimensional set of covariates and machine learning methods to deal with selection into treatment, the ultimate parameter of interest, being a function of a low-dimensional subset of the covariates, is then targeted by a traditional nonparametric estimator. We also complement \cite{NW17} by establishing both pointwise and uniform inference procedures. In a related paper, \cite{ZL19} consider the local constant estimation of CATE in the high-dimensional setting but only provide pointwise asymptotic results.

Another closely related paper, \cite{CS19}, proposes an approach to CATE estimation that also includes a dimension-reduction step. There are, however, substantial technical differences between their paper and ours. First, the traditional nonparametric estimator used by  \cite{CS19} in the second stage is series regression rather than local linear regression. Second, they only consider the cross-fitting approach and do not address the problem of estimating both the nuisance functions and the target function on the full sample. Third, we also provide a reasonably detailed discussion of the primitive conditions under which lasso estimation fulfills the high-level conditions posited in the paper, while \cite{CS19} restrict attention to high-level analysis.

Finally, our estimation method based on doubly robust moments is tied to the classic literature on regular estimation and semiparametric efficiency \citep{B83, P90, B93, N94, V00}. As mentioned above, in a parametric setting, estimation of CATE based on doubly robust moments is consistent as long as either the treatment assignment process or the outcome processes is correctly specified. If both processes are nonparametrically estimated, the method can achieve a faster convergence rate than the nuisance estimators employed. The use of doubly robust methods for causal inference has also been considered by \cite{RR95}, \cite{H98}, \cite{VR03}, \cite{HIR03}, \cite{V06}, \cite{f07}, \cite{t07}, \cite{VR11}, \cite{BCFH13}, \cite{F15}, \cite{K17}, \cite{R17}, \cite{WA18}, and \cite{suz19}, among others.

In addition to providing theoretical results, we study and illustrate our methods through Monte Carlo simulations. The proposed estimators perform well in terms of bias, MSE, and coverage rates. In general, we find that the cross-fitting estimator has somewhat better finite sample properties than the full sample estimator, and thus we suggest using the cross-fitting estimator in empirical studies with reasonably large sample sizes.

Our application uses vital statistics data from North Carolina to estimate the effect of a (first-time) mother's smoking during pregnancy on the baby's birth weight as a function of the mother's age. Despite a number of previous analyses, the application is well worth revisiting with the help of machine learning methods, as there are a large number of covariates describing the mother's characteristics and events during pregnancy, and the specification of the propensity score is known to have a substantial impact on the results (see AHL, Section 4.2). Our results provide some corroborating evidence that the negative effect of smoking on birth weight becomes more detrimental with age. This pattern is less prevalent than some of the results reported in AHL but stronger than that found by LOW.

The rest of the paper proceeds as follows. In Section 2 we describe the formal setup and the estimators. Section 3 states and discusses the assumptions underlying the first-order asymptotic theory and provides the main results. Section 4 describes how to conduct uniform inference using the multiplier bootstrap. The application is presented in Section 5, while Section 6 concludes. An online supplement contains additional empirical studies, the Monte Carlo exercise as well as detailed proofs of the theoretical results.

\section{The formal framework, identification and the estimators}
\label{sec: estimators}

Population units are characterized by a random vector $(D, Y(1),Y(0), X)$, where $D\in\{0,1\}$ indicates the receipt of a binary treatment, $Y(1)$ and $Y(0)$ are the potential outcomes with and without the treatment, respectively, and $X$ is a vector of pre-treatment covariates. The observed variables are given by the vector $W=(D, Y, X)$, where $Y=DY(1)+(1-D)Y(0)$. The distribution of $(D, Y(1),Y(0), X)$, and hence $W$, is induced by an underlying probability measure $\P$; parameter values computed under $\P$ will be denoted by the subscript ``0'' and represent the true values of these parameters. The expectation operator corresponding to $\P$ is denoted by $\E$, but we also use the linear functional notation $\P f:=\int f(w)d\P=\E[f(W)]$.

In order to accommodate high-dimensional data, we follow the conceptual considerations in \cite{F15} and treat the DGP (the measure $\P$) as dependent on the sample size $N$, allowing, in particular, the dimension of $X$ to grow with $N$.\footnote{This implies that the nuisance functions $\mu_0(j,X)$, $j=0,1$ and $\pi_0(X)$, to be defined below, may generally depend on $N$ as well.} This has two practical interpretations. First, the number of raw controls can already be comparable to the sample size or, second, $X$ may be composed of a large dictionary of sieve bases derived from a fixed dimensional vector $X^*$ through suitable transformations (e.g., powers and interactions). Thus, the high dimensionality of $X$ can also stem from the desire to provide a flexible approximation to the required nuisance functions. We explicitly allow for the use of lasso-type methods in the first stage that select a smaller subset of terms from the dictionary to approximate these functions.

To ease the already heavy notational burden in the paper, the dependence of the DGP on the sample size is left implicit throughout, but is of course accounted for in the theoretical analysis. Most arguments in the paper are based on concentration inequalities, which are non-asymptotic in nature. In our Assumption \ref{assn: 1st stage full} below, we also take into account the fact that, as the dimension of $X$ grows, the complexity of the first-stage estimator will generally diverge, which can affect the rate of convergence of our second-stage estimator. Furthermore, we establish the uniform inference results using the multiplier bootstrap based on the strong approximation theory developed by \cite{CCK14}, which does not require the existence of an asymptotic distribution.

Given a $d$-dimensional subvector $X_1\subset X$ composed of continuous variables, the reduced dimensional CATE function is defined as
\[
\tau_0(x_1)=CATE(x_1)=\E[Y(1)-Y(0)|X_1=x_1].\footnote{The most relevant case in practice is $d=1$ or perhaps $d=2$, otherwise the motivating properties of the reduced dimensional CATE function (interpretabilty and presentability) are lost. As we will see below, under a fourth-moment condition on $Y$, the general theory requires $d\le 3$. There are no restrictions on $d$ for bounded outcomes.}
\]

The identification of $\tau_0(x_1)$ from the joint distribution of $W$ is facilitated by the unconfoundedness assumption along with some technical conditions:
\begin{assumption}\label{assn: uc}
The distribution $\P$ satisfies:
\begin{itemize}
    \item [(i)] (Unconfoundedness) $(Y(1), Y(0))\perp D \big | X$.
    \item [(ii)] (Moments) $\E\big[|Y(j)|^q\big]<\infty$, $j=0,1$ and $q\ge 4$. 
    \item [(iii)] (Propensity score) Let $\pi_0(x)=\P(D=1|X=x)$. There exists some constant $\underline{C}>0$ such that $\P(\underline{C}\le \pi_0(X)\le 1-\underline{C})=1$.

\end{itemize}
\end{assumption}

Assumption \ref{assn: uc}(i) is the standard unconfoundedness condition. Although we are interested in CATE for a low-dimensional subset $X_1$ of the covariates, we still use the full vector of $X$ to address selection into treatment. Allowing for $X$ to be high-dimensional makes it more plausible to have conditional independence between the potential outcomes and the treatment indicator.  Assumption \ref{assn: uc}(ii) is a usual sufficient condition for the estimation of standard errors. Assumption \ref{assn: uc}(iii) is the overlapping support condition commonly assumed in the literature. We also need it to establish that our CATE$(x_1)$ estimator converges at the usual nonparametric rate. 

Let $\mu_0(j,x)=\E[Y|X=x, D=j]$, $j=0,1$. It follows immediately from Assumption~\ref{thm: id} that $E[Y(j)|X_1=x_1]=E[\mu_0(j,X)|X_1=x_1]$, and hence $\tau_0(x_1)$ is identified as
$
\tau_0(x_1)=\E[\mu_0(1,X)-\mu_0(0,X)|X_1=x_1].
$

We now state a less obvious but more robust result based on a Neyman-orthogonal moment condition. Given any probability measure satisfying Assumption \ref{assn: uc}, let $\tau(\cdot)$, $\mu(1,\cdot)$, $\mu(0,\cdot)$, $\pi(\cdot)$ denote the functions corresponding to $\tau_0(\cdot)$, $\mu_0(1,\cdot)$, $\mu_0(0,\cdot)$, $\pi_0(\cdot)$, respectively. Let $\eta=(\pi(\cdot),\mu(1,\cdot),\mu(0,\cdot))$ represent the infinite dimensional nuisance parameters needed to identify CATE, and define
\[
\psi(W;\eta)=\frac{D(Y-\mu(1,X))}{\pi(X)}+\mu(1,X)-\frac{(1-D)(Y-\mu(0,X))}{1-\pi(X)}-\mu(0,X).
\]
The following theorem gives a moment condition that is (at least approximately) satisfied at $(\tau_0, \eta)$ even when $\eta$ deviates from $\eta_0$.

\begin{thm}\label{thm: id} (i) Under Assumption \ref{assn: uc},
\begin{eqnarray*}
&&\E\left[\frac{D(Y-\mu_0(1,X))}{\pi_0(X)}+\mu_0(1,X)\Big|X_1=x_1\right] = \E[Y(1)\mid X_1=x_1]\\
&&\E\left[\frac{(1-D)(Y-\mu_0(0,X))}{1-\pi_0(X)}+\mu_0(0,X)\Big|X_1=x_1\right] = \E[Y(0)\mid X_1=x_1]
\end{eqnarray*}
for all $x_1$ in the support of $X_1$.

(ii)  $\E[\psi(W;\eta_0)-\tau_0(X_1)|X_1=x_1]=0$ by part (i), and this moment equation satisfies the Neyman-orthogonality condition
\begin{equation}\label{eqn: Neyman}
\partial_r \E\big[\psi\big(W;\eta_0+r(\eta-\eta_0)\big)-\tau_0(X_1)\big|X_1=x_1\big]\big|_{r=0}=0.
\end{equation}
\end{thm}

\paragraph{Remarks:}

\begin{itemize}

    \item [1.] Assumption \ref{assn: uc}(iii) is not necessary for Theorem~\ref{thm: id}; a weaker moment condition such as $\E[1/\pi^2_0(X)]<\infty$ would suffice. Nevertheless, the overlap condition stated under Assumption \ref{assn: uc}(iii) is indispensable for subsequent results concerned with the asymptotic distribution of our CATE estimators. Similarly, for identification only, the fourth moment condition in  Assumption~\ref{assn: uc}(ii) could be replaced by a second moment condition.

    \item [2.] If $\eta=(\pi_0, \mu(0,\cdot), \mu(1,\cdot))$ or $\eta=(\pi, \mu_0(0,\cdot), \mu_0(1,\cdot))$, i.e., $\eta$ deviates from $\eta_0$ along one set of coordinates at a time, then $\E\big[\psi\big(W;\eta_0+r(\eta-\eta_0)\big)-\tau_0(X_1)\big|X_1=x_1\big]=0$ for any value of $r$, which of course implies (\ref{eqn: Neyman}). This is the ``double robustness property'' emphasized by LOW; it implies that if $\pi(\cdot)$ and $(\mu(0,\cdot), \mu(1,\cdot))$ are parametric models for $\pi_0(\cdot)$ and $(\mu_0(0,\cdot), \mu_0(1,\cdot))$, respectively, and one of these models is misspecified, then one can still consistently estimate $\tau_0(x_1)$ based on the moment condition $\E\big[\psi(W;\eta)-\tau_0(X_1)\big|X_1=x_1\big]=0$.

\end{itemize}

The following assumption describes the properties and use of the sample data:

\begin{assumption}\label{assn: sample}
	\begin{itemize}
		\item [(i)] The observed data consist of $N$ independent and identically distributed (i.i.d.) random vectors $\{W_i\}_{i=1}^N=\{(D_i,Y_i,X_i)\}_{i=1}^N$with the same distribution as the population distribution of $W$.
		
		\item [(ii)] Let $K$ be a (small) positive integer, and (for simplicity) suppose that $n=N/K$ is also an integer. Let $I_1,\ldots,I_K$ be a random partition of the index set $I=\{1,\ldots, N\}$ so that $\# I_k=n$ for $k=1,\ldots K$.
		
	\end{itemize}
	
\end{assumption}

We now propose two versions of the CATE estimator, depending on whether the first-stage approximation to $\eta_0$ and the second-stage local linear regression targeting $\tau_0$ take place over the same sample or not.
\begin{itemize}
\item {\bf The full-sample estimator:}\\
Let $\hat\eta(I)=(\hat\mu(0,\cdot; I), \hat\mu(1,\cdot; I), \hat \pi(\cdot; I))$, where $\hat\mu(0,\cdot; I)$, $\hat\mu(1,\cdot; I)$ and  $\hat \pi(\cdot; I)$ are
estimators of $\mu_0(0,\cdot)$, $\mu_0(1,\cdot)$ and  $\pi_0(\cdot)$ respectively, over the full sample $I$. Furthermore, let
$\mcalK$ be a $d$-dimensional product kernel, $h$ be a smoothing parameter (bandwidth), and $\mcalK_h\left(u\right)  = \mcalK\left(\frac{u}{h}\right)$.\footnote{To be specific, let $k_j(\cdot)$ be a one-dimensional kernel, then a $d$-dimensional product kernel $\mathcal{K}$ where bandwidth $h$ is defined as $\mathcal{K}_h(u)=\Pi_{j=1}^d k_j(u_j/h)$.  More generally, we can allow $h$ to be different for each $j$ such that $\mathcal{K}_h(u)=\Pi_{j=1}^d k_j(u_j/h_j)$ given that $h_j$'s are of the same order. For notational simplicity, we focus on the first case in the theory.} The second stage of the full-sample estimator $\hat{\tau}(x_1)$ is obtained as the intercept of the local linear regression
\begin{align}
\label{full sample 2nd stage}
(\hat{\tau}(x_1),\hat{\beta}(x_1)) = \arg\min_{a,b} \sum_{i \in I} \left[\psi(W_i,\hat\eta(I)) - a - (X_{1i} - x_1)'b \right]^2\mcalK_h\left(X_{1i}-x_1\right). 
\end{align}

\item {\bf The $K$-fold cross-fitting estimator:}\\
For each $k=1,\ldots, K$, let
$\hat\eta(I_k^c)=(\hat\mu(0,\cdot;I_k^c), \hat\mu(1,\cdot; I_k^c), \hat \pi(\cdot; I_k^c))$, where $\hat\mu(0,\cdot; I_k^c)$, $\hat\mu(1,\cdot; I_k^c)$ and  $\hat \pi(\cdot; I_k^c)$ are
estimators for $\mu_0(0,\cdot)$, $\mu_0(1,\cdot)$ and  $\pi_0(\cdot)$ respectively, constructed over the subsample $I_k^c=I\setminus I_k$.
The second stage of the $K$-fold cross-fitting estimator consists of $K$ nonparametric regressions over the samples $I_1,\ldots, I_K$:
\begin{align}
\label{split sample 2nd stage}
(\hat{\tau}_k(x_1),\hat{\beta}_k(x_1)) = \arg\min_{a,b} \sum_{i \in I_k} \left[\psi(W_i,\hat\eta(I_k^c)) - a - (X_{1i} - x_1)'b \right]^2\mcalK_h\left(X_{1i}-x_1\right). 
\end{align}
Finally, in the third stage we take the average of the $K$ preliminary estimates to obtain an efficient estimator:
$\check\tau(x_1)=\frac{1}{K}\sum_{k=1}^K \hat\tau_k(x_1).$

\end{itemize}

We use the local linear smoother studied by \cite{f92,fg92,f93} to estimate $\tau_0(x_1)$ in the second stage. While it is possible to extend out results to local polynomial estimators with an extra degree of smoothness, we focus on the linear case for simplicity.\footnote{An early version of the paper considered kernel-based (local constant) nonparametric regression in the second stage. The results are available upon request.}

Our second-stage estimator is related to the partial mean estimator studied by \cite{n94b} and \cite{lee18}. However, \cite{lee18} and our paper are distinct in two important ways. First, the parameters of interest, and thus the estimators, are different. We are interested in CATE$(x_1)$ when the treatment variable is binary, while \cite{lee18} considers a model with a continuous treatment. Second, as explained above, our analysis is compatible with the use of high-dimensional data. For the full-sample first-stage estimation, we allow the complexity of our first-stage estimator to increase with the dimensionality of the data and investigate its impact on the rate of convergence. For the split-sample first-stage estimation, we show that the impact of the increasing complexity is eliminated due to the independence between the observations used in the first- and second-stage estimations.  

When $X_1$ is discrete and takes the values $x_{1,1},\cdots,x_{1,M}$, the function CATE($x_{1,m}$), $m = 1,\ldots,M$ can be interpreted as the average treatment effect for the subpopulation $X_1 = x_{1,m}$. In this case one can restrict the sample to observations with $X_1 = x_{1,m}$, and directly apply the full-sample or cross-fitting estimation methods developed in \cite{BCFH13} and \cite{CC17}. 

\section{Asymptotic properties of CATE estimators}
\label{sec:asy}
\subsection{CATE estimators based on general first-step ML estimators }\label{subsec: ML}
In this section we first provide the fundamental asymptotic results for our CATE estimators which form the basis of the uniform inference procedures to be given in Section~\ref{sec: unif infer}. To this end, we state and discuss several assumptions. Let $\mathcal{X}_1\subset \reals^d$ denote the support of $X_1$ and let $\overline{\mathcal{X}}_1$ be the subset of $\mathcal{X}_1$ over which $\tau_0(x_1)$ is to be estimated. In addition, let $f(x_1)$ denote the p.d.f.\ of $X_1$.

\begin{assumption}\label{assn: regularity}
Assume that
\begin{enumerate}

\item [(i)]  The set $\overline{\mathcal{X}}_1$ is contained in the interior of $\mathcal{X}_1$ and is the Cartesian product of closed intervals, i.e., $\overline{\mathcal{X}}_1=\Pi_{j=1}^d[x^{(j)}_{1\ell},x^{(j)}_{1u}]$ with $x^{(j)}_{1\ell}<x^{(j)}_{1u}$. Furthermore, there exist positive constants $\underline{C}$ and $\overline{C}$ such that:
\begin{equation*}
\begin{aligned}
        & \underline{C} \leq \inf_{x_1 \in \overline{\mathcal{X}}_1}f(x_1) \leq \sup_{x_1 \in \overline{\mathcal{X}}_1}f(x_1) \leq \overline{C} \quad \text{and} \quad \sup_{x_1 \in \overline{\mathcal{X}}_1}(|\mathbb{E}[Y(1)|X_1=x_1]|+|\mathbb{E}[Y(0)|X_1=x_1]|) \leq \overline{C}.
        \end{aligned}
        \end{equation*}

\item [(ii)] The functions
$f(x_1)$, $\mathbb{E}[Y(0)|X_1=x_1]$, and $\mathbb{E}[Y(1)|X_1=x_1]$ are twice differentiable with bounded derivatives over $\overline{\mathcal{X}}_1$; more formally,
  \begin{equation*}
        \begin{aligned}
        \sup_{x_1 \in \overline{\mathcal{X}}_1, 1 \leq j,s \leq d}& \biggl(|\partial_{j}f(x_1)| + |\partial_{j,s}f(x_1)| + |\partial_{j}\tau_0(x_1)|+ |\partial_{j,s}\tau_0(x_1)|\biggr) \leq \overline{C},
        \end{aligned}
        \end{equation*}
where $\partial_{j,s}f(x_1)$ is the derivative of $f(x_1)$ w.r.t. $x_{1j}$ and $x_{1s}$.
\item [(iii)] For $u\in\reals^d$, $\mcalK(u)=\kappa(u_1)\times\ldots\times \kappa(u_d)$, where $\kappa$  is a bounded, symmetric p.d.f.\ with $\int t \kappa(t)dt = 0$ and $\int
t^2 \kappa(t)dt = \nu<\infty$. Furthermore, there exists a positive constant
$\overline{C}$ such that $|t|\kappa\left( t\right) \leq
\overline{C}$ for all $t\in\reals$.

\item [(iv)]
The bandwidth $h=h_N$ satisfies $h = C N^{-H}$ for some $H>1/(4+d)$ and $H<(1-2/q)/d$, where $C>0$ and $q$ satisfies Assumption~\ref{assn: uc}(ii).

\item[(v)] Let $\beta_0(x_1) = \partial_{x_1}\tau_0(x_1)$ and $\tau_0^{(2)}(x_1) = \partial_{x_1x_1^T}\tau_0(x_1)$. Then $\sup_{ x_1 \in \overline{\mathcal{X}}_1 }\lambda_{\max}(\tau_0^{(2)}(x_1)) < \overline{C}$, where $\lambda_{\max}(G)$ denotes the maximum singular value of matrix $G$. In addition, we have 
\begin{align*}
\sup_{ x_1,x_1' \in \overline{\mathcal{X}}_1 }\frac{\left|\tau_0(x_1')-\tau_0(x_1) - (x_1'-x_1)^T\beta_0(x_1) - \frac{1}{2}(x_1'-x_1)^T\tau_0^{(2)}(x_1)(x_1'-x_1)\right|}{||x_1'-x_1||_2^3} \leq \overline{C},
\end{align*} 
where $||\cdot||_2$ denotes the Euclidean norm of a vector. 
\end{enumerate}

 \label{ass:regularity}
\end{assumption}

For the most part, Assumption~\ref{assn: regularity} is a collection of standard regularity conditions used in the nonparametric treatment effect estimation literature. The functions $f(x_1)$, $\mu_0(0,x_1)$, and $\mu_0(1,x_1)$ are required to be sufficiently smooth over $\mcalXbar$, the density of $X_1$ must be bounded away from zero over the same set, and the kernel function $\mcalK$ must obey some mild restrictions, satisfied by usual choices of $\kappa$ such as the Gaussian or the Epanechnikov kernel (in the simulations and the empirical study we use the former). Of course, Assumption~\ref{assn: regularity}(ii) also implies that we restrict attention to the technically more interesting case in which the distribution of $X_1$ is continuous, which means that one cannot simply use sample splitting to estimate CATE at various points in the support of $X_1$.

The conditions imposed on the bandwidth in Assumption~\ref{assn: regularity}(iv) are motivated as follows. The restriction $H>1/(4+d)$ means that $h$ converges to zero faster than the MSE-optimal bandwidth choice; this undersmoothing condition is needed to ensure that the bias from the second-stage kernel regression is asymptotically negligible. In addition, we require $H<(1-2/q)/d$ to be able to use a Gaussian approximation as in \citet[Proposition  3.2]{CCK14}. If the outcome variable is bounded, one can set $q=\infty$ in Assumption \ref{assn: uc}(ii) so that $H<1/d$ as in \citet[Proposition  3.1]{CCK14}. If one only assumes $q=4$, then the convergence rate must satisfy $H \in (1/(4+d),1/2d)$. For this interval to be nonempty, $d$ can be at most 3, which is consistent with Assumption 1 in LOW. In principal, it would also be possible to use the optimal bandwidth, i.e., $H = 1/(4+d)$, and conduct bias correction as in \cite{c18}, while accounting for the impact of the estimated bias correction term on the standard error. This approach is, however, beyond the scope of the present paper. A key conceptual difference between our setup and \cite{c18} is that in our case the dependent variable is not directly observed, but is rather constructed based on first-stage nuisance estimators. Finally, Assumption~\ref{assn: regularity}(v) is a standard bound for the Taylor remainder, commonly assumed in local linear regression theory. See, for example, \citet[Chapter 2]{lr07}. 

We now state high-level conditions that specify the convergence rates required of the first-stage nuisance function estimators. The stated rates are linked to the bandwidth sequence $h$ used in the second-stage regressions.  More specifically, we make the following assumption about the full-sample first-stage estimator $\hat\eta(I)$.

\begin{assumption}[Full sample, first stage]\label{assn: 1st stage full}
Let $\delta_{1N}$, $\delta_{2N}$, $\delta_{4N}$ and $A_N$ be sequences of positive numbers,
and $\mathcal{G}_{N}^{(j)}$, $j\in\{0,1,\pi\}$ be classes of real-valued functions defined on the support of $X$ with corresponding envelope functions $G_N^{(j)}$, $j\in\{0,1,\pi\}$. For $\epsilon>0$, let $\mathcal{N}(\mathcal{G}_{N}^{(j)}, \|\cdot\|, \epsilon)$ be the covering number associated with $\mathcal{G}_{N}^{(j)}$ under some norm $\|\cdot\|$ defined on $\mathcal{G}_{N}^{(j)}$.\footnote{The covering number is the minimal number of balls with radius $\epsilon$ needed to cover $\mathcal{G}_{N}^{(j)}$. A ball with radius $\epsilon$ centered on $g$ is the collection of functions $g'\in\mathcal{G}_{N}^{(j)}$ with $\|g'-g\|<\epsilon$.} The following conditions are satisfied.

(i) The estimator $\hat\eta(I)$ obeys the error bounds
\begin{align}
\label{assn: crate L2 full}
 \sup_{ x_1 \in \overline{\mathcal{X}}_1, j=0,1 }&\left\Vert(\hat{\mu}(j,X;I) - \mu_0(j,X))\mathcal{K}_h^{1/2}\left(X_1-x_1\right)\right\Vert_{\P,2} \notag \\
& \times  \left\Vert(\hat\pi(X;I)-\pi_0(X))\mathcal{K}_h^{1/2}\left(X_1-x_1\right)\right\Vert_{\P,2}
= O_p(\delta_{1N}^2)\\
&\sum_{j=0,1}\left\Vert\hat{\mu}(j,\cdot;I) - \mu_0(j,\cdot)\right\Vert_{\P,\infty} + \left\Vert\hat\pi(X; I)-\pi_0(X)\right\Vert_{\P,\infty}= O(\delta_{2N}).
\label{assn: crate Linf full}\\
 \sup_{ x_1 \in \overline{\mathcal{X}}_1,j=0,1}&\left\Vert(\hat{\mu}(j,X;I) - \mu_0(j,X))||X_1-x_1||_2^{1/2}\mathcal{K}_h^{1/2}\left(X_1-x_1\right)\right\Vert_{\P,2} \notag \\
& \times \left\Vert(\hat\pi(X;I)-\pi_0(X))||X_1-x_1||_2^{1/2}\mathcal{K}_h^{1/2}\left(X_1-x_1\right)\right\Vert_{\P,2}
= O_p(\delta_{3N}^2).
\label{assn: crate Linf full 2}
\end{align}

(ii) With probability approaching one,
\begin{align*}
\hat{\mu}(j,\cdot;I) \in \mathcal{G}_{N}^{(j)}, \quad j=0,1, \quad \text{and} \quad \hat{\pi}(\cdot;I) \in \mathcal{G}_{N}^{(\pi)}
\end{align*}
where the classes of functions $\mathcal{G}_{N}^{(j)}$, $j\in\{0,1,\pi\}$ are such that
\begin{align}
\label{eq:G01}
\sup_Q \log \mathcal{N}\big(\mathcal{G}_N^{(j)},||\cdot||_{Q,2},\eps||G_N^{(j)}||_{Q,2}\big) \leq \delta_{4N}(\log(A_N) + \log(1/\eps)\vee 0), \quad j=0,1,\pi
\end{align}
with the supremum taken over all finitely supported discrete probability measures $Q$.

(iii) The sequences $\delta_{1N}$, $\delta_{2N}$, $\delta_{3N}$, $\delta_{4N}$ and $A_N$ satisfy:
\begin{align}
& \min(\delta_{1N}/h^{d/2},\delta_{2N})=o\big((\log(N)Nh^d)^{-1/4}\big), \quad \delta_{2N} = o(1), \label{assn: rate cond 1 full}\\
& \min(\delta_{3N}/h^{d/2+1},\delta_{2N})=o\big((\log(N)Nh^{d+2})^{-1/4}\big) \label{eq:fullx}, \\
& \delta_{4N}\log(A_N \vee N)\log(N)\delta_{2N}^2 = o(1), \label{assn: rate cond 2 full}\\
& \text{and} \quad \delta_{2N}\delta_{4N}\log^{1/2}(N)N^{1/q}\log(A_N \vee N) = o((Nh^d)^{1/2}).\label{assn: rate cond 3 full}
\end{align}
\end{assumption}

\paragraph{Remarks:}
\begin{itemize}
    \item [1.] Part (i) of Assumption \ref{assn: 1st stage full} controls the difference between $\eta_0$ and $\hat\eta(I)$ (i.e., the estimation error) in various norms.

    \item [2.] Part (ii) controls the complexity of the nuisance functions and the estimators through restrictions on the entropy of the classes $\mathcal{G}_{N}^{(j)}$.

    \item [3.] It is of course part (iii) that fills parts (i) and (ii) with content through specifying the behavior of the sequences $\delta_{1N}$, $\delta_{2N}$, $\delta_{4N}$ and $A_N$. In particular, conditions (\ref{assn: rate cond 1 full}) and (\ref{eq:fullx}) extend the fairly standard requirement in semiparametric settings that the first-stage nuisance function estimators converge faster than $N^{-1/4}$; see \cite{AC03} and \cite{BCFH13}. However, in estimating $\tau_0(x_1)$ and $\beta_0(x_1)$, the second-stage kernel regression relies only on observations local to $x_1$, and hence the relevant effective sample size is $Nh^d$ and $Nh^{d+2}$ rather than $N$. The extra $\log(N)$ factor that appears in the required convergence rate is the price to pay for uniform results in $x_1$.

    \item [4.] If the first-stage estimators are based on (correctly specified) parametric models, then, under standard regularity conditions, $\hat\eta(I)$ converges to $\eta_0$ at the rate of $N^{1/2}$ both in $L_2$ and $L_\infty$ norm. Thus, in this case (\ref{assn: crate L2 full}) and (\ref{assn: crate Linf full}) both hold with $\delta_{1N}=\delta_{2N}=N^{-1/2}$ (recall that $\mcalK_h$ is bounded). In addition, conditions \eqref{assn: rate cond 1 full}, \eqref{assn: rate cond 2 full} and \eqref{assn: rate cond 3 full} are also easily satisfied with $\delta_{4N} = O(1)$, and $A_N = O(1)$. This is essentially the setting in LOW (with allowance for partial misspecification).

    \item[5.] Assumption \ref{assn: 1st stage full} imposes rate restrictions on the complexity of the first-stage estimators. The lasso-type regularization method achieves variable selection along with estimation, which greatly reduces the complexity of the estimator. Thus, it is especially suitable for first-stage estimation when using the full sample.
    
   \item[6.] It is possible to establish sufficient conditions on the convergence rate of $\hat{\pi}$, $\hat{\mu}$ 
   and the kernel individually. For example, following \cite{K17}, we can assume 
   \begin{align*}
   \sup_{x_1 \in \overline{\mathcal{X}}_1 }\sqrt{\mathbb{E}\left[\left(\hat{\mu}(j,X;I) - \mu_0(j,X)\right)^2|X_1=x_1\right]\mathbb{E}\left[\left(\hat{\pi}(X;I) - \pi_0(X)\right)^2|X_1=x_1\right]} = O_p(\delta_{5N}^2). 
   \end{align*}
   Note here we require the bound to hold uniformly over $\overline{\mathcal{X}}_1$ rather than a neighborhood of $x_1$, because we aim to conduct uniform inference over $\overline{\mathcal{X}}_1$. Then, it is easy to see that our $\delta_{1N} = \delta_{5N}h^{d/2}$ and $\delta_{3N} = \delta_{5N}h^{(d+1)/2}$. Consequently, \eqref{assn: rate cond 1 full} and \eqref{eq:fullx} reduce to 
   $\delta_{5N} = o\big((\log(N)Nh^d)^{-1/4}\big)$ as $\delta_{2N} \leq \delta_{5N}$. However, this condition is sufficient but necessary. It is possible to bound directly the estimation error of the nuisance parameters weighted by the kernel, as shown in \cite{suz19}. 
    
  \item[7.] Alternatively, because the kernel function is bounded, we can write
   \begin{align*}
   \sup_{ x_1 \in \overline{\mathcal{X}}_1, j=0,1 }&\left\Vert(\hat{\mu}(j,X;I) - \mu_0(j,X))\mathcal{K}_h^{1/2}\left(X_1-x_1\right)\right\Vert_{\P,2} \left\Vert(\hat\pi(X;I)-\pi_0(X))\mathcal{K}_h^{1/2}\left(X_1-x_1\right)\right\Vert_{\P,2} \\
   \leq & M\max_{j=0,1 }\left\Vert(\hat{\mu}(j,X;I) - \mu_0(j,X))\right\Vert_{\P,2}\left\Vert(\hat\pi(X;I)-\pi_0(X))\right\Vert_{\P,2}.
   \end{align*}
   for some constant $M>0$. Thus, one could also state sufficient conditions for \eqref{assn: crate L2 full} solely in terms of the $L_2$-norm of the error bounds associated with $\hat{\mu}(j,X;I)$ and $\hat{\pi}(X;I)$.
    
    \item[8.] 
    Note that we only require bounds on the product of the $L_2$-norms of two estimation errors as in (\ref{assn: crate L2 full}) and (\ref{assn: crate Linf full 2}). The product structure of these conditions  allows for tradeoffs between how fast $\hat\pi(\cdot;I)$ versus $\hat\mu(j,\cdot;I)$ converges.  
\end{itemize}

The corresponding assumption about the cross-fitting (split-sample) estimator is as follows.
\begin{assumption}[Split sample, first stage]\label{assn: 1st stage split}
The split-sample first-stage estimators $\hat\eta(I^c_k)$, $k=1,\ldots, K$ are assumed to satisfy:
 \begin{eqnarray}
 && \sup_{x_1\in\mcalXbar}\Bigg\{\left\|(\hat\pi(X; I_k^c)-\pi_0(X))\mcalK^{1/2}_h\left(X_1-x_1\right)\right\|_{\P_{I_k},2} \notag \\
 &&\qquad\;\,\times \left\|(\hat\mu(j,X; I_k^c)-\mu_0(j,X))\mcalK^{1/2}_h\left(X_1-x_1\right)\right\|_{\P_{I_k},2}\Bigg\}
 = O_p\big(\delta_{1n}^2\big),\label{assn: crate L2 split}\\
 &&||\pi_0(X) - \hat\pi_0(X; I^c_k)||_{\P,\infty}  + \sum_{j=0,1}||\mu_0(j,X) - \hat{\mu}(0,X;I^c_k)||_{\P,\infty} = O(\delta_{2n}),\label{assn: crate Linf split} \\
 &&  \left\Vert (\hat\mu(j,X; I_k^c)-\mu_0(j,X))||X_1 - x_1||_2^{1/2}\mathcal{K}^{1/2}_h\left(X_1 - x_1\right)\right\Vert_{\P_{I_k},2} \nonumber \\
 &&\qquad\;\,\times \left\Vert (\hat\pi(X; I_k^c)-\pi_0(X))||X_1 - x_1||_2^{1/2}\mcalK_h^{1/2}\left(X_1-x_1\right)\right\Vert_{\P_{I_k},2} = O_p(\delta^2_{3n}) \label{assn: crate L2 splitx}
\end{eqnarray}
where  $\P_{I_k}f = \E(f(W_1,\cdots,W_N)|W_i,i\in I_k^c)$ for a generic function $f$, $h^{-d}\delta_{1n}^2 = o((\log(n)nh^d)^{-1/2})$, $\delta_{2n} = o((\log(n))^{-1})$, and $h^{-d-2}\delta_{3n}^2 = o((\log(n)nh^{d+2})^{-1/2})$.
\end{assumption}

\paragraph{Remarks:}
\begin{itemize}
	\item [1.]
	Because $K$ is fixed and $n = N/K$, $\log(N)Nh^d$ and $\log(n)nh^d$ have the same order of magnitude.
	
	\item [2.]
	For the split-sample estimation, there is no requirement on the entropy of the space where the estimated nuisance functions take values.
	This weakening of the theoretical conditions is due to the fact that, because of the cross-fitting technique, we can treat the estimators of the nuisance parameters as fixed by conditioning on the subsample of the data used for the estimation.
	\item[3.] 
	Assumption \ref{assn: 1st stage split} does not impose restrictions on the complexity of the first-stage estimator and thus accommodates various machine learning methods. One can verify Assumption \ref{assn: 1st stage split} given the error bounds of machine learning first-stage estimators in both $L_\infty$ and $L_2$ norms by the same argument as described in Section \ref{subsec: LASSO} below. Deriving these error bounds for various machine learning methods is beyond the scope of our paper. Partial results are available in the literature. For example, the $L_2$ bounds for the random forest method and deep neural networks have already been established in \cite{WA18} and \cite{f18}, respectively. 
	\item[4.] Remarks 6--8 after Assumption \ref{assn: 1st stage full} apply here as well. In essence, Assumption \ref{assn: 1st stage split} is the local analog of Assumption 5.1(f) used by \cite{CC18} to estimate the \emph{unconditional} average treatment effect via the cross-fitting (split-sample) approach.
\item[5.] Similarly to Remark 7 after Assumption \ref{assn: 1st stage full}, one sufficient condition for the requirement on $\delta_{1N}$ is that 
	\begin{align*}
	\sqrt{n/h^d}\max_{j=0,1 }\left\Vert(\hat{\mu}(j,X;I_k^c) - \mu_0(j,X))\right\Vert_{\P,2}\left\Vert(\hat\pi(X;I_k^c)-\pi_0(X))\right\Vert_{\P,2} = o((\log(n)^{-1/2})). 
	\end{align*}
	\cite{CS19} consider sieve estimation of CATE with high-dimensional control variables and require
	\begin{align*}
	\sqrt{nr}\max_{j=0,1 }\left\Vert(\hat{\mu}(j,X;I_k^c) - \mu_0(j,X))\right\Vert_{\P,2}\left\Vert(\hat\pi(X;I_k^c)-\pi_0(X))\right\Vert_{\P,2} = o(1),
	\end{align*}
	where $r$ is the dimension of the sieve bases. In nonparametric estimation, we know the variances of sieve- and kernel-based estimators are of order $r/n$ and $1/(nh^{d})$, respectively. This implies our rate requirement is equivalent to \citet[Assumption 4.4]{CS19} up to some logarithmic factor. The requirement on $\delta_{3N}$ is not essential and can be avoided if one uses the local constant regression instead.
	
\end{itemize}

\begin{thm}
\label{thm:cate}
(a) If Assumptions \ref{assn: uc}, \ref{assn: sample}, \ref{assn: regularity} and \ref{assn: 1st stage full} are satisfied, then
\begin{align*}
\hat{\tau}(x_1) - \tau_0(x_1) = (\P_N - \P)\left[\frac{1}{h^df(x_1)}(\psi(W,\eta_0) - \tau_0(x_1))\mcalK_h\left(X_1-x_1\right)\right] + R_{\tau}(x_1)
\end{align*}
where $\P_Nf = \frac{1}{N}\sum_{i=1}^N f(W_i)$ for a generic function $f(\cdot)$ and $\sup_{ x_1 \in \overline{\mathcal{X}}_1 }|R_{\tau}(x_1)| = o_p( (\log(N)Nh^d)^{-1/2})$.

(b) If Assumptions \ref{assn: uc}, \ref{assn: sample}, \ref{assn: regularity} and \ref{assn: 1st stage split} are satisfied, then the representation established in part (a) also holds for the $K$-fold cross-fitting estimator $\check\tau(x_1)$, i.e.,
\begin{align*}
\check{\tau}(x_1) - \tau_0(x_1) = (\P_N - \P)\left[\frac{1}{h^df(x_1)}(\psi(W,\eta_0) - \tau_0(x_1))\mcalK_h\left(X_1-x_1\right)\right] + \check{R}_\tau(x_1)
\end{align*}
where $\sup_{x_1 \in \overline{\mathcal{X}}_1 }|\check{R}_{\tau}(x_1)| = o_p( (\log(N)Nh^d)^{-1/2})$.
\end{thm}

Theorem \ref{thm:cate} provides the linear (Bahadur) representations of the nonparametric estimators $\hat{\tau}(x_1)$ and $\check \tau(x_1)$ with uniform control of the remainder terms. It serves as a building block for both pointwise and uniform inference about $\tau_0(x_1)$.\footnote{In the  online supplement, we provide the linear (Bahadur) representations of the nonparametric estimators $\hat{\beta}(x_1)$ and $\check \beta(x_1)$ with uniform control of the remainder terms, which can be of independent interest.}
Starting with the former, we define
\[
\sigma^2_N(x_1)=h^d Var\left(\frac{1}{h^df(x_1)}(\psi(W,\eta_0) - \tau_0(x_1))\mcalK_h\left(X_1-x_1\right)\right),
\]
and suppose that $\sigma^2_N(x_1)$ satisfies:
\begin{assumption}\label{assn: cate_var}
There exists some $\underline{C}>0$ such that $\min_{x_1\in\overline{\mathcal{X}}_1}\sigma^2_N(x_1)\geq
\underline{C}$ for all $N$. \label{assn: std}
\end{assumption}


Then Theorem \ref{thm:cate}, together with Lyapunov's CLT, implies
\begin{align}
\frac{\sqrt{Nh^d} \big(\hat{\tau}(x_1) - \tau_0(x_1)\big)}{\sigma_N(x_1)}\stackrel{d}{\rightarrow} \mathcal{N}(0,1)\label{eqn: pwise asy norm}
\end{align}
for any fixed $x_1\in\mcalXbar$. One can estimate the variance $\sigma^2_N(x_1)$ as
\[
\hat\sigma^2_N(x_1)=\frac{1}{Nh^d\hat f^2(x_1; I)}\sum_{i=1}^n\big(\psi(W_i,\hat\eta(I)) - \hat\tau(x_1)\big)^2\mcalK^2_h\left(X_{1i}-x_1\right),
\]
and we will show that
\begin{align*}
\sup_{x_1\in\overline{\mathcal{X}}_1}|\widehat{\sigma}_N(x_1)-{\sigma}_N(x_1)|=o_p(1)~\text{and }
\sup_{x_1\in\overline{\mathcal{X}}_1}|\widehat{\sigma}^{-1}_N(x_1)-{\sigma}^{-1}_N(x_1)|=o_p(1).
\end{align*}
Of course, this means that inference in practice can proceed based on (\ref{eqn: pwise asy norm}) with  $\hat\sigma_N(x_1)$ replacing $\sigma_N(x_1)$. Furthermore, result (\ref{eqn: pwise asy norm}) remains valid if one uses the estimator $\check \tau(x_1)$ in place of $\hat\tau(x_1)$; in this case $\sigma^2_N(x_1)$ can be estimated as
\begin{align*}
&\check\sigma^2_N(x_1) = \frac{1}{K}\sum_{k=1}^K\check\sigma^2_k(x_1),~\text{where}\\
&\check\sigma^2_k(x_1)=\frac{1}{nh^d}\sum_{i \in I_k}\frac{1}{\hat f^2(x_1; I_{k})}\big(\psi(W_i,\hat\eta(I^c_{k})) - \check\tau_{k}(x_1)\big)^2\mcalK^2_h\left(X_{1i}-x_1\right).
\end{align*}

\begin{thm}\label{thm: var consistency}
If Assumptions in Theorems \ref{thm:cate} and Assumption \ref{assn: cate_var} hold, then
\begin{align*}
&\sup_{x_1\in\overline{\mathcal{X}}_1}|\widehat{\sigma}_N(x_1)-{\sigma}_N(x_1)|=o_p(1),~\
\sup_{x_1\in\overline{\mathcal{X}}_1}|\widehat{\sigma}^{-1}_N(x_1)-{\sigma}^{-1}_N(x_1)|=o_p(1),\\
&\sup_{x_1\in\overline{\mathcal{X}}_1}|\check{\sigma}_N(x_1)-{\sigma}_N(x_1)|=o_p(1),~\text{and }
\sup_{x_1\in\overline{\mathcal{X}}_1}|\check{\sigma}^{-1}_N(x_1)-{\sigma}^{-1}_N(x_1)|=o_p(1).
\end{align*}
\end{thm}

As can be seen from the proof, the $o_p(1)$ term actually vanishes polynomially in $N$.

\subsection{CATE estimators based on first-stage lasso estimators}\label{subsec: LASSO}
While the high-level assumptions stated in Section \ref{subsec: ML} can accommodate multiple machine learning procedures for estimating $\eta_0$, here we describe the first stage using lasso estimation as a leading example.
We now discuss some primitive conditions under which lasso estimation of $\eta_0$ will satisfy Assumptions \ref{assn: 1st stage full} and \ref{assn: 1st stage split}.
Specifically, let $b(X)=(b_1(X),\ldots,b_p(X))$ be a dictionary of control terms based on $X$, where $p$ is potentially larger than the sample size $N$ and can grow with $N$.\footnote{To be fully consistent with the general notation, it would be more precise to denote the dictionary as $X=b(X^*) = (b_1(X^*), \ldots, b_p(X^*)) $; see the discussion in the second paragraph of Section 2. We opt for simplicity at a small cost in notational consistency.} Typically,  $b(X)$ consists of $X$, and powers and interactions of the components of $X$. The lasso approximates the nuisance functions $\eta_0$ with linear combinations of the components $b_i(X)$; in particular, for $p$-vectors $\beta$, $\alpha$ and $\theta$, set
\begin{eqnarray}
r_\alpha(x):=\mu_0(0,x)-b(x)'\alpha,\quad r_\beta(x):=\mu_0(1,x)-b(x)'\beta,\quad r_\theta(x):=\pi_0(x)-\Lambda(b(x)'\theta)\label{first stage approx}
\end{eqnarray}
where $\Lambda(\cdot)$ is the logistic c.d.f. A primitive condition that justifies using the lasso is approximate sparsity. Intuitively, this means that it is possible to make the approximation errors $r_\alpha$, $r_\beta$, $r_\theta$ small with just a small number of approximating terms, i.e., with $\alpha$, $\beta$ and $\theta$ having only a handful of non-zero components.\footnote{The linear index structure and approximate sparsity are specific to the lasso; other machine learning methods provide different types of approximations which do not necessarily rely on sparsity.}
The coefficients $\alpha$, $\beta$ and $\theta$ are estimated by penalized least squares or maximum likelihood, where a penalty is imposed for any non-zero component. 

In the lasso computations, we set the tuning parameter to be $2c\sqrt N\Phi^{-1}(1-0.1/(\log(N)2p))$ and $c\sqrt N\Phi^{-1}(1-0.1/(\log(N)4p))$ for the conditional mean and propensity score functions estimation, respectively, following \cite{BCH14a} and \cite{BCFH13}.\footnote{The constant values are usually chosen as $c=1.1$. Also in practice, one could use cross-validations to choose the tuning parameter here.}

To formalize the idea that the dimension $p$ of $b(X)$ is comparable with or larger than the sample size, we let $p=p_N$ be a function of $N$ and allow $p_N$ to grow to infinity as $N$ increases, possibly (much) faster than $N$. For example, one could set $p_N=O(N^\lambda)$ for any $\lambda>0$, but even $\log(p_N)=O(N^\lambda)$ is allowed if $\lambda$ is not too large. The linear approximation errors to the components of $\eta_0$, defined in display (\ref{first stage approx}), will be controlled by the sparsity index $s=s_N$, a nondecreasing sequence of positive numbers potentially converging to infinity with $N$. Also needing control is the upper bound on the components of $b(X)$; to this end, let $\zeta=\zeta_N=\max_{1 \le j\le p_N}\|b_j(X)\|_{\P,\infty}$, and note that $\zeta$ (weakly) increases as $p_N\rar\infty$. The following assumption formalizes the notion of approximate sparsity. 

\begin{assumption}\label{assn: sparsity}
Let $s_\pi$ and $s_\mu$ denote the individual sparsity index sequences associated with $\pi_0(\cdot)$ and $\mu_0(j,\cdot)$, respectively. There exist sequences of coefficients $\alpha=\alpha_N$, $\beta=\beta_N$ and $\theta=\theta_N$ such that the linear approximations defined in (\ref{first stage approx}) satisfy the following conditions.
\begin{itemize}

    \item [(i)] The number of nonzero coefficients is bounded by $s_\mu$ and $s_\pi$, i.e.,
    $\max\{ \| \alpha \|_0, \|\beta\|_0 \}\leq s_\mu$ and $\|\theta\|_0 \leq s_\pi$.

    \item [(ii)] The approximation errors are asymptotically small in the sense that
    \begin{align*}
    &||r_\alpha(X)||_{\P,2} +||r_\beta(X)||_{\P,2} = O\left(\sqrt{s_\mu\log(p)/N}\right),\\
    &||r_\alpha(X)||_{\P,\infty} + ||r_\beta(X)||_{\P,\infty} = O\left(\sqrt{s_\mu^2\zeta^2\log(p)/N}\right),\\
&||r_\theta(X)||_{\P,2} = O\left(\sqrt{s_\pi\log(p)/N}\right),~~
||r_\theta(X)||_{\P,\infty} = O\left(\sqrt{s_\pi^2\zeta^2\log(p)/N}\right),
\end{align*}
where $s_\mu^2\zeta^2\log(p)/N\rar 0$ and $s_\pi^2\zeta^2\log(p)/N\rar 0$ (and therefore $s_\mu\log(p)/N\rar 0$ and $s_\pi\log(p)/N\rar 0$).

%


\end{itemize}
\end{assumption}

Part (i) of Assumption~\ref{assn: sparsity} states that the number of nonzero coefficients in the $b(X)$-based linear approximations to $\eta_0$ is at most $s$. Part (ii) requires that the approximation errors associated with these linear combinations asymptotically vanish both in $L_2$ and $L_\infty$ norm. This generally requires $s\rar\infty$, but $s$ needs to stay small relative to $N$ in the sense that $s^2\zeta^2\log(p)/N\rar 0$. 

Given Assumption~\ref{assn: sparsity} and additional regularity conditions, results by \cite{BCFH13} imply that conditions (\ref{assn: crate L2 full}) and (\ref{assn: crate Linf full}) hold with
\begin{align}\label{eqn: lasso conv rates}
\delta_{1N}^2 \leq & \sup_{ x_1 \in \overline{\mathcal{X}}_1, j=0,1 }\left\Vert(\hat{\mu}(j,X;I) - \mu_0(j,X))\right\Vert_{\P,2} \left\Vert(\hat\pi(X;I)-\pi_0(X))\right\Vert_{\P,2}\leq  (s_\mu s_\pi)^{1/2} \log(p\vee N)/N, \notag \\
\delta^2_{2N}=&(s_\mu^2+s_\pi^2)\zeta^2\log(p\vee N)/N,\text{ and } 
\delta_{3N}^2 \leq  (s_\mu s_\pi)^{1/2} \log(p\vee N)h/N,
\end{align}
where the last inequality holds because $||X_1 - x_1||_2\mathcal{K}_h(X_1-x_1) \lesssim h$.  
Furthermore, \cite{BCFH13} also establish \eqref{eq:G01} with $\delta_{4N} = s$ and $A_N = p$ for the following function classes:
\begin{align*}
	&\mathcal{G}_N^{(j)} = \{b(X)'\beta: ||\beta||_0 \leq \ell_N s_\mu, \quad \sup_{x \in \overline{X}}|b(x)'\beta - \mu_0(j,x)| \leq M\delta_{2N}\}, \quad j=0,1,\\
	&\mathcal{G}_N^{(\pi)} = \{\Lambda(b(X)'\theta): ||\beta||_0 \leq \ell_N s_\pi, \quad \sup_{x \in \overline{X}}|\Lambda(b(x)'\theta) - \pi_0(x)| \leq M\delta_{2N}\},
	\end{align*}	
where $\ell_N$ is some slowly diverging sequence, e.g., $\ell_N=\log(\log(N))$ and $M>0$. (As $\pi_0(\cdot)$, $\mu_0(1,\cdot)$, and $\mu_0(0,\cdot)$ are uniformly bounded, $\mathcal{G}_{N}^{(0)}$, $\mathcal{G}_{N}^{(1)}$, $\mathcal{G}_{N}^{\pi}$ have bounded envelope functions.)
	
Given these results, Assumption \ref{assn: 1st stage full} with first-stage lasso estimation boils down to the following conditions:	
\begin{align}
	& \min\left(\frac{(s_\mu s_\pi)^{1/2}\log(p \vee N) \log^{1/2}(N)}{(Nh^d)^{1/2}}, \frac{\zeta^2(s_\mu^2+s_\pi^2) \log(p \vee N)\log^{1/2}(N)h^{d/2}}{N^{1/2}}\right) = o(1), \notag \\
& \frac{\zeta^2 (s_\mu+s_\pi)^3\log^2(p \vee N) \log(N)}{N} = o(1),  \quad \text{and} \quad \frac{\zeta^2 (s_\mu+s_\pi)^4 \log^3(p \vee N) \log(N)}{N^{2-2/q}h^d} = o(1).
	\end{align}	
	These conditions all hold if $\frac{s_\mu s_\pi\log^2(p \vee N) \log(N)}{Nh^d} = o(1)$ and $\zeta^2 (s_\mu +s_\pi)^2\log(p \vee N)\log(N) = o(N^{1-2/q})$. For example, if $q=4$, $p=O(N^\lambda)$, $\lambda>0$, and $\zeta=O(N^{1/4})$, then $\max(s_\mu,s_\pi)=o(\sqrt{Nh^d})$ is essentially sufficient for Assumption~\ref{assn: 1st stage full}, ignoring logarithmic factors of $N$.
	

By contrast, Assumption \ref{assn: 1st stage split} holds under substantially weaker sparsity conditions. Given the rates in (\ref{eqn: lasso conv rates}), the l.h.s.\ of (\ref{assn: crate L2 split}) is at most of order $O\big(\sqrt{s_\pi}\sqrt{s_\mu}\log(p)/N\big)$, as $\mcalK_h$ is a bounded function. Hence, Assumption \ref{assn: 1st stage split} essentially reduces to $\sqrt{s_\pi}\sqrt{s_\mu}\log(p)/(Nh^d)=o\big((\log(N)Nh^d)^{-1/2}\big)$. Again, setting $p=O(N^\lambda)$, $\lambda>0$, and ignoring the logged factors of $N$ gives $s_\pi s_\mu=o(Nh^d)$. This condition is of course satisfied if $s_\pi=s_\mu=o(\sqrt{Nh^d})$, but there can be tradeoffs between the two sparsity indexes. For example, if $s_\pi=O(1)$, i.e., the propensity score essentially obeys a finite dimensional model linear in parameters, then $s_\mu=o(Nh^d)$ is possible, i.e., $\mu_0(j,\cdot)$ can be a function that is substantially harder to approximate. 
Given Remark 5 after Assumption \ref{assn: 1st stage split}, we can see that our sparsity conditions for Assumption \ref{assn: 1st stage split} are essentially equivalent to those in \cite{CS19}.  On the other hand, \cite{LOW16} is based on parametric first-stage estimators with the dimension of the regressors fixed. Therefore, they do not need sparsity conditions (though one could regard the parametric assumption as an extreme form of sparsity).

Other types of lasso methods such as the group lasso by \cite{F15} and the penalized local least squares and maximum likelihood methods by \cite{suz19} can also be used. One can verify the rate restrictions in a manner similar to the above.

\section{Uniform inference based on the multiplier bootstrap}\label{sec: unif inference}
Turning to uniform inference, one option is to construct uniform confidence bands analytically similarly to LOW. We provide an alternative method based on the multiplier bootstrap. Our multiplier bootstrap procedure is computationally efficient and takes the nuisance function estimators from the first stage as given and only recomputes the nonparametric regression estimator(s) from the second stage. This step simply involves a random rescaling of the terms in the sums (\ref{full sample 2nd stage}) and (\ref{split sample 2nd stage}). As lasso estimation is usually time consuming, our procedure is less costly to implement than, say, a standard nonparametric bootstrap requiring new samples from the original data and recomputing the whole estimator.

To describe the procedure formally, we make the following assumption.

\begin{assumption}
The random variable $\xi$ is independent of $W$ with $\E(\xi)=var(\xi)=1$, and its distribution has sub-exponential tails.\footnote{A random variable $\xi$ has sub-exponential tails
if $\P(|\xi|>x) \leq K\exp(-Cx)$ for every $x$ and some constants
$K$ and $C$.} \label{assn: boot_xi}
\end{assumption}
Assumption \ref{assn: boot_xi} is standard for multiplier bootstrap inference. For example, a normal random variable with unit mean and standard deviation satisfies this assumption. The bootstrap is implemented as follows:
\begin{enumerate}
\item
Compute the first-stage nuisance function estimates $\hat{\mu}(0,x;I)$, $\hat{\mu}(1,x;I)$, $\hat{\pi}(x;I)$ OR $\hat{\mu}(0,x;I_k^c)$, $\hat{\mu}(1,x;I_k^c)$, $\hat{\pi}(x;I_k^c)$, $k=1,\ldots, K$.

\item
Draw an i.i.d.\ sequence $\{\xi_i\}_{i=1}^N$ from the distribution of $\xi$.

\item
Choose the number of bootstrap replications $B$, e.g., $B=1000$. Compute $\hat\tau^b(x_1)$ by the local linear regression, for $b=1, \dots, B$,
\begin{align*}
& (\hat{\tau}^b(x_1),\hat{\beta}^b(x_1)) = \arg\min_{a,b} \sum_{i \in I} \xi_i\left[\psi(W_i,\hat\eta(I)) - a - (X_{1i} - x_1)'b \right]^2\mcalK_h\left(X_{1i}-x_1\right),
\end{align*}
or $\check\tau^b(x_1)=\frac{1}{K}\sum_{k=1}^K\hat\tau_k^b(x_1)$, where for $k = 1,\cdots,K$, 
\begin{align*}
& (\hat{\tau}_k^b(x_1),\hat{\beta}_k^b(x_1)) = \arg\min_{a,b} \sum_{i \in I_k}\xi_i \left[\psi(W_i,\hat\eta(I_k^c)) - a - (X_{1i} - x_1)'b \right]^2\mcalK_h\left(X_{1i}-x_1\right).
\end{align*}
\end{enumerate}

The following theorem is the bootstrap version of Theorem \ref{thm:cate}, and it forms the basis of our inference procedure.

\begin{thm}\label{thm:cate bootstrap}

    (a) If Assumptions \ref{assn: uc}, \ref{assn: sample}, \ref{assn: regularity}, \ref{assn: 1st stage full} and \ref{assn: boot_xi} are satisfied, then
    \begin{align*}
    \hat{\tau}^b(x_1) - \hat\tau(x_1) = (\P_N - \P)\left[\frac{\xi-1}{h^df(x_1)}(\psi(W,\eta_0) - \tau_0(x_1))\mcalK_h\left(X_1-x_1\right)\right] + R^b_\tau(x_1)
    \end{align*}
    where $\sup_{x_1 \in \overline{\mathcal{X}}_1}|R^b_\tau(x_1)| = o_p((\log(N)Nh^d)^{-1/2})$.

    (b) If Assumptions \ref{assn: uc}, \ref{assn: sample}, \ref{assn: regularity}, \ref{assn: 1st stage split} and \ref{assn: boot_xi} are satisfied, the representation established in part (a) also holds for $\check\tau^b(x_1)-\check\tau(x_1)$, i.e.,
       \begin{align*}
    \check{\tau}^b(x_1) - \check\tau(x_1) = (\P_N - \P)\left[\frac{\xi-1}{h^df(x_1)}(\psi(W,\eta_0) - \tau_0(x_1))\mcalK_h\left(X_1-x_1\right)\right] + \check{R}^b_\tau(x_1)
    \end{align*}
    where $\sup_{x_1 \in \overline{\mathcal{X}}_1}|\check{R}^b_\tau(x_1)| = o_p((\log(N)Nh^d)^{-1/2})$.

\end{thm}

Theorem \ref{thm:cate bootstrap} justifies the validity of the multiplier bootstrap in implying that $\sqrt{Nh^d}(\hat{\tau}^b(x_1)-\hat{\tau}_0(x_1))$ converges
in distribution to the limiting distribution of $\sqrt{Nh^d}(\hat{\tau}(x_1)-{\tau}_0(x_1))$ conditional on the sample path (data) with probability 1. Therefore, if Assumption \ref{assn: std} also holds, then, conditional on data,
\begin{align}\label{eqn: pwise boot}
\frac{\sqrt{Nh^d} \big(\hat{\tau}^b(x_1) - \hat{\tau}(x_1)\big)}{\widehat{\sigma}_N(x_1)}\stackrel{d}{\rightarrow} \mathcal{N}(0,1),
\end{align}
The same statements of course hold true if $\check\tau^b(x_1)$ and $\check\tau(x_1)$ replaces $\hat\tau^b(x_1)$ and $\hat\tau(x_1)$, respectively.\footnote{In the online supplement, we also shows similar results regarding $\hat\beta^b(x_1)$ and $\check\beta^b(x_1)$ that might be of separate interest.} In addition to pointwise inference, the uniform control of the error term $R_N^b(\cdot)$ in Theorem \ref{thm:cate bootstrap} makes it possible to employ the multiplier bootstrap for uniform inference. For the rest of the paper, we focus on the inference of $\tau_0(x_1)$. The uniform inference of $\beta_0(x_1)$ can be implemented in the same manner. We propose the following algorithm.

\paragraph{Uniform Confidence Band Implementation Procedure}\label{sec: unif infer}

\begin{enumerate}
\item
Compute $\hat{\tau}(x_1)$ and $\hat{\sigma}_N(x_1)$ for a suitably fine grid over $\mcalXbar$.

\item
Compute $\hat{\tau}^b(x_1)$ over the same grid for $b=1,\ldots, B$ while generating a new set of i.i.d.\ $\mathcal{N}(1,1)$ random variables $\{\xi^b_i\}_{i=1}^N$ in each step $b$.

\item
For $b=1,\ldots, B$, compute
\begin{align*}
&M_b^{\text{1-sided}}=\sup_{x_1\in\overline{\mathcal{X}}_1}\frac{\sqrt{Nh^d}(\hat{\tau}^b(x_1)-\hat{\tau}(x_1))}{\widehat{\sigma}_N(x_1)},~~~
&M_b^{\text{2-sided}}=\sup_{x_1\in\overline{\mathcal{X}}_1}\frac{\sqrt{Nh^d}\big|\hat{\tau}^b(x_1)-\hat{\tau}(x_1)\big|}{\widehat{\sigma}_N(x_1)},
\end{align*}
where the supremum is approximated by the maximum over the chosen grid points.

\item
Given a confidence level $1-\alpha$, find the empirical $(1-\alpha)$
quantile of the sets of numbers $\{M_b^{\text{1-sided}}:b=1,\ldots,B\}$ and
$\{M_b^{\text{2-sided}}:b=1,\ldots,B\}$. Denote these quantiles as $\widehat{C}^{\text{1-sided}}_\alpha$ and $\widehat{C}^{\text{2-sided}}_\alpha$, respectively.

\item
The uniform confidence bands are constructed as
\begin{align*}
I_L&=\Big\{\Big(\hat{\tau}(x_1)-\widehat{C}^{\text{1-sided}}_\alpha \frac{\widehat{\sigma}_N(x_1)}{\sqrt{Nh^d}},\infty \Big):x_1\in\overline{\mathcal{X}}_1\Big\},\\
I_R&=\Big\{\Big(-\infty, \hat{\tau}(x_1)+\widehat{C}^{\text{1-sided}}_\alpha \frac{\widehat{\sigma}_N(x_1)}{\sqrt{Nh^d}}\Big):x_1\in\overline{\mathcal{X}}_1\Big\},\\
I_2&=\Big\{\Big(\hat{\tau}(x_1)-\widehat{C}^{\text{2-sided}}_\alpha \frac{\widehat{\sigma}_N(x_1)}{\sqrt{Nh^d}},\; \hat{\tau}(x_1)+\widehat{C}^{\text{2-sided}}_\alpha \frac{\widehat{\sigma}_N(x_1)}{\sqrt{Nh^d}}\Big):x_1\in\overline{\mathcal{X}}_1\Big\}.
\end{align*}
\end{enumerate}

The following theorem formally states the asymptotic validity of the confidence regions proposed above.

\begin{thm}\label{thm:boot-CS}
If Assumptions \ref{assn: uc},  \ref{assn: sample}, \ref{assn: regularity}, \ref{assn: 1st stage full}, \ref{assn: std} and \ref{assn: boot_xi} are satisfied, then
\[
\lim_{N\rightarrow\infty}\P\big(\tau_0\in I_L \big)=\lim_{N\rightarrow\infty}\P\big(\tau_0\in I_R \big)=\lim_{N\rightarrow\infty}\P\big(\tau_0\in I_2 \big)=1-\alpha.
\]

\end{thm}
\noindent\textbf{Remarks:}
\begin{itemize}

    \item [1.] Theorem~\ref{thm:boot-CS} states that the random confidence bands $I_R$, $I_L$ and $I_2$ contain the \emph{entire} function $\tau_0$ with the prescribed probability $1-\alpha$ in large samples. 

    \item [2.] If the grid in step 1 is chosen to be a single point $x_1$, then the algorithm provides pointwise confidence intervals $I_L(x_1)$, $I_R(x_1)$ and $I_2(x_2)$.

    \item [3.] One can construct uniform confidence bands for $\tau_0$ based on the cross-fitting estimator $\check\tau$ following the exact same steps as above; of course, one needs to replace $\hat\tau$, $\hat\tau^b$ and $\hat\sigma_N$ with $\check\tau$, $\check\tau^b$ and $\check\sigma_N$, respectively.

    \item [4.] 	It is also possible to construct the uniform confidence band by approximating the supremum of the empirical process via a Gumbel distribution. However, the Gumbel approximation is accurate only up to the logarithmic rate, as pointed out by \cite{LOW16}. The bootstrap approximation proposed in this paper has the advantage that the approximation error has a geometric rate of decline and the quality of the approximation is better than that of Gumbel.\footnote{We thank an anonymous referee for this excellent comment.} We also note that both the bootstrap and the Gumbel approximations rely on the linear expansions established in Theorem \ref{thm:cate}. 

\end{itemize}

We discuss bandwidth choice in practice. To obtaining our theoretical results, we require undersmoothing to eliminate bias asymptotically.  When $d=1$ as in the simulations, we suggest setting $h_N=\hat h \times N ^{1/5} \times N^{-2/7}$, where $\hat h=1.06\cdot \hat{\sigma}_{x_1} N^{-1/5}$ and $\hat{\sigma}_{x_1}$ is the estimated standard deviation of $X_1$. The formula for $\hat h$ corresponds to the rule-of-thumb bandwidth with a Gaussian kernel suggested by \cite{Silverman1986}.\footnote{When $d=2$ or $3$, we suggest setting for $j=1,\ldots,d$, $h_{jN}=\hat h_j \times N ^{1/(4+d)} \times N^{-2/(4+3d)}$ and $\hat h_j=1.06\cdot \hat{\sigma}_{x_{1j}} N^{-1/(4+d)}$ and $\hat{\sigma}_{x_{1j}}$ is the estimator of the standard deviation of the $j$-th element of $X_1$.} The bandwidth selection is done with the entire sample even for the cross-fitting method. Also the same selection method is employed in the empirical application.

We investigate the finite sample properties of the proposed high-dimensional CATE estimators and the inference procedure outlined above using Monte Carlo experiments. Because of the space constraint, this material can be found in the online supplement.

\section{Empirical application}\label{sec: empirical}

In this section, we employ the proposed high-dimensional CATE estimators to analyze the average effect of maternal smoking on birth weight while allowing for virtually unrestricted treatment effect heterogeneity conditional on the mother's age. Birth weight has been associated with health and human capital development throughout life (\cite{BDS07}, \cite{AC11}), and maternal smoking is considered to be the most important preventable cause of low birth weight \citep{K87}. In recent studies, AHL and LOW both explored this causal relationship using the CATE approach, and  found different degrees of heterogeneity by age. Using observations from 3,754 white mothers in Pennsylvania, LOW found that the CATE of smoking is decreasing from 17 to around 29 years of age, but they differ from AHL in that the contrast between young and 30-year-old mothers is still not large.\footnote{As the smoking effect is negative, ``decreasing'' means that the detrimental effects of smoking become stronger with age.} 

Our study improves on these previous investigations by considering a much larger pool of covariates and explicitly incorporating a variable selection mechanism into the estimation. This initial pool consists of a vector $X$ of raw covariates as well as technical regressors (powers and interactions) to account for the fact that the functional form of $\pi_0$ and $\mu_0$ is unknown. By contrast, AHL assume that a low dimensional parametric model (known up to its coefficients) is correctly specified for $\pi_0$, while LOW assume that either $\pi_0$ or $\mu_0$ obeys such a model. While we still assume that $\pi_0$ and $\mu_0$ are sparse functions, we let a data-driven procedure (lasso) select the most relevant regressors.

\subsection{Data Description}

We start with the same data set as AHL, composed of vital statistics collected by the North Carolina State Center Health Services, and extract the records of first-time mothers\footnote{The motivation for focusing on first-time mothers is discussed in AHL. In effect, the restricted sample enables more  credible identification of the causal effect, as there cannot be uncaptured feedback from the previous birth experience to the current one.} between 1988 and 2002. The variables include whether the mother smokes (the treatment dummy), the baby's birth weight (the main outcome variable, measured in grams), the parents' socio-economic information, such as age, education, income, race, etc., as well as the mothers' medical and health records. The dataset includes 45 raw covariates and 591,547 observations in total. Table \ref{sumstat} summarizes the most important pre-treatment covariates in the data set.\footnote{We drop some covariates from the analysis for various reasons. For example, the mother's weight gain during pregnancy is arguably not a pre-treatment variable, and the Kessner index of prenatal care is basically a function of the number of prenatal visits and the timing of the first visit.}

\begin{table}[!ht]
\centering
  \caption{\small \textbf{Variable definitions} }
\footnotesize
\renewcommand\arraystretch{0.6}
\begin{threeparttable}
    \begin{tabular}{lcllm{5cm}} \toprule[1pt]
\multicolumn{2}{c}{} &\textbf{Name} & \textbf{Type} &  \textbf{Description} \\
\hline
{\textbf{Outcome Variable}}& & bweight & real number & birth weight(g)\\
\cline{1-2}
{\textbf{Treatment}} && smoke & dummy & Whether mother smokes or not\\
\cline{1-2}
\multirow{19}[2]{*}{\textbf{Covariates}} & \multirow{7}[2]{*}{\textbf{\makecell[c]{Parents \\ Basic \\ Info}}}
&mage & real number$^*$ & mother's age\\
& & meduc & integer & mother's years of schooling \\
& & fage    & integer & father's age \\
& & feduc   & integer   & father's years of schooling \\
& & fagemiss   & integer   & Whether  or not father's age is missing \\
& & married & dummy & Whether or not mother is married  \\
& & popdens & real number & population density in mother's zip code (units/km$^2$) \\
\cline{2-2}
&\multirow{8}[2]{*}{\textbf{\makecell[c]{Mothers' \\ Medical\\ Care \\ \& \\ Health\\ Status}}} &
 prenatal & integer & month of first prenatal visit (=10 if prenatal care is foregone) \\
& & pren\_visits    & integer& number of prenatal visits \\
& & terms    & integer & previous (terminated) pregnancies \\
& & amnio & dummy & Did mother take amniocentesis? \\
& & anemia  & dummy & Did mother suffer from anemia?  \\
& & diabetes & dummy & Did mother suffer from gestational diabetes?  \\
& & hyperpr & dummy & Did mother suffer from hypertension?  \\
& & ultra & dummy   & Did mother take ultrasound exams?  \\
\cline{2-2}
& \multirow{4}{*}{\textbf{Others}} &
male & dummy    & Whether or not baby is male  \\
& & drink   & dummy  & mother's alcohol use \\
& & by88-02 & dummy & 13 birth year dummies (from 1988 to 2002) \\
\hline
\end{tabular}
$^*$Note: mother's age is originally recorded as an integer but for the purposes of this exercise we add a uniform $[-1,1]$ random number to this value to make it a continuous variable. The main results are robust when we add  a uniform [-0.5,0.5] random number to the age variable. See more empirical results in the online supplement.
\end{threeparttable}

\label{sumstat}
\end{table}

\subsection{High-Dimensional CATE Estimation}
In this section we estimate the CATE of maternal smoking on the baby's birth weight with mother's age as the conditioning variable. Following AHL and LOW, we estimate CATE separately for black mothers and white mothers. We only report the estimation results for white mothers in this section; the results for the black mothers can be found in the online supplement. The dependent variable $Y$ is the baby's birth weight measured in grams. The treatment dummy $D$ takes on the value 1 if the mother smokes and 0 otherwise. We start from the set of variables displayed in Table \ref{sumstat}, and construct an even larger dictionary $b(X)$ by adding polynomial terms to account for the unknown form of the nuisance functions in a flexible way. Specifically, we include, up to degree 3, the powers and interaction terms of key dummy variables and continuous and integer covariates. We then end up with 792 covariates in total. 

With such a large set of covariates, it is not clear which variables are important in estimating the CATE function. The true set of variables which belong to the estimating equations is assumed to be sparse, as discussed in the previous sections. We hence apply the lasso method in \cite{BCFH13} to estimate propensity score ($\pi_0$) and conditional mean function ($\mu_0$). We then compute the robust score function $\psi$ for each observation $i$, and run a local linear regression of $\psi_i$ on mother's age evaluated at numerous grid points in the interval [15, 36] (years of age). We use the cross-fitting variant of the estimator, i.e., the nuisance function estimation and the kernel regression
take place in different subsamples, and then these roles are rotated. In the empirical study, we use the same $K$ ($=4$) as in the simulations. Granted that the theoretical property of the proposed $K$-fold cross-fitting estimator is the same as the full-sample estimator in large samples, we recommend using sample-splitting estimator with $K=4$ or 5 following the suggestion of \cite{CC18}. We refer to the resulting point estimates as HDCATE (HD stands for ``high-dimensional'').

AHL include the mother's age, education, month of first prenatal visit (=10 if prenatal care is foregone), number of prenatal visits, and indicators for the baby's gender, the mother's marital status, whether or not the father's age is missing, gestational diabetes, hypertension, amniocentesis, ultrasound exams, previous (terminated) pregnancies, and alcohol use as the confounding factors. The variables selected by our first-step estimation are similar to those used in AHL, with some notable differences.\footnote{Given that we use the cross-fitting method, there are $K=4$ first-stage estimates, and each has its own variable selection. The reported set of variables selected in the first stage is the union of the selected variables in the four split-sample first stages.} In the propensity score function, we also select father's age, and father's education, besides the ones used in AHL, but not gestational diabetes and amniocentesis. In the conditional mean function, we have father's education, and the rest overlap with that of  AHL.

\begin{figure}[htpb]
\begin{minipage}[c]{0.5\linewidth}
\centering

\includegraphics[width=1\textwidth,height=0.3\textheight]{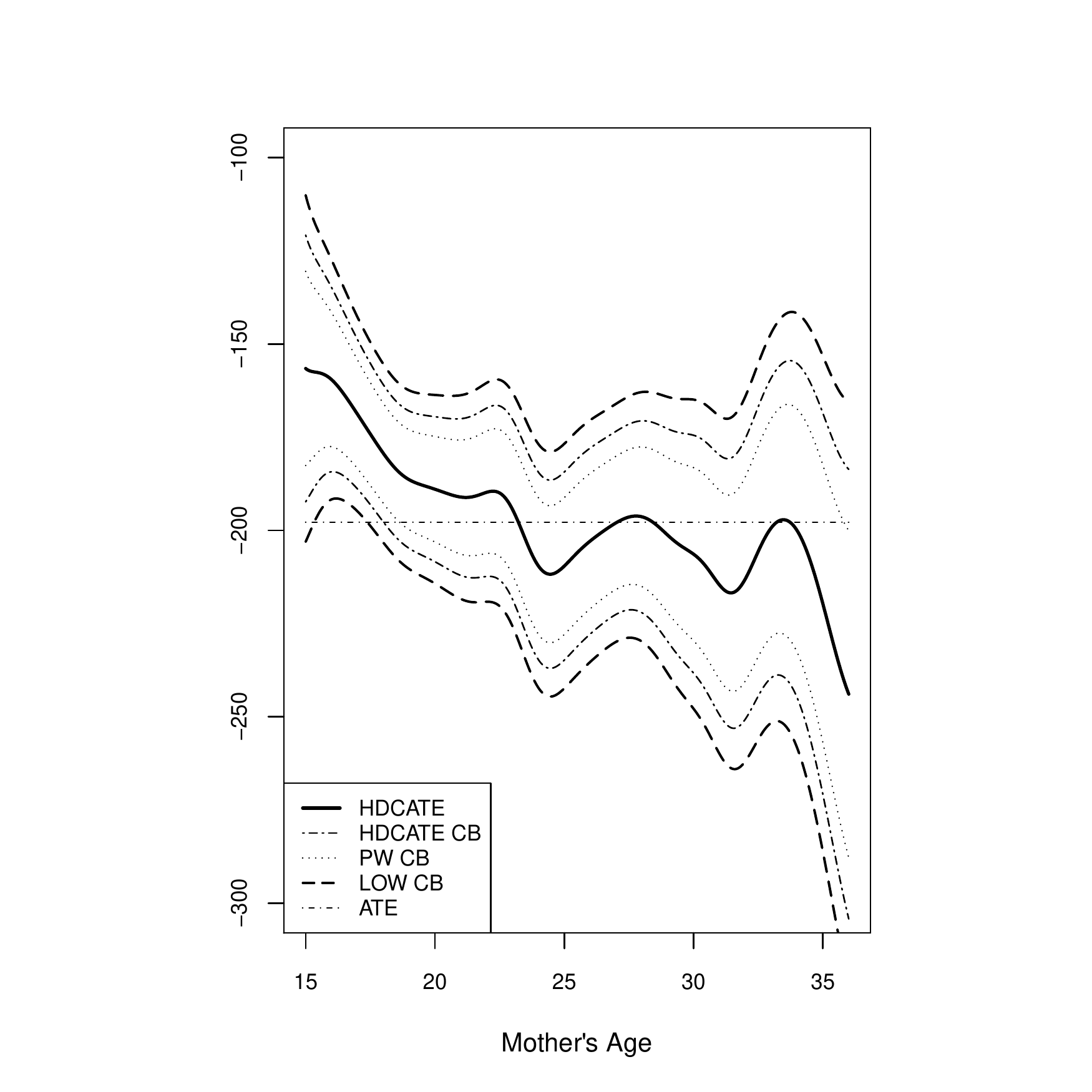}
\caption{CATE for the effect of smoking\\ on birth weights conditional on mother's\\ age, 99\% confidence bands.}
\label{fig1}
\end{minipage}%
\begin{minipage}[c]{0.5\linewidth}
\centering

\includegraphics[width=1\textwidth,height=0.3\textheight]{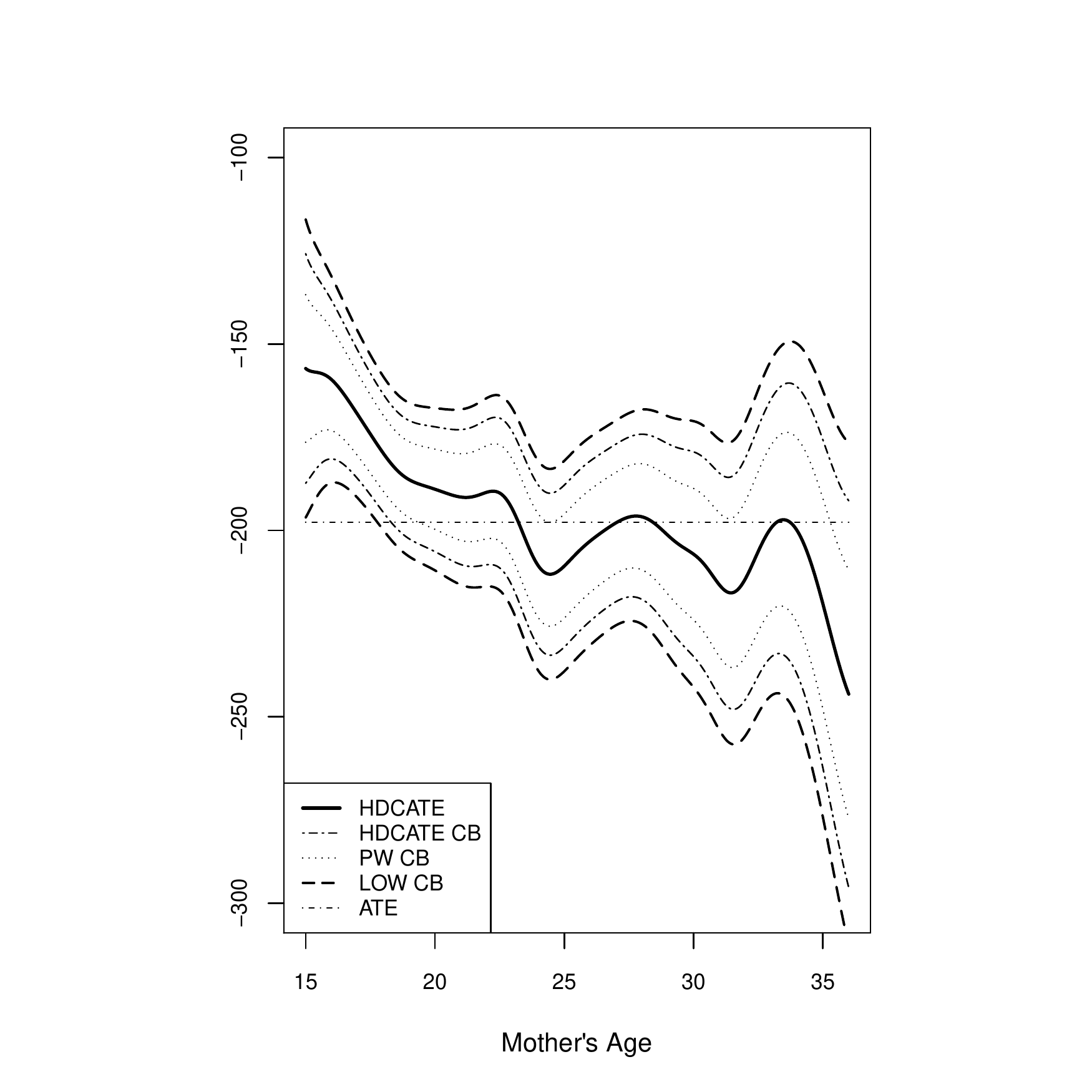}
\caption{CATE for the effect of smoking  on birth weights conditional on mother's age, 95\% confidence bands.}
\label{fig2}
\end{minipage}
\begin{minipage}[c]{1\linewidth}
\centering

\includegraphics[width=0.5\textwidth,height=0.3\textheight]{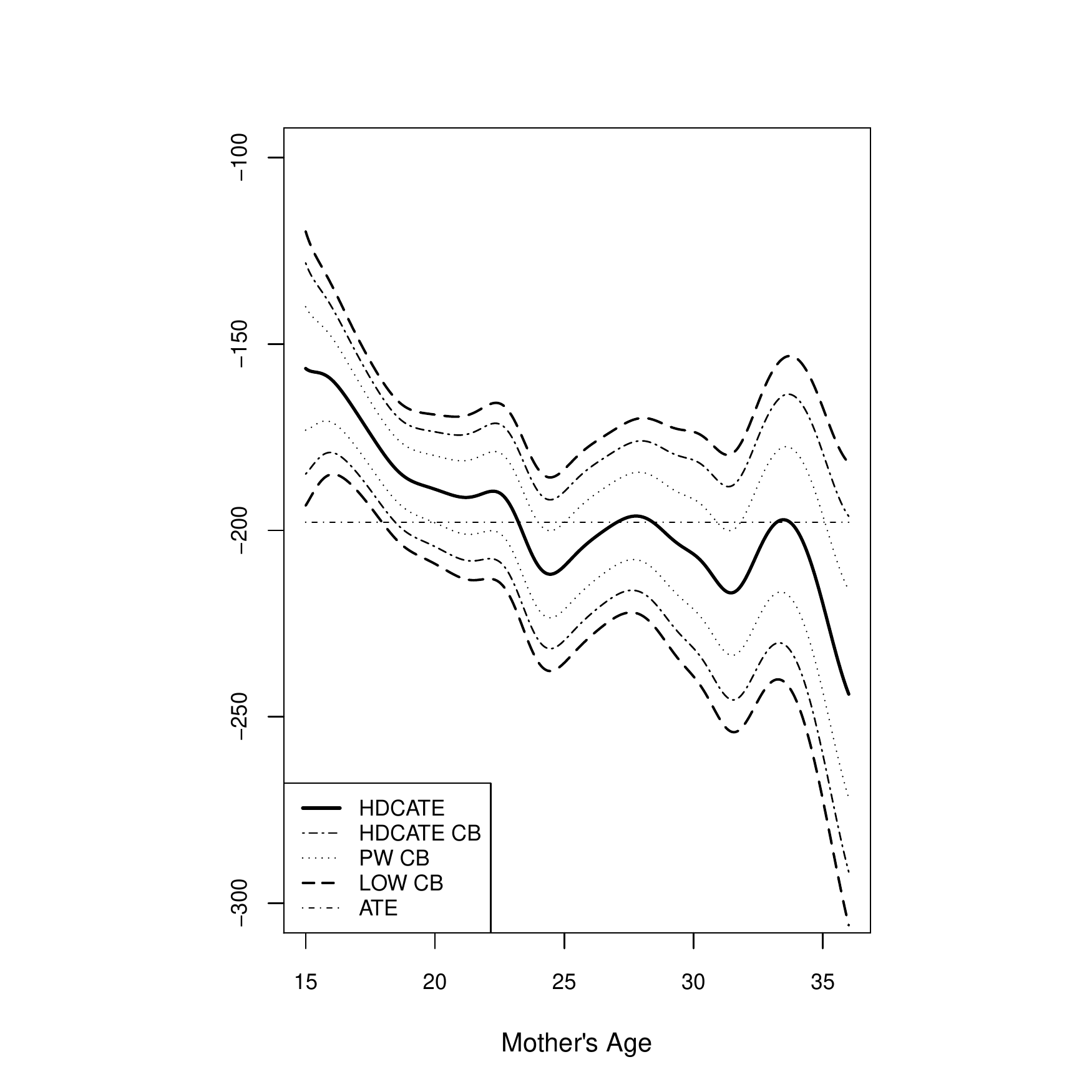}
\caption{CATE for the effect of smoking on birth weights conditional on mother's age, 90\% confidence bands.}
\label{fig3}
\end{minipage}%
\end{figure}

The HDCATE estimates are displayed in Figures \ref{fig1}, \ref{fig2} and \ref{fig3}, along with 90\%, 95\% and 99\% confidence bands, respectively. For a given confidence level, we compute two types of intervals. ``HDCATE CB'' is the proposed uniform confidence band computed according to the algorithm given in Section 4. ``PW CB'' is a pointwise confidence band, given for purposes of comparison, where the critical value $\widehat{C}^{\text{2-sided}}_\alpha$ is replaced by the corresponding value from the standard normal distribution (e.g., 1.96 for $\alpha=5\%$), and ``LOW CB'' is the uniform confidence band by that of LOW. The constant function labeled ``ATE'' represents the estimated average treatment effect across all ages.

Figures \ref{fig1}--\ref{fig3} show that maternal smoking has a negative effect on birth weight at all ages (the upper bounds of the confidence bands are negative), and the average effect is likely to become \emph{more negative} with age. For example, the point estimates show that for teenage mothers of age 18 or younger the negative effect of smoking is, on average, less than 180 grams in absolute value. For mothers around age 24, the same effect is $-220$ grams, and it approaches $-250$ above 35 years of age.\footnote{The non-monotonicities in the point estimate between ages 25 and 35 could be due to undersmoothing and the quickly declining number of first-time mothers toward the top of this age range.} Thus, there is substantial variation in the estimated average treatment effect by age. A potential explanation is that older mothers are likely to have smoked for a longer period, and the detrimental effects of smoking are cumulative (the smoking dummy does not control for duration or intensity of smoking). The figures also shed some light on the gains from using the proposed method compared with the original study of  AHL.  The CATE function changes the  shape and location of the treatment effect estimates, especially for the younger mothers, and shows that the estimated treatment effects are always significantly negative using our method.\footnote{ Granted, the choice of bandwidth is different for the two studies, which affects the shape of the estimated heterogeneous treatment curve to some extent, but it is not the key reason for the different results. If we were to use the same bandwidth choice as in  AHL, we would still observe the difference, as we mentioned in the main text.} Another important difference is that in this study we provide a valid uniform confidence band.

Examining the confidence bands qualifies the analysis of the point estimate in important ways. In Figure \ref{fig1}, the lower bound of the 99\% uniform confidence band (dashed line) attains its maximum at around 16 years of age, and the value of this maximum lies just below the minimum of the upper bound attained at around age 24. Thus, it is possible to fit a constant function (at about $-185$ grams) inside the uniform confidence bands. Nevertheless, if one is less conservative and uses the 95\% or 90\% uniform confidence bands displayed in Figures \ref{fig2} and \ref{fig3}, respectively, then it is no longer possible to do so. Thus, there is fairly compelling (statistically significant) evidence that the smoking effect becomes more negative at least between the ages of 16 and 24. Based on the pointwise confidence band, there is some evidence of further decline in HDCATE at higher ages but it is possible to fit constant functions even within the 90\% uniform confidence bands over the interval $[25,35]$. (Again note that these bands become rather wide at higher ages due to the relatively small number of observations.) The LOW confidence band is visibly wider than our confidence band, which is consistent with our simulation results.

\section{Conclusion}

We advance the literature on the estimation of the reduced dimensional CATE function by proposing that the nuisance functions necessary for identification be estimated by flexible machine learning methods, followed by a traditional local linear regression. 
The asymptotic theory we develop builds on previous work by \cite{BCFH13} and \cite{CC18}. Nevertheless, the theory requires non-trivial modifications to accommodate local linear regression in the second stage. Moreover, CATE is a functional parameter, and our results can be used to conduct uniform inference through a bootstrap procedure. 
In line with \cite{CC18}, we also advocate using the cross-fitting approach to estimate the nuisance functions and conduct the second-stage regression. 

Using the proposed methods, we revisited the problem of estimating the average effect of smoking during pregnancy on birth weight as a function of the mother's age. Our results fall in between AHL and LOW in the sense that we do find age-related heterogeneity (unlike LOW), but it is less marked than in the former study. In particular, there is evidence that the negative effect of smoking becomes somewhat more pronounced with age. 

\newpage
\appendix
\begin{center}
	\Large{Supplement to ``Estimation of Conditional Average Treatment Effects with
		High-Dimensional Data"}
\end{center}

\begin{abstract}
	This paper collects the supplementary materials to the original paper. Section \ref{sec:emp} contains the estimation results for black mothers and robustness check results of the main text. In Section \ref{sec:notation}, we introduce some notation. In Section \ref{sec:lc}, we describe the local constant second-stage estimation. In Sections \ref{sec:full} and \ref{sec:split}, we prove Theorems \ref{thm:cate} and \ref{thm:cate bootstrap} for the full- and split-sample estimators, respectively. In Section \ref{sec:var}, we prove Theorem \ref{thm: var consistency}. In Section \ref{sec:CS}, we prove  Theorem \ref{thm:boot-CS}. Section \ref{sec:lem} collects the proofs of technical lemmas. Simulation results are put in Section \ref{sec:sim}.
	
	\medskip
	\noindent \textbf{Keywords:} heterogeneous treatment effects, high-dimensional data, uniform confidence band
	
	\noindent \textbf{JEL codes:} C14, C21, C55
\end{abstract}

\setcounter{page}{1} \renewcommand\thesection{\Alph{section}} %

\renewcommand{\thefootnote}{\arabic{footnote}} \setcounter{footnote}{0}

\setcounter{equation}{0}

\section{Extra Empirical Results}\label{sec:emp}
\subsection{Estimation Results for Black Mothers}
\label{sec:black}
We also conduct our study on the subsample of black mothers. The sample size is 157,956 with the same covariates as discussed in the main text. The estimation results are summarized in the following figures. We observe similar pattern of increasingly negative treatment effect of smoking on birth weight. But different from the white mothers' sample, the shape of the CATE for black mothers is not as steep. For the younger aged mothers, the negative treatment effect for black mothers is around 100 grams, and for the older aged cohort, the estimated treatment is around negative 200 grams, whereas the two tails are around 180 grams and 250 grams for the white mothers. The confidence band is also wider using the black mother's sample. The socioeconomic factors associated with behavioral features such as smoking and other activities should be explored further in future studies. 

\begin{figure}[H]
	\begin{minipage}[c]{0.5\linewidth}
		\centering
		\includegraphics[width=1\textwidth,height=0.3\textheight]{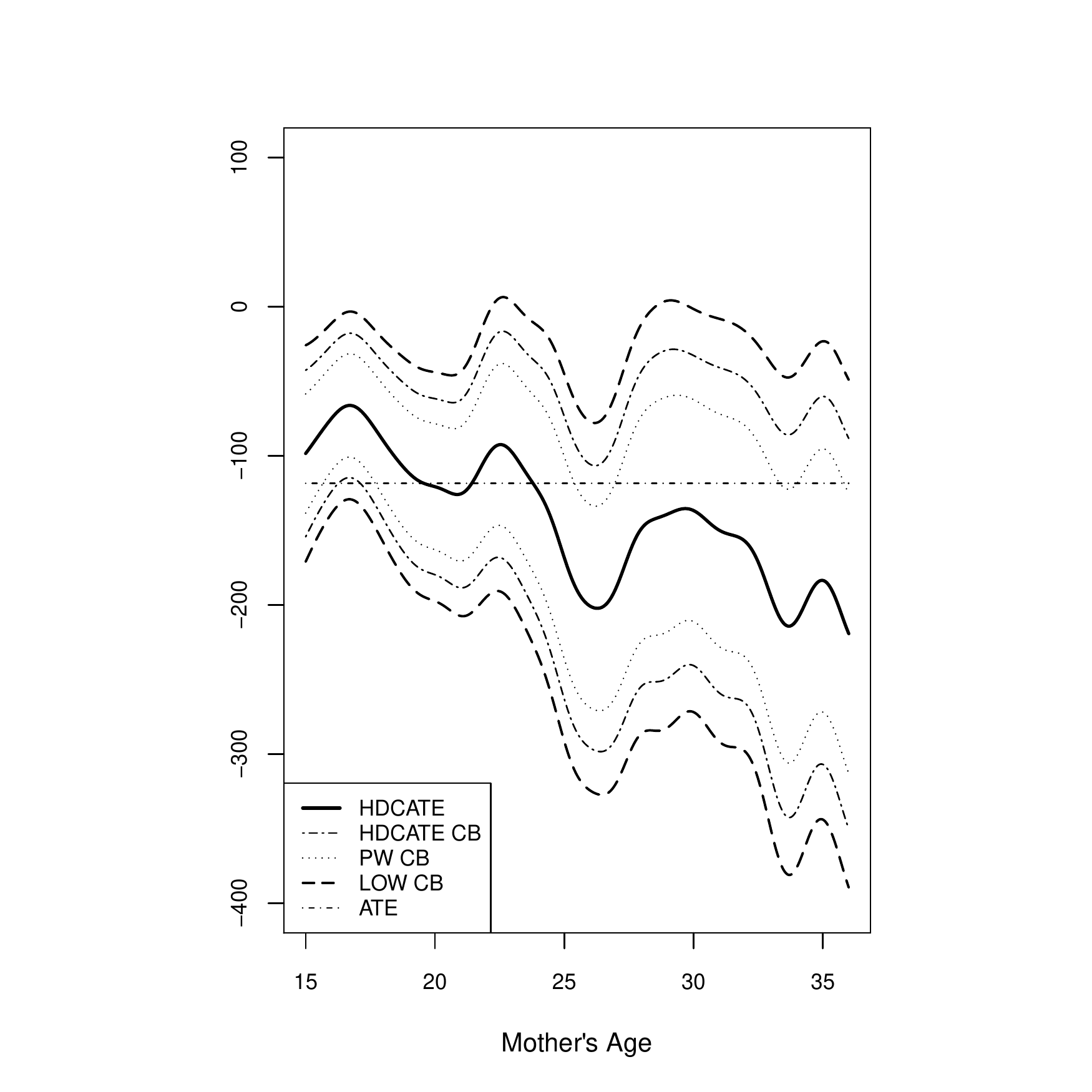}
		\caption{CATE for the effect of smoking\\ on birth weights conditional on mother's\\ age, 99\% confidence bands.}
		\label{fig1}
	\end{minipage}%
	\begin{minipage}[c]{0.5\linewidth}
		\centering
		\includegraphics[width=1\textwidth,height=0.3\textheight]{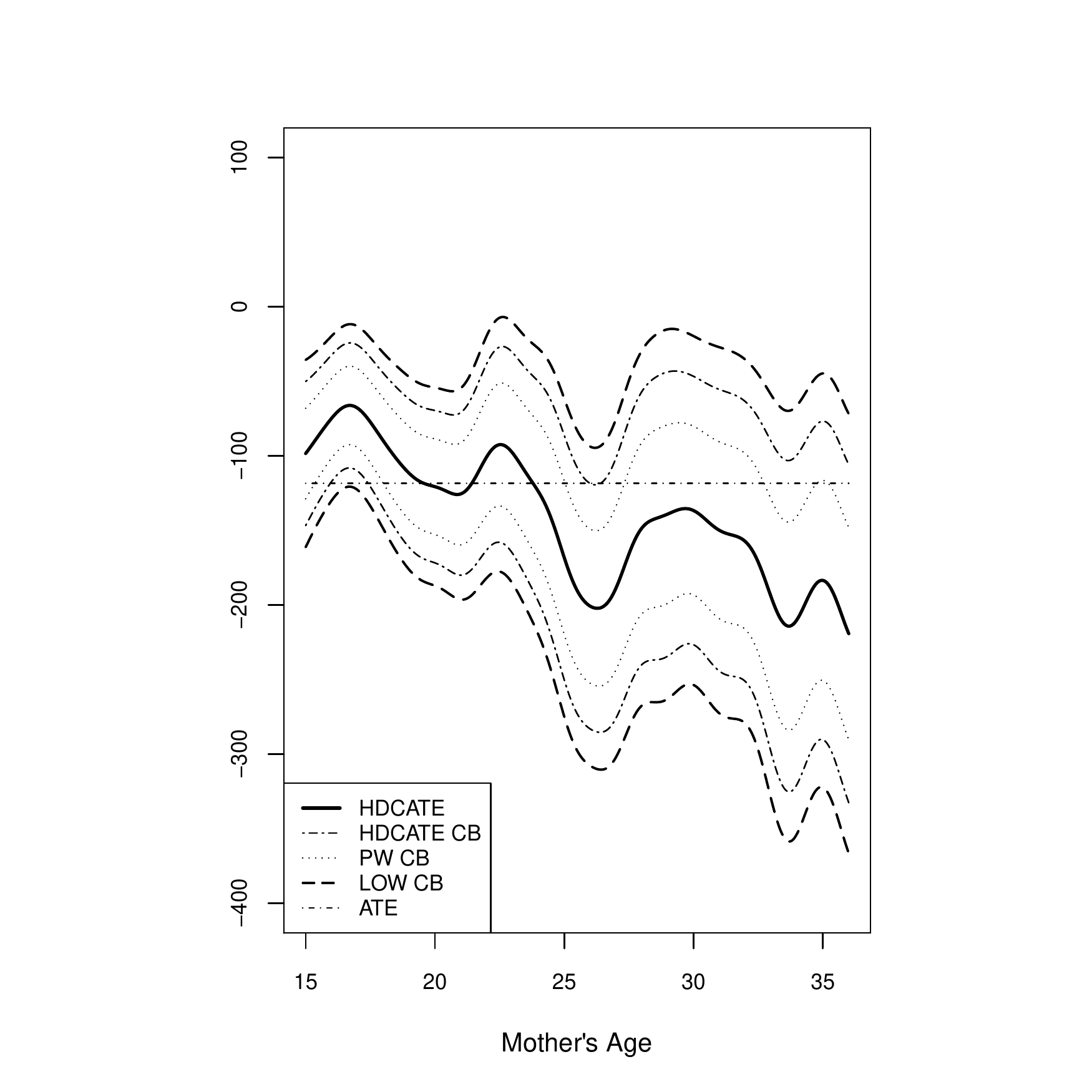}
		\caption{CATE for the effect of smoking  on birth weights conditional on mother's age, 95\% confidence bands.}
		\label{fig2}
	\end{minipage}
	\begin{minipage}[c]{1\linewidth}
		\centering
		\includegraphics[width=0.5\textwidth,height=0.3\textheight]{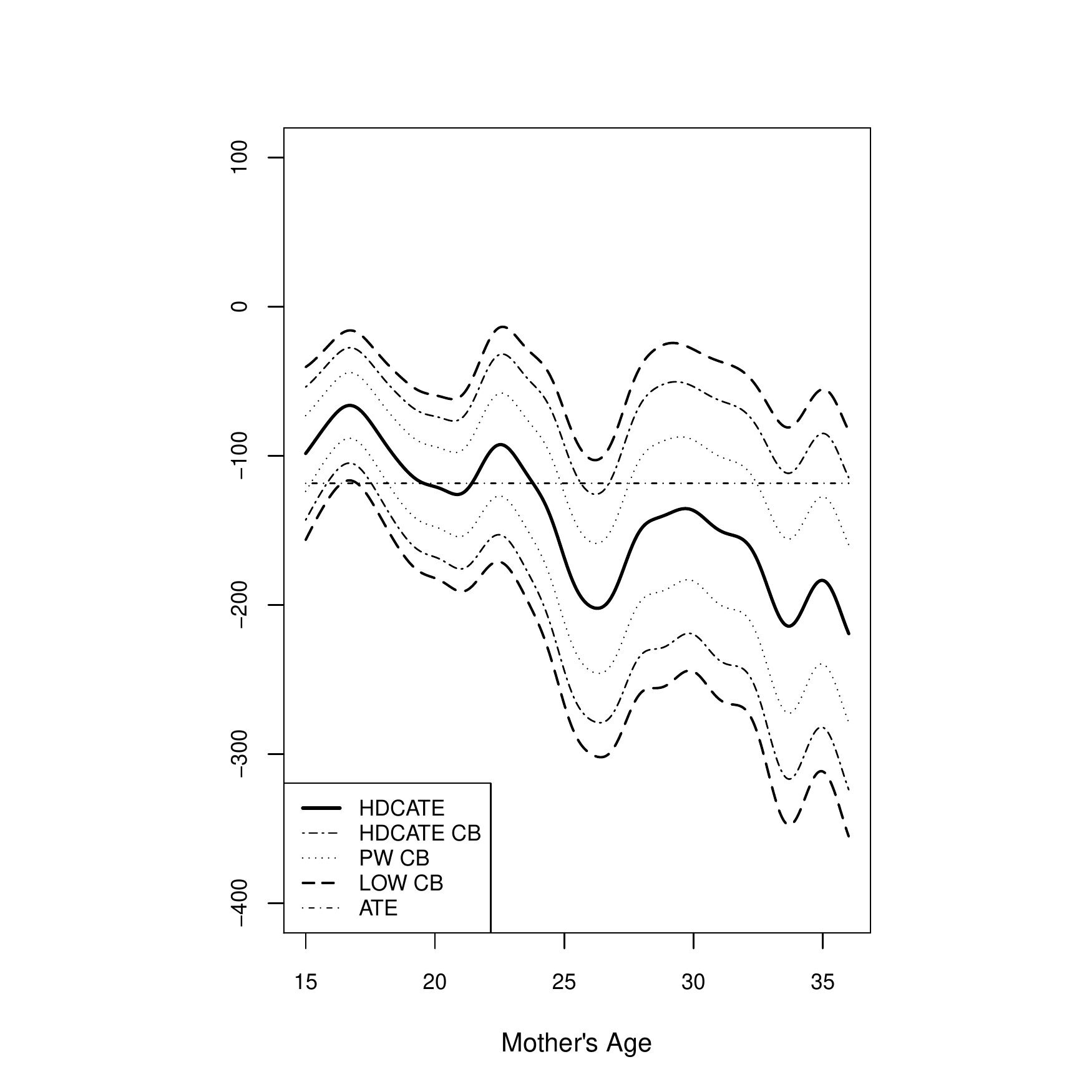}
		\caption{CATE for the effect of smoking on birth weights conditional on mother's age, 90\% confidence bands.}
		\label{fig3}
	\end{minipage}%
\end{figure}
\subsection{Robustness Check under Different ``Continuitization" Settings}
\label{sec:cont}
In the main text we add a uniform $[-1,1]$ random variable to the age variable. Here we show that the empirical results are robust under different ``continuitization" settings. Specifically, we add a uniform $[-0.5,0.5]$ error to age variable as in AHL. From the following figures, we could see the heterogeneous (with respect to age) treatment effect of smoking on birth weight is similar to what we found in the main text.

\begin{figure}[H]
	\begin{minipage}[c]{0.5\linewidth}
		\centering
		\includegraphics[width=1\textwidth,height=0.3\textheight]{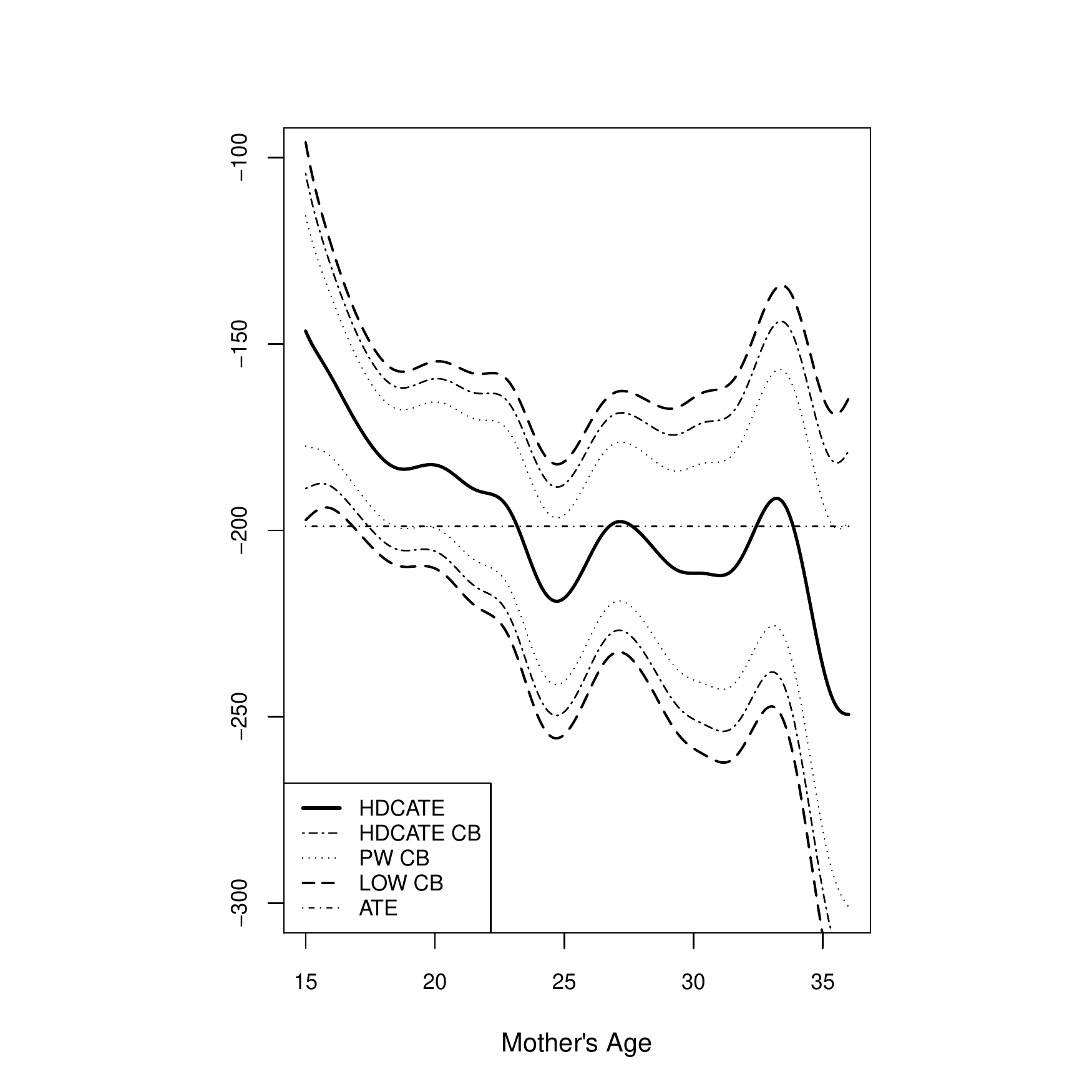}
		\caption{CATE for the effect of smoking\\ on birth weights conditional on mother's\\ age, 99\% confidence bands.}
		\label{fig1}
	\end{minipage}%
	\begin{minipage}[c]{0.5\linewidth}
		\centering
		\includegraphics[width=1\textwidth,height=0.3\textheight]{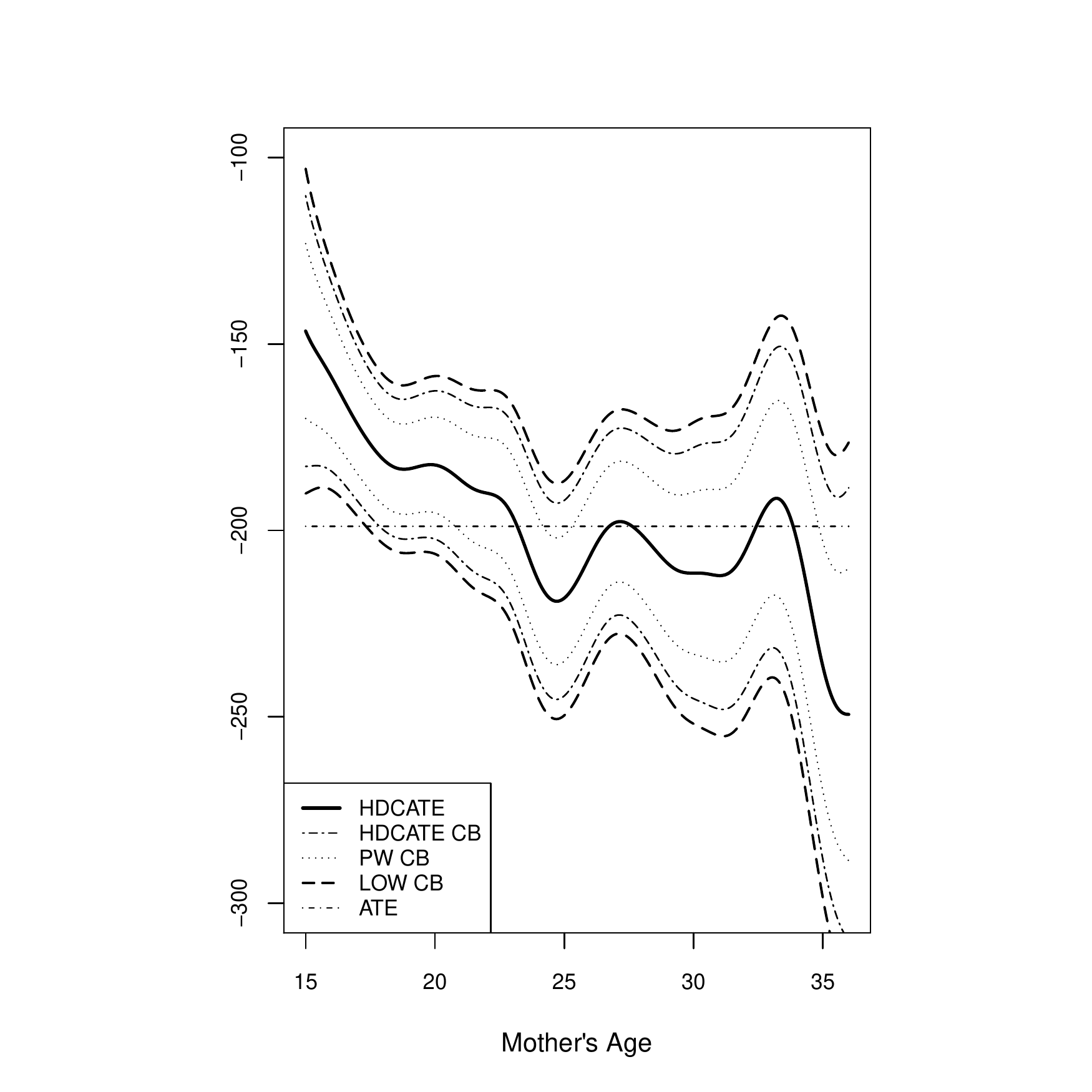}
		\caption{CATE for the effect of smoking  on birth weights conditional on mother's age, 95\% confidence bands.}
		\label{fig2}
	\end{minipage}
	\begin{minipage}[c]{1\linewidth}
		\centering
		\includegraphics[width=0.5\textwidth,height=0.3\textheight]{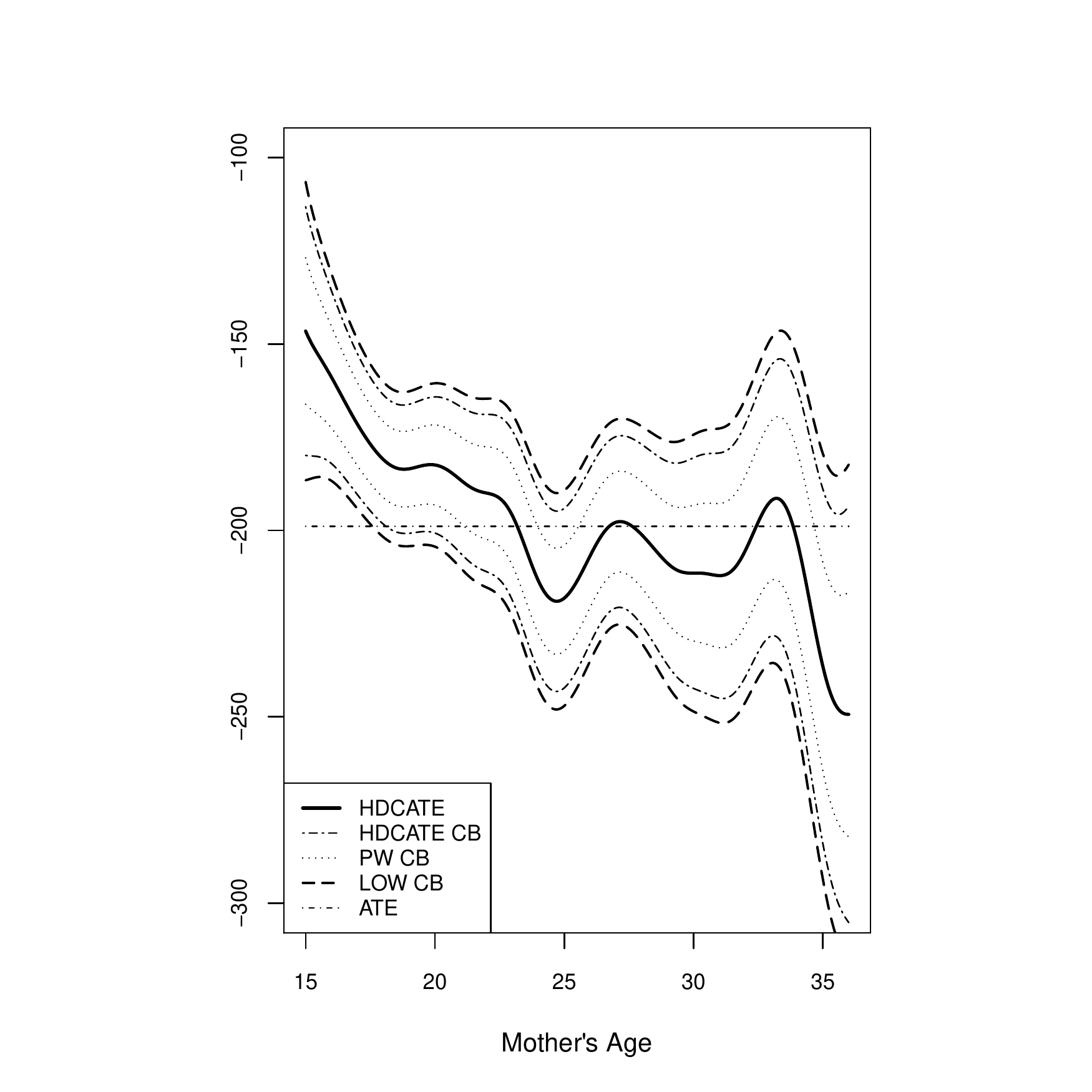}
		\caption{CATE for the effect of smoking on birth weights conditional on mother's age, 90\% confidence bands.}
		\label{fig3}
	\end{minipage}%
\end{figure}

\section{Notation}
\label{sec:notation}
In the following, we define notation that will be used later. 
\begin{itemize}
	\item Full-sample estimator:
	\begin{itemize}
		\item Let $\xi^*$ be either $1$ or $\xi$,  a random variable that satisfies Assumption \ref{assn: boot_xi}. The following proof is valid for the original and bootstrap estimators with $\xi^* = 1$ and $\xi$, respectively.
		\item Let $\psi(1,W,\eta) = \frac{D(Y-\mu(1,X))}{\pi(X)} + \mu(1,X)$ and $\psi(0,W,\eta) = \frac{(1-D)(Y-\mu(0,X))}{1-\pi(X)} + \mu(0,X)$ such that $\psi(W,\eta) = \psi(1,W,\eta)-\psi(0,W,\eta)$.
		\item Let $\tau_0(j,x_1) = \E(\mu_0(j,X)|X_1=x_1)$, $j=0,1$.
		\item Let $(\hat{\tau}^*(x_1),\hat{\beta}^*(x_1))$ be obtained as 
		\begin{align*}
		(\hat{\tau}^*(x_1),\hat{\beta}^*(x_1)) = \arg\min_{a,b} \sum_{i \in I} \xi_i^*\left[\psi(W_i,\hat\eta(I)) - a - (X_{1i} - x_1)'b \right]^2\mcalK_h\left(X_{1i}-x_1\right). 
		\end{align*}
		\item When $\xi_i^* = 1$, $(\hat{\tau}^*(x_1),\hat{\beta}^*(x_1)) = (\hat{\tau}(x_1),\hat{\beta}(x_1))$; when $\xi_i^* = \xi_i$, $(\hat{\tau}^*(x_1),\hat{\beta}^*(x_1)) = (\hat{\tau}^b(x_1),\hat{\beta}^b(x_1))$
	\end{itemize}
	\item Split-sample estimator:
	\begin{itemize}
		\item Let $\P_{n,k}f = \frac{1}{n}\sum_{i \in I_k}f(W_i)$  for a generic function $f(\cdot)$.
		\item Let $\P_{I_k}f = \E(f(W_1,\cdots,W_N)|W_i,i\in I_k^c)$ for a generic function $f$.
		\item Let $(\hat{\tau}^*_k(x_1),\hat{\beta}^*_k(x_1))$ be obtained as  
		\begin{align*}
		(\hat{\tau}_k^*(x_1),\hat{\beta}_k^*(x_1)) = \arg\min_{a,b} \sum_{i \in I_k} \xi_i^*\left[\psi(W_i,\hat\eta(I_k^c)) - a - (X_{1i} - x_1)'b \right]^2\mcalK_h\left(X_{1i}-x_1\right). 
		\end{align*}
		\item When $\xi_i^* = 1$, $(\hat{\tau}^*_k(x_1),\hat{\beta}^*_k(x_1)) = (\hat{\tau}_k(x_1),\hat{\beta}_k(x_1))$; when $\xi_i^* = \xi_i$, $(\hat{\tau}^*_k(x_1),\hat{\beta}^*_k(x_1)) = (\hat{\tau}^b_k(x_1),\hat{\beta}^b_k(x_1))$
		\item Denote $\check{\tau}^*(x_1) = \frac{1}{K}\sum_{k=1}^K\hat{\tau}_k^*(x_1).$ 
	\end{itemize}
\end{itemize}

\section{Local Constant Second-Stage Estimation}
\label{sec:lc}
In this section, we describe the local constant estimation in our second stage. 
\paragraph{The full-sample estimator} Let $\hat\mu(0,x; I)$, $\hat\mu(1,x; I)$ and $\hat\pi(x)$ be the first-stage estimator computed over the full sample $I$.  Furthermore, let
\[
\hat f(x_1; I)=\frac{1}{Nh^d}\sum_{i\in I}\mcalK_h\left(X_{1i}-x_1\right)
\]
denote a kernel density estimator of the p.d.f.\ of $X_1$ over $I$, where $\mcalK$ is a $d$-dimensional product kernel, $h$ is a smoothing parameter (bandwidth), and $\mcalK_h\left(u\right)  = \mcalK\left(\frac{u}{h}\right)$. Set $\hat\eta(I)=(\hat\mu(0,\cdot; I), \hat\mu(1,\cdot; I), \hat \pi(\cdot; I))$. The second stage of the full-sample estimator is
\begin{equation}\label{full sample 2nd stage k}
\hat\tau(x_1)=\frac{1}{Nh^d \hat f(x_1; I)} \sum_{i\in I}\psi(W_i,\hat\eta(I))\mcalK_h\left(X_{1i}-x_1\right).
\end{equation}

\paragraph{The $K$-fold cross-fitting estimator}  For each $k=1,\ldots, K$, let $\hat\mu(0,x; I_k^c)$, $\hat\mu(1,x; I_k^c)$ and $\hat\pi(x; I_k^c)$ be the first-stage estimators computed over the subsample $I_k^c=I\setminus I_k$. Furthermore, let
\[
\hat f(x_1; I_k)=\frac{1}{nh^d}\sum_{i\in I_k}\mcalK_h\left(X_{1i}-x_1\right)
\]
denote a kernel density estimator of the p.d.f.\ of $X_1$ over the subsample $I_k$. Set $\hat\eta(I_k^c)=(\hat\mu(0,\cdot; I_k^c), \hat\mu(1,\cdot; I_k^c), \hat \pi(\cdot; I_k^c) )$. The second stage of the $K$-fold cross-fitting estimator over the sample $I_k$ is
\begin{equation}\label{split sample 2nd stage k}
\tilde\tau_{k}(x_1)=\frac{1}{n h^d \hat f(x; I_k)} \sum_{i\in I_k}\psi(W_i,\hat\eta(I_k^c))\mcalK_h\left(X_{1i}-x_1\right).
\end{equation}
Finally, in the third stage we take the average of the $K$ preliminary estimates to obtain an efficient estimator:
\[
\check\tau(x_1)=\frac{1}{K}\sum_{k=1}^K \tilde\tau_k(x_1).
\]
Under the conditions for $\delta_{1N}$, $\delta_{2N}$, and $\delta_{4N}$ in Assumptions \ref{assn: 1st stage full} and  \ref{assn: 1st stage split} and all the other regularity conditions in the paper, we can derive the first-order linear expansions for $\hat{\tau}(x_1)$ and $\check\tau(x_1)$, which are the same as those in Theorem \ref{thm:cate}. We do not need the assume the condition for $\delta_{3N}$ as we cannot estimate the derivative of $\tau_0(x_1)$ for this local constant method.

For the multiplier bootstrap estimator can be computed by taking the weighted average in \eqref{full sample 2nd stage k} and \eqref{split sample 2nd stage k} with bootstrap weights $\xi_i$, i.e.,  
\begin{equation*}
\hat\tau^b(x_1)=\frac{1}{Nh^d \hat f^b(x_1; I)} \sum_{i\in I}\xi_i\psi(W_i,\hat\eta(I))\mcalK_h\left(X_{1i}-x_1\right).
\end{equation*}
and 
\begin{equation*}
\tilde\tau^b_{k}(x_1)=\frac{1}{n h^d \hat f^b(x; I_k)} \sum_{i\in I_k}\xi_i\psi(W_i,\hat\eta(I_k^c))\mcalK_h\left(X_{1i}-x_1\right),
\end{equation*}
where 
\[
\hat f^b(x_1; I)=\frac{1}{Nh^d}\sum_{i\in I}\xi_i\mcalK_h\left(X_{1i}-x_1\right)
\]
and 
\[
\hat f^b(x_1; I_k)=\frac{1}{nh^d}\sum_{i\in I_k}\xi_i\mcalK_h\left(X_{1i}-x_1\right).
\]
We can derive the first-order linear expansions for $\hat{\tau}^b(x_1)$ and $\check\tau^b(x_1)$, which are the same as those in Theorem \ref{thm:cate bootstrap}. Then, we can construct the uniform confidence band as in Section \ref{sec: unif inference}. Theorem \ref{thm:boot-CS} remains unchanged as the first-order linear expansions are the same. For the detailed proofs of all statements in this section, we refer readers to the previous version of this paper on arXiv.

\section{The proof of Theorems \ref{thm:cate} and \ref{thm:cate bootstrap} for the full-sample estimator}
\label{sec:full}
Let $D_N = \text{diag}((Nh^d)^{-1},(Nh^{d+2})^{-1}\iota_d)$ and 
$$F^{-1}(x_1) = \begin{pmatrix}
1/f(x_1) & 0 \\
- f^{(1)}/f^2(x_1) & I_d/(\nu f(x_1))
\end{pmatrix},$$ 
where $\iota_d$ is a $d \times 1$ vector of one. We aim to show 
\begin{align}
\label{eq:tb1}
\begin{pmatrix}
\hat{\tau}^*(x_1) - \tau_0(x_1) \\
\hat{\beta}^*(x_1) - \beta_0(x_1)
\end{pmatrix} = F^{-1}(x_1)\left(D_N \sum_{i \in I}\xi_i^*\begin{pmatrix}
1 \\
X_{1i} - x_1
\end{pmatrix} (\psi(W_i,\eta) - \tau_0(X_1)) \mathcal{K}_h(X_{1i}-x_1)\right) + \begin{pmatrix}
R_{\tau}(x_1) \\
R_{\beta}(x_1) \\
\end{pmatrix},
\end{align}
such that $\sup_{ x_1 \in \overline{\mathcal{X}}_1 }|R_{\tau}(x_1)| = o_p( (\log(N)Nh^d)^{-1/2})$ and $\sup_{ x_1 \in \overline{\mathcal{X}}_1 }||R_{\beta}(x_1)||_2 = o_p( (\log(N)Nh^{d+2})^{-1/2})$. When $\xi_i^*=1$, \eqref{eq:tb1} leads to the first result in Theorem \ref{thm:cate}. Taking the difference of \eqref{eq:tb1} for $\xi_i^* = \xi_i$ and $\xi_i^* =1$, we obtain the first result in Theorem \ref{thm:cate bootstrap}. 

First note 
\begin{align*}
\begin{pmatrix}
\hat{\tau}^*(x_1) \\
\hat{\beta}^*(x_1)
\end{pmatrix} = & \left[\sum_{i \in I} \xi_i^*D_N\begin{pmatrix}
1 & (X_{1i} - x_1)^T\\
X_{1i} - x_1 & (X_{1i} - x_1)(X_{1i} - x_1)^T
\end{pmatrix}\mathcal{K}_h(X_{1i}-x_1) \right]^{-1}\\
& \times \left[D_N \sum_{i \in I}\xi_i^*\begin{pmatrix}
1 \\
X_{1i} - x_1
\end{pmatrix} \psi(W_i,\hat{\eta}(I)) \mathcal{K}_h(X_{1i}-x_1)  \right] \\
\equiv &  (F^*_n(x_1))^{-1}G_n^*(x_1).
\end{align*}	
It can be shown by standard kernel estimation results (see, for example, \citet[Chapter 2]{lr07}) that 
\begin{align*}
F^*_n(x_1) = \begin{pmatrix}
f(x_1) + R_{1,F}(x_1) &  h^2(\nu f^{(1)}(x_1) + R_{2,F}(x_1))^T\\
\nu f^{(1)}(x_1) + R_{2,F}(x_1) & \nu (f(x_1) + R_{3,F}(x_1)) I_d
\end{pmatrix},
\end{align*}
where 
$$\sup_{x_1 \in \overline{\mathcal{X}}_1 }|R_{1,F}(x_1)| = O_p(\log^{1/2}(N)(Nh^d)^{-1/2}),$$ $$\sup_{x_1 \in \overline{\mathcal{X}}_1 }|R_{2,F}(x_1)| = O_p(\log^{1/2}(N)(Nh^{d+2})^{-1/2}),$$ $$\sup_{x_1 \in \overline{\mathcal{X}}_1 }|R_{3,F}(x_1)| = O_p(\log^{1/2}(N)(Nh^d)^{-1/2}),$$ $\nu = \int u^2 k(u)du$, $f^{(1)}(x_1)$ is the derivative of $f(x_1)$ w.r.t. $x_1$, and $I_d$ is the $d\times d$ identity matrix. Then, 

\begin{align}
\label{eq:Finver}
(F^*_n(x_1))^{-1} = F^{-1}(x_1) + R_{F}(x_1),
\end{align}
where 
\begin{align*}
\sup_{ x_1 \in \overline{\mathcal{X}}_1 }|R_{F}(x_1)| = & \sup_{ x_1 \in \overline{\mathcal{X}}_1 }(F_n^*(x_1))^{-1}(F(x_1) - F_n^*(x_1))(F(x_1))^{-1} \\
= & \begin{pmatrix}
O_p(\log^{1/2}(N)(Nh^d)^{-1/2})&  O_p(h^2)\\
O_p(\log^{1/2}(N)(Nh^{d+2})^{-1/2}) & O_p(\log^{1/2}(N)(Nh^d)^{-1/2}) I_d
\end{pmatrix}.
\end{align*}

For $G_n^*(x_1)$, we have 
\begin{align*}
G_n^*(x_1) = D_N \sum_{i \in I}\xi_i^*\begin{pmatrix}
1 \\
X_{1i} - x_1
\end{pmatrix} \psi(W_i,\eta) \mathcal{K}_h(X_{1i}-x_1) + R_n^*(x_1),
\end{align*}
where $R_n^*(x_1) = (R_{n,1}^*(x_1),(R_{n,2}^*(x_1))^T)^T$, 
\begin{align*}
R_{n,1}^*(x_1) = \sum_{i \in I}\frac{\xi_i^*}{Nh^d}\left(\psi(W_i,\hat{\eta}(I))-\psi(W_i,\eta)\right) \mathcal{K}_h(X_{1i}-x_1), 
\end{align*}
and 
\begin{align*}
R_{n,2}^*(x_1) = \sum_{i \in I}\frac{\xi_i^*}{Nh^{d+2}}\left(X_{1i} - x_1\right) \left(\psi(W_i,\hat{\eta}(I))-\psi(W_i,\eta)\right) \mathcal{K}_h(X_{1i}-x_1). 
\end{align*}

Then, 
\begin{align}
\label{eq:A0}
&\begin{pmatrix}
\hat{\tau}^*(x_1) - \tau_0(x_1) \\
\hat{\beta}^*(x_1) - \beta_0(x_1)
\end{pmatrix} \notag \\
= & (F^*_n(x_1))^{-1}G_n^*(x_1) - \begin{pmatrix}
\tau_0(x_1) \\
\beta_0(x_1)
\end{pmatrix} \notag \\
= & (F^*_n(x_1))^{-1}\left(D_N \sum_{i \in I}\xi_i^*\begin{pmatrix}
1 \\
X_{1i} - x_1
\end{pmatrix} \psi(W_i,\eta) \mathcal{K}_h(X_{1i}-x_1) + R_n^*(x_1)\right) - \begin{pmatrix}
\tau_0(x_1) \\
\beta_0(x_1)
\end{pmatrix} \notag \\
= & \left[(F^*_n(x_1))^{-1}\left(D_N \sum_{i \in I}\xi_i^*\begin{pmatrix}
1 \\
X_{1i} - x_1
\end{pmatrix} (\psi(W_i,\eta) - \tau_0(X_1)) \mathcal{K}_h(X_{1i}-x_1)\right) \right] \notag \\
& + \left[(F^*_n(x_1))^{-1}\left(D_N \sum_{i \in I}\xi_i^*\begin{pmatrix}
1 \\
X_{1i} - x_1
\end{pmatrix} (\tau_0(X_1) - \tau_0(x_1) - (X_1-x_1)^T\beta_0(x_1) ) \mathcal{K}_h(X_{1i}-x_1)\right)\right] \notag \\
& +\biggr[(F^*_n(x_1))^{-1}R_n^*(x_1) \biggl] \notag \\
\equiv & A_1(x_1) + A_2(x_2) + A_3(x_1). 
\end{align}

We note that, component-wise, 
\begin{align*}
\E \left(D_N \sum_{i \in I}\xi_i^*\begin{pmatrix}
1 \\
X_{1i} - x_1
\end{pmatrix} (\psi(W_i,\eta) - \tau_0(X_1)) \mathcal{K}_h(X_{1i}-x_1)\right)=0.
\end{align*}
Then, by \citet[Corollary 5.1]{CCK14}, we have 
\begin{align*}
\sup_{ x_1 \in \overline{\mathcal{X}}_1 }\left|\left(D_N \sum_{i \in I}\xi_i^*\begin{pmatrix}
1 \\
X_{1i} - x_1
\end{pmatrix} (\psi(W_i,\eta) - \tau_0(X_1)) \mathcal{K}_h(X_{1i}-x_1)\right)\right| = \begin{pmatrix}
O_p(\sqrt{\log(N)/(Nh^d)}) \\
O_p(\sqrt{\log(N)/(Nh^{d+2})})\iota_d \\
\end{pmatrix}.
\end{align*}
Then, by \eqref{eq:Finver}, 
\begin{align}
\label{eq:A1}
A_1(x_1) = F^{-1}(x_1)\left(D_N \sum_{i \in I}\xi_i^*\begin{pmatrix}
1 \\
X_{1i} - x_1
\end{pmatrix} (\psi(W_i,\eta) - \tau_0(X_1)) \mathcal{K}_h(X_{1i}-x_1)\right) + R_{1,A}(x_1),
\end{align}
where 
\begin{align*}
\sup_{ x_1 \in \overline{\mathcal{X}}_1 }|R_{1,A}(x_1)| = &  \sup_{ x_1 \in \overline{\mathcal{X}}_1 }|R_{F}(x_1)| \times \begin{pmatrix}
O_p(\sqrt{\log(N)/(Nh^d)}) \\
O_p(\sqrt{\log(N)/(Nh^{d+2})})\iota_d \\
\end{pmatrix} = \begin{pmatrix}
o_p(\left(\log(N)Nh^d\right)^{-1/2}) \\
o_p(\left(\log(N)Nh^{d+2}\right)^{-1/2})\iota_d \\
\end{pmatrix}.
\end{align*}

For the first element of $A_2(x_1)$, by the Taylor expansion to the second order, we have 
\begin{align*}
\left|\frac{1}{Nh^d}\sum_{i \in I}\xi_i^*(\tau_0(X_1) - \tau_0(x_1) - (X_1-x_1)^T\beta_0(x_1) ) \mathcal{K}_h(X_{1i}-x_1)\right| \leq \frac{1}{Nh^d}\sum_{i \in I}|\xi_i^*|||X_{1i}-x_1||_2^2\mathcal{K}_h(X_{1i}-x_1).
\end{align*}
By \citet[Corollary 5.1]{CCK14}, we have 
\begin{align}
\label{eq:2}
& \sup_{ x_1 \in \overline{\mathcal{X}}_1 } \frac{1}{Nh^d}\sum_{i \in I}|\xi_i^*|||X_{1i}-x_1||_2^2\mathcal{K}_h(X_{1i}-x_1) \notag \\
\leq & \sup_{ x_1 \in \overline{\mathcal{X}}_1 }(\P_N - \P)h^{-d}|\xi^*|||X_{1}-x_1||_2^2\mathcal{K}_h(X_{1}-x_1) + \sup_{ x_1 \in \overline{\mathcal{X}}_1 }\P h^{-d}|\xi^*|||X_{1}-x_1||_2^2\mathcal{K}_h(X_{1}-x_1) \notag \\
= & O_p(h^2\sqrt{\log(N)/Nh^d}  +h^2) = O_p(h^2) = o_p(\left(\log(N)Nh^d\right)^{-1/2}).
\end{align}
For the rest of the elements in $A_2(x_1)$, we need to conduct Taylor expansion to the third order so that 
\begin{align*}
& \left\Vert \frac{1}{Nh^{d+2}}\sum_{i \in I}\xi_i^*(X_1 - x_1)(\tau_0(X_1) - \tau_0(x_1) - (X_1-x_1)^T\beta_0(x_1) ) \mathcal{K}_h(X_{1i}-x_1)\right\Vert_2 \\
\lesssim & \frac{1}{Nh^{d+2}}\biggl\Vert\sum_{i \in I}\xi_i^*(X_{1i}-x_1)(X_{1i}-x_1)^T\tau^{(2)}(x_1)(X_{1i}-x_1)\mathcal{K}_h(X_{1i}-x_1)\biggr\Vert_2 \\
&+ \frac{1}{Nh^{d+2}}\biggl|\sum_{i \in I}\xi_i^*||X_{1i}-x_1||_2^4\mathcal{K}_h(X_{1i}-x_1)\biggr|.
\end{align*}
For the first term of the above display, by \citet[Corollary 5.1]{CCK14}, 
\begin{align*}
& \sup_{ x_1 \in \overline{\mathcal{X}}_1 }\frac{1}{Nh^{d+2}}\biggl\Vert \sum_{i \in I}\xi_i^*(X_{1i}-x_1)(X_{1i}-x_1)^T\tau^{(2)}(x_1)(X_{1i}-x_1)\mathcal{K}_h(X_{1i}-x_1)\biggr\Vert_2 \\
\leq & \sup_{ x_1 \in \overline{\mathcal{X}}_1 } \biggl\Vert(\P_N-\P)h^{-d-2}\xi^*(X_{1}-x_1)(X_{1}-x_1)^T\tau^{(2)}(x_1)(X_{1}-x_1)\mathcal{K}_h(X_{1}-x_1)\biggr\Vert_2 \\ 
& +  \sup_{ x_1 \in \overline{\mathcal{X}}_1 } \biggl\Vert \P h^{-d-2}\xi^*(X_{1}-x_1)(X_{1}-x_1)^T\tau^{(2)}(x_1)(X_{1}-x_1)\mathcal{K}_h(X_{1}-x_1)\biggr\Vert_2
\\
\leq & O_p(\sqrt{\log(N)h^2/(Nh^d)}) + O(h^2) = O_p(h^2) = o_p(\left(\log(N)Nh^{d+2}\right)^{-1/2}).  
\end{align*}
For the second term of the above display, by the same argument as \eqref{eq:2}, we have 
\begin{align*}
\frac{1}{Nh^{d+2}} \sum_{i \in I}|\xi_i^*| ||X_{1i}-x_1||_2^4\mathcal{K}_h(X_{1i}-x_1) = O_p(h^2) = o_p(\left(\log(N)Nh^{d+2}\right)^{-1/2}).
\end{align*}

Last, by \eqref{eq:Finver}, we have
\begin{align}
\label{eq:A2}
\sup_{ x_1 \in \overline{\mathcal{X}}_1 }||A_2(x_1)||_2 = O_p(h^2). 
\end{align}

By Lemma \ref{lem:1}, we have 
\begin{align}
\label{eq:R1}
\sup_{x_1 \in \overline{\mathcal{X}}_1}|R_{n,1}^*(x_1)| = o_p((\log(N)Nh^d)^{-1/2})
\end{align}
and 
\begin{align}
\label{eq:R2}
\sup_{x_1 \in \overline{\mathcal{X}}_1}|R_{n,2}^*(x_1)| = o_p((\log(N)Nh^{d+2})^{-1/2}).
\end{align}
Combining \eqref{eq:Finver}, \eqref{eq:R1}, and \eqref{eq:R2}, we have, component-wise 
\begin{align}
\label{eq:A3}
\sup_{ x_1 \in \overline{\mathcal{X}}_1 }|A_3(x_1)| = \begin{pmatrix}
O_p((\log(N)Nh^d)^{-1/2}) \\
O_p((\log(N)Nh^{d+2})^{-1/2})\iota_d
\end{pmatrix},
\end{align}

Combining \eqref{eq:A0}, \eqref{eq:A1}, \eqref{eq:A2}, and \eqref{eq:A3}, we obtain \eqref{eq:tb1}. This concludes the proof.

\section{The proof of Theorems \ref{thm:cate} and \ref{thm:cate bootstrap} for the split-sample estimator}
\label{sec:split}
Recall that $n=N/K$. Let $D_n = \text{diag}((nh^d)^{-1},(nh^{d+2})^{-1}\iota_d) = K D_N$. We aim to show 
\begin{align}
\label{eq:tb2}
\begin{pmatrix}
\breve{\tau}^*(x_1) - \tau_0(x_1) \\
\breve{\beta}^*(x_1) - \beta_0(x_1)
\end{pmatrix} = F^{-1}(x_1)\left(D_N \sum_{i \in I}\xi_i^*\begin{pmatrix}
1 \\
X_{1i} - x_1
\end{pmatrix} (\psi(W_i,\eta) - \tau_0(X_1)) \mathcal{K}_h(X_{1i}-x_1)\right) + \begin{pmatrix}
\check{R}_{\tau}(x_1) \\
\check{R}_{\beta}(x_1) \\
\end{pmatrix},
\end{align}
such that $\sup_{x_1 \in \overline{\mathcal{X}}_1 }|\check{R}_{\tau}(x_1)| = o_p( (\log(N)Nh^d)^{-1/2})$ and $\sup_{x_1 \in \overline{\mathcal{X}}_1 }||\check{R}_{\beta}(x_1)||_2 = o_p( (\log(N)Nh^{d+2})^{-1/2})$. When $\xi_i^*=1$, \eqref{eq:tb2} leads to the second result in Theorem \ref{thm:cate}. Taking the difference of \eqref{eq:tb2} for $\xi_i^* = \xi_i$ and $\xi_i^* =1$, we obtain the second result in Theorem \ref{thm:cate bootstrap}.

Following the same argument in the previous section, we have
\begin{align*}
\begin{pmatrix}
\hat{\tau}_k^*(x_1) \\
\hat{\beta}_k^*(x_1)
\end{pmatrix} = & \left[\sum_{i \in I_k} \xi_i^*D_n\begin{pmatrix}
1 & X_{1i} - x_1\\
X_{1i} - x_1 & (X_{1i} - x_1)(X_{1i} - x_1)^T
\end{pmatrix}\mathcal{K}_h(X_{1i}-x_1) \right]^{-1}\\
& \times \left[D_n \sum_{i \in I_k}\xi_i^*\begin{pmatrix}
1 \\
X_{1i} - x_1
\end{pmatrix} \psi(W_i,\hat{\eta}(I_k^c)) \mathcal{K}_h(X_{1i}-x_1)  \right] \\
\equiv &  (F^*_{n,k}(x_1))^{-1}G_{n,k}^*(x_1).
\end{align*}	
Similar to \eqref{eq:Finver}, we have 
\begin{align}
\label{eq:Finvers}
(F^*_{n,k}(x_1))^{-1} = F^{-1}(x_1) + R_{F,k}(x_1),
\end{align}
where $F^{-1}(x_1)$ is defined as in the previous section and 
$$ \sup_{x_1 \in \overline{\mathcal{X}}_1, k =1,\cdots,K}| R_{F,k}(x_1) | = \begin{pmatrix}
O_p(\log^{1/2}(N)(Nh^d)^{-1/2})&  O_p(h^2)\\
O_p(\log^{1/2}(N)(Nh^{d+2})^{-1/2}) & O_p(\log^{1/2}(N)(Nh^d)^{-1/2}) I_d
\end{pmatrix}.$$

For $G_{n,k}^*(x_1)$, we have 
\begin{align*}
G_{n,k}^*(x_1) = D_n \sum_{i \in I_k}\xi_i^*\begin{pmatrix}
1 \\
X_{1i} - x_1
\end{pmatrix} \psi(W_i,\eta) \mathcal{K}_h(X_{1i}-x_1) + R_{n,k}^*(x_1),
\end{align*}
where $R_{n,k}^*(x_1) = (R_{n,1,k}^*(x_1),(R_{n,2,k}^*(x_1))^T)^T$, 
\begin{align*}
R_{n,1,k}^*(x_1) = \sum_{i \in I}\frac{\xi_i^*}{nh^d}\left(\psi(W_i,\hat{\eta}(I_k^c))-\psi(W_i,\eta)\right) \mathcal{K}_h(X_{1i}-x_1), 
\end{align*}
and 
\begin{align*}
R_{n,2,k}^*(x_1) = \sum_{i \in I_k}\frac{\xi_i^*}{nh^{d+2}}\left(X_{1i} - x_1\right) \left(\psi(W_i,\hat{\eta}(I_k^c))-\psi(W_i,\eta)\right) \mathcal{K}_h(X_{1i}-x_1). 
\end{align*}
By Lemma \ref{lem:2}, we have  
\begin{align}
\label{eq:R1k}
\sup_{x_1 \in \overline{\mathcal{X}}_1,k=1,\cdots,K}|R_{n,1,k}^*(x_1)| = o_p((\log(N)Nh^d)^{-1/2})
\end{align}
and 
\begin{align}
\label{eq:R2k}
\sup_{x_1 \in \overline{\mathcal{X}}_1,k=1,\cdots,K}|R_{n,2,k}^*(x_1)| = o_p((\log(N)Nh^{d+2})^{-1/2}).
\end{align}

Then, following the same argument as those in the previous section, we have 
\begin{align}
\label{eq:A0k}
&\begin{pmatrix}
\hat{\tau}_k^*(x_1) - \tau_0(x_1) \\
\hat{\beta}_k^*(x_1) - \beta_0(x_1)
\end{pmatrix} \notag \\
= & \left[(F^*_{n,k}(x_1))^{-1}\left(D_n \sum_{i \in I_k}\xi_i^*\begin{pmatrix}
1 \\
X_{1i} - x_1
\end{pmatrix} (\psi(W_i,\eta) - \tau_0(X_1)) \mathcal{K}_h(X_{1i}-x_1)\right) \right] \notag \\
& + \left[(F^*_{n,k}(x_1))^{-1}\left(D_n \sum_{i \in I_k}\xi_i^*\begin{pmatrix}
1 \\
X_{1i} - x_1
\end{pmatrix} (\tau_0(X_1) - \tau_0(x_1) - (X_1-x_1)^T\beta_0(x_1) ) \mathcal{K}_h(X_{1i}-x_1)\right)\right] \notag \\
& +\biggr[(F^*_{n,k}(x_1))^{-1}R_{n,k}^*(x_1) \biggl] \notag \\
\equiv & A_{1k}(x_1) + A_{2k}(x_2) + A_{3k}(x_1). 
\end{align}
We can show that 
\begin{align}
\label{eq:A1k}
A_{1k}(x_1) = F^{-1}(x_1)\left(D_n \sum_{i \in I_k}\xi_i^*\begin{pmatrix}
1 \\
X_{1i} - x_1
\end{pmatrix} (\psi(W_i,\eta) - \tau_0(X_1)) \mathcal{K}_h(X_{1i}-x_1)\right) + R_{1,A,k}(x_1),
\end{align}
\begin{align}
\label{eq:A2k}
\sup_{ x_1 \in \overline{\mathcal{X}}_1 }||A_{2k}(x_1)||_2 = O_p(h^2),
\end{align}
and component-wise 
\begin{align}
\label{eq:A3k}
\sup_{ x_1 \in \overline{\mathcal{X}}_1 }|A_{3k}(x_1)| = \begin{pmatrix}
o_p((\log(N)Nh^d)^{-1/2}) \\
o_p((\log(N)Nh^{d+2})^{-1/2})\iota_d
\end{pmatrix},
\end{align}
where 
\begin{align*}
\sup_{k \leq K, x_1 \in \overline{\mathcal{X}}_1 }|R_{1,A,k}(x_1)| = \sup_{ x_1 \in \overline{\mathcal{X}}_1 }|R_{F,k}(x_1)| \times \begin{pmatrix}
O_p(\sqrt{\log(N)/(Nh^d)}) \\
O_p(\sqrt{\log(N)/(Nh^{d+2})})\iota_d \\
\end{pmatrix} = \begin{pmatrix}
o_p((\log(N)Nh^d)^{-1/2}) \\
o_p((\log(N)Nh^{d+2})^{-1/2})\iota_d
\end{pmatrix}.
\end{align*}
Combining \eqref{eq:A0k}--\eqref{eq:A3k}, we have 
\begin{align*}
\begin{pmatrix}
\hat{\tau}_k(x_1) - \tau_0(x_1) \\
\hat{\beta}_k(x_1) - \beta_0(x_1)
\end{pmatrix} = F^{-1}(x_1)\left(D_n \sum_{i \in I_k}\xi_i^*\begin{pmatrix}
1 \\
X_{1i} - x_1
\end{pmatrix} (\psi(W_i,\eta) - \tau_0(X_1)) \mathcal{K}_h(X_{1i}-x_1)\right) + \begin{pmatrix}
R_{\tau,k}(x_1) \\
R_{\beta,k}(x_1) \\
\end{pmatrix},
\end{align*}
such that $\sup_{k \leq K, x_1 \in \overline{\mathcal{X}}_1 }|R_{\tau,k}(x_1)| = o_p( (\log(N)Nh^d)^{-1/2})$ and $\sup_{k \leq K, x_1 \in \overline{\mathcal{X}}_1 }||R_{\beta,k}(x_1)||_2 = o_p( (\log(N)Nh^{d+2})^{-1/2})$. Taking average over $k$ on both sides, we obtain \eqref{eq:tb2}. This concludes the proof.

\section{Proof of Theorem \ref{thm: var consistency}}
\label{sec:var}
We focus on the split-sample estimator $\check{\sigma}_N^2(x_1)$. The proof for the full-sample estimator $\widehat{\sigma}_N^2(x_1)$ is similar but simpler. Let $$\sigma_k^2(x_1) = Var\left(\frac{\sqrt{nh^d}}{h^d f(x_1)}\P_{n,k}(\psi(W,\eta_0) - \tau_0(x_1))\mathcal{K}_h\left(X_1-x_1\right)\right)$$ and recall that $\sigma^2_N(x_1)$ is defined as
$$\sigma^2_N(x_1)=Var\left(\frac{\sqrt{Nh^d}}{h^d f(x_1)}\P_{N}(\psi(W,\eta_0) - \tau_0(x_1))\mathcal{K}_h\left(X_1-x_1\right)\right).$$

Then, we have
\begin{align*}
\frac{1}{K}\sum_{k=1}^K \sigma_k^2(x_1) = \sigma^2_N(x_1).
\end{align*}
Therefore, it suffices to show that
\begin{align*}
\sup_{k \leq K, x_1 \in \overline{\mathcal{X}}_1 }|\check{\sigma}^2_k(x_1) - \sigma_k^2(x_1)| = o_p(1).
\end{align*}

Let $\Gamma(W,x_1) =\frac{\psi(W_i,\eta_0)-\tau_0(x_1) }{h^d f(x_1)} \mathcal{K}_h(X_1-x_1)$ and $\ddot{\sigma}_k^2(x_1)=h^d\P_{n,k}(\Gamma(W,x_1))^2$.  We aim to show that, for $k=1,\cdots,K$,
\begin{align}
\sup_{x_1\in\overline{\mathcal{X}}_1}|\ddot{\sigma}_k^2(x_1)-\sigma_k^2(x_1)|=o_p(1)
\label{eq: ddot sigma}
\end{align}
and
\begin{align}
\sup_{x_1\in\overline{\mathcal{X}}_1}|\ddot{\sigma}_k^2(x_1)-\check{\sigma}_k^2(x_1)|=o_p(1).
\label{eq: diff sigma}
\end{align}
We first show (\ref{eq: ddot sigma}). We claim that $\sup_{x_1\in\overline{\mathcal{X}}_1}|Var(\sqrt{nh^d}\P_{n,k}\Gamma(W,x_1))-\E(nh^d[\P_{n,k} \Gamma(W,x_1)]^2)|=o_p(1).$  Because $Var(A)=\E[A^2]-\E[A]^2$, it is equivalent to show that $\E[\sqrt{nh^d}\P_{n,k}\Gamma(W,x_1)]=o(1)$ uniformly over $x_1$.  By standard arguments and Assumption \ref{ass:regularity}, we have, uniformly over $x_1$
\begin{align*}
\sup_{x_1 \in \overline{\mathcal{X}}_1 }\E[\sqrt{nh^d}\P_{n,k}\Gamma(W,x_1)]= O\left(\sqrt{nh^d}h^2\right)=o(1),
\end{align*}
Similarly, we have
\begin{align*}
\sup_{x_1 \in \overline{\mathcal{X}}_1 }|\E(nh^d[\P_{n,k} \Gamma(W,x_1)]^2) - \E(h^d\P_{n,k} \Gamma^2(W,x_1))| = O(nh^{d+4}) = o(1).
\end{align*}

We next show that
\begin{align}
\label{eq:a13'}
\sup_{x_1\in\overline{\mathcal{X}}_1}|h^d\P_{n,k}\Gamma^2(W,x_1)-\E [h^d\P_{n,k}\Gamma^2(W,x_1)]|=o_p(1).
\end{align}

By \citet[Section 2.6]{VW96}, we have
\begin{align*}
\mathcal{F}_K=\left\{ \mathcal{K}\Big(\frac{X_1-x_1}{h}\Big): x_1\in\overline{\mathcal{X}}_1 \right\}
\end{align*}
is of VC type with envelop function $\overline{K}=\sup_{u}|\mathcal{K}(u)|$ which is bounded.  This implies that
\begin{align*}
\mathcal{F}_{K^2}=\left\{ \mathcal{K}^2\Big(\frac{X_1-x_1}{h}\Big): x_1\in\overline{\mathcal{X}}_1 \right\}
\end{align*}
is of VC type with envelop function $\overline{K}^2$. Similarly,
\begin{align*}
\mathcal{F}_{h^d\Gamma^2}=\left\{h^d\Gamma^2(W,x_1): x_1\in\overline{\mathcal{X}}_1 \right\}
\end{align*}
is of VC type with an envelop function $Ch^{-d}\overline{K}^2 \cdot (\psi(W,\eta_0)-\tau_0(x_1))^2$ for some constant $C>0$. In addition,
\begin{align*}
\sup_{x_1 \in \overline{\mathcal{X}}_1 }\E h^{2d}\Gamma^4(W,x_1) \lesssim \sup_{x_1 \in \overline{\mathcal{X}}_1 } h^{-2d}\E \left(\mathcal{K}^4\Big(\frac{X_1-x_1}{h}\Big) \right) \lesssim h^{-d}.
\end{align*}
Therefore, by \citet[Corollary 5.1]{CCK14}, we have
\begin{align}
\label{eq:Epnk}
\E \left[\sup_{x_1\in\overline{\mathcal{X}}_1}|h^d\P_{n,K}\Gamma^2(W,x_1)-\E [h^d\P_{n,k}\Gamma^2(W,x_1)]|\right]\lesssim \sqrt{\frac{\log(n)}{nh^d}}+\frac{\log(n)n^{1/q}}{nh^d} = o(1),
\end{align}
implying that \eqref{eq:a13'} holds, and thus,
\begin{align*}
\sup_{x_1\in\overline{\mathcal{X}_1}}
\left|\ddot{\sigma}_k^2(x_1)-\sigma_k^2(x_1)\right| =o_p(1).
\end{align*}
This shows (\ref{eq: ddot sigma}). Next, we show (\ref{eq: diff sigma}). Denote
\begin{align*}
\widetilde{\Gamma}(W,x_1) =  \frac{(\psi(W,\hat{\eta}(I_k^c))-\check{\tau}_k(x_1))}{\hat{f}(x_1;I_k)h^d}\mathcal{K}_h(X_1-x_1)\quad \text{and} \quad R_k(W_i,x_1) = \widetilde{\Gamma}(W_i,x_1) - \Gamma(W_i,x_1).
\end{align*}
If
\begin{align}
\label{eq:Rnk}
\sup_{x_1\in\overline{\mathcal{X}}_1}h^d||R_k(\cdot,x_1)||_{P_{n,k},2}^2 = o_p(1),
\end{align}
then
\begin{align*}
&\sup_{x_1\in\overline{\mathcal{X}}_1}| \ddot{\sigma}_k^2(x_1)-\check{\sigma}_k^2(x_1) |\\
\leq & \sup_{ x_1 \in \overline{\mathcal{X}}_1 } h^d \left[2|\P_{n,k}\Gamma(W,x_i)R_k(W,x_1)| + \P_{n,k}R^2_k(W,x_1) \right]\\
\leq & \sup_{x_1\in\overline{\mathcal{X}}_1}\biggl[2h^d || \Gamma(\cdot,x_1)||_{\P_{n,k},2} || R_k(\cdot,x_1)||_{\P_{n,k},2} +h^d|| R_k(\cdot,x_1)||_{\P_{n,k},2}^2 \biggr]\\
= & 2\sqrt{\left[\sup_{x_1\in\overline{\mathcal{X}}_1}|h^d\P_{n,K}\Gamma^2(W,x_1)-\E[h^d\P_{n,k}\Gamma^2(W,x_1)]|\right] + \sup_{ x_1 \in \overline{\mathcal{X}}_1 } \E[h^d\P_{n,k}\Gamma^2(W,x_1)]} \\
& \times h^{d/2}\sup_{x_1\in\overline{\mathcal{X}}_1}||R_k(\cdot,x_1)||_{P_{n,k},2}  + h^d|| R_k(\cdot,x_1)||_{P_{n,k},2}^2 \\
= &o_p(1),
\end{align*}
where the last equality holds because of \eqref{eq:Epnk}, \eqref{eq:Rnk}, and the fact that $\sup_{ x_1 \in \overline{\mathcal{X}}_1 } \E[h^d\P_{n,k}\Gamma^2(W,x_1)]$ is bounded. Therefore, we only need to verify \eqref{eq:Rnk}. We have
\begin{align*}
& \sup_{x_1\in\overline{\mathcal{X}}_1}h^d||R_k(\cdot,x_1)||_{\P_{n,k},2}^2 \\
\leq & \sup_{ x_1 \in \overline{\mathcal{X}}_1 }\frac{|\hat{f}^2(x_1;I_k) - f^2(x_1)|}{\hat{f}^2(x_1;I_k)f^2(x_1)}\frac{1}{nh^d}\sum_{i \in I_k} (\psi(W,\eta_0)-\tau_0(x_1))^2\mathcal{K}^2_h\left(X_{1i}-x_1\right) \\
& + \sup_{ x_1 \in \overline{\mathcal{X}}_1 } \frac{1}{\hat{f}^2(x_1;I_k)}\frac{1}{nh^d}\sum_{i \in I_k}\left(\frac{D}{\hat{\pi}(X;I_k^c)} -1\right)^2(\hat{\mu}(1,X;I_k^c) - \mu_0(1,X))^2 \mathcal{K}^2_h\left(X_{1i}-x_1\right) \\
& + \sup_{ x_1 \in \overline{\mathcal{X}}_1 } \frac{1}{\hat{f}^2(x_1;I_k) }\frac{1}{nh^d}\sum_{i \in I_k}\left(\frac{1-D}{1-\hat{\pi}(X;I_k^c)} -1\right)^2(\hat{\mu}(1,X;I_k^c) - \mu_0(1,X))^2 \mathcal{K}^2_h\left(X_{1i}-x_1\right) \\
& + \sup_{ x_1 \in \overline{\mathcal{X}}_1 } \frac{1}{\hat{f}^2(x_1;I_k)}\frac{1}{nh^d}\sum_{i \in I_k}\left(\frac{\hat{\pi}(X;I_k^c) - \pi_0(X) }{P(X)\hat{\pi}(X;I_k^c)}D(Y - \mu_0(1,X))\right)^2 \mathcal{K}^2_h\left(X_{1i}-x_1\right) \\
& + \sup_{ x_1 \in \overline{\mathcal{X}}_1 } \frac{1}{\hat{f}^2(x_1;I_k)}\frac{1}{nh^d}\sum_{i \in I_k}\left(\frac{\hat{\pi}(X;I_k^c) - \pi_0(X) }{\pi_0(X)\hat{\pi}(X;I_k^c)}(1-D)(Y - \mu_0(0,X))\right)^2 \mathcal{K}^2_h\left(X_{1i}-x_1\right) \\
\equiv & \sum_{q=1}^5 A_{qk}.
\end{align*}
We want to show $A_{qk} = o_p(1)$ for $q = 1,\cdots,5$, $k=1,\cdots,K$. Note $A_{1k} = o_p(1)$ because
\begin{align*}
\sup_{ x_1 \in \overline{\mathcal{X}}_1 }|\hat{f}(x_1;I_k) - f(x_1)| = O_p\left(\sqrt{\frac{\log(n)}{nh^d}}\right) = o_p(1),
\end{align*}
\begin{align*}
\sup_{ x_1 \in \overline{\mathcal{X}}_1 }\E\frac{1}{nh^d}\sum_{i \in I_k} (\psi(W,\eta_0) - \tau_0(x_1))^2\mathcal{K}^2_h\left(X_{1i}-x_1\right) = O(1),
\end{align*}
and
\begin{align*}
\sup_{ x_1 \in \overline{\mathcal{X}}_1 }(\P_{n,k}-\P)h^{-d} (\psi(W,\eta_0) - \tau_0(x_1))^2\mathcal{K}^2_h\left(X_{1i}-x_1\right) = o_p(1).
\end{align*}
For $A_{2k}$, by Assumptions \ref{ass:regularity}(ii) and \ref{assn: 1st stage split},
\begin{align*}
\left\Vert \left(\frac{D}{\hat{\pi}(X;I_k^c)} -1\right)^2\right\Vert_{\P,\infty} = O_p(1).
\end{align*}
Therefore,
\begin{align*}
0 \leq A_{2k} \leq O_p(1) \times h^{-d}\delta_{2n}^2\sup_{ x_1 \in \overline{\mathcal{X}}_1 }\left\Vert\mathcal{K}\left(\frac{X_1 - x_1}{h^d} \right) \right\Vert_{\P_{n,k},2}^2 = O_p(\delta_{2n}^2) = o_p(1).
\end{align*}
Similarly, $A_{3k} = o_p(1)$. Next, by Assumption \ref{assn: 1st stage split},
\begin{align*}
0 \leq A_{4k} \leq O_p(\delta^2_{n2}) \sup_{ x_1 \in \overline{\mathcal{X}}_1 } \P_{n,k}h^{-d}(D(Y-\mu_1(X)))^2\mathcal{K}^2_h\left(X_{1i}-x_1\right) = O_p(\delta^2_{n2}) = o_p(1).
\end{align*}
Similarly, $A_{5k} = o_p(1)$. This completes the derivation of \eqref{eq:Rnk}, and thus, the whole proof.

\section{Proof of Theorem \ref{thm:boot-CS}}
\label{sec:CS}
Theorem \ref{thm:boot-CS} is a direct consequence of \citet[Corollary 3.1 ]{CCK14-anti}. In order to apply \citet[Corollary 3.1 ]{CCK14-anti}, we need to verify Conditions H1--H4. Our Theorems \ref{thm:cate} and \ref{thm:cate bootstrap} have already established that the original and multiplier bootstrap estimators can be approximated by local empirical processes with a kernel function and the approximation errors are $o_P((\log(n))^{-1/2})$ uniformly over $x_1 \in \overline{X}_1$. Then, following \citet[Proposition 3.2 and Remark 3.2]{CCK14}, the approximation errors are asymptotically negligible. Focusing on the local empirical process part, Conditions H1--H4 can be verified by \citet[Theorem 3.2]{CCK14-anti}. Specifically, Condition VC in \cite{CCK14-anti} holds where, in their notation, $a_n$ and $v_n$ are constants, $b_n = h^{-d/2}$, $K_n = \log(n)$, $\sigma_n^2$ is bounded, and 
\begin{align*}
\log^4(n)/nh^d = o(n^{-c}),
\end{align*}
for some constant $c>0$ as we assume $h = cn^{-H}$ for $H < 1/d$.

%
%

\section{Technical Lemmas}
\label{sec:lem}
\begin{lemma}
	\label{lem:1}
	If the assumptions in Theorems \ref{thm:cate} and \ref{thm:cate bootstrap} hold, then \eqref{eq:R1} and \eqref{eq:R2} hold. 
\end{lemma}
\begin{proof}
	Note that 
	\begin{align*}
	\psi(W_i,\hat{\eta}(I))-\psi(W_i,\eta) = \psi(1,W_i,\hat{\eta}(I))-\psi(1,W_i,\eta) - (\psi(0,W_i,\hat{\eta}(I))-\psi(0,W_i,\eta)).
	\end{align*}
	Thus, we only need to derive the bound for 
	\begin{align*}
	\tilde{R}_{n,1}^*(x_1) = \sum_{i \in I}\frac{\xi_i^*}{Nh^{d}} \left(\psi(1,W_i,\hat{\eta}(I))-\psi(1,W_i,\eta)\right) \mathcal{K}_h(X_{1i}-x_1)
	\end{align*}
	and the bound for 
	\begin{align*}
	\sum_{i \in I}\frac{\xi_i^*}{Nh^{d}}\left(\psi(0,W_i,\hat{\eta}(I))-\psi(0,W_i,\eta)\right) \mathcal{K}_h(X_{1i}-x_1)
	\end{align*}	
	can be established in the same manner. We have 
	\begin{align*}
	& \tilde{R}_{n,1}^*(x_1)\\
	= & (\P_N - \P)\frac{\xi_i^*}{h^{d}} \left(\psi(1,W_i,\hat{\eta}(I))-\psi(1,W_i,\eta)\right) \mathcal{K}_h(X_{1i}-x_1) + \P\frac{\xi_i^*}{h^{d}} \left(\psi(1,W_i,\hat{\eta}(I))-\psi(1,W_i,\eta)\right) \mathcal{K}_h(X_{1i}-x_1) \\
	\equiv & I(x_1) + II(x_1). 
	\end{align*}
	
	By Assumption \ref{assn: 1st stage full},  for any $\eps>0$, there exists a positive constant $M$, such that, with probability greater than $1-\eps$, $(\widehat{\mu}(j,\cdot),\widehat{\pi}(\cdot)) \in  \mathcal{F}_{n}^{(j)},$ where
	\begin{align*}
	\mathcal{F}_n^{(j)} = \begin{Bmatrix} & (\pi(\cdot),\mu(j,\cdot)) \in \mathcal{G}_{n}^{\pi}) \times \mathcal{G}_{n}^{(j)}:  \sup_{ x_1 \in \overline{\mathcal{X}}_1 }\left\Vert(\mu(j,\cdot) - \mu_0(j,\cdot))\mathcal{K}_h^{1/2}\left(X_1-x_1\right)\right\Vert_{\P,2} \\
	& \times \left\Vert(\pi(\cdot) - \pi_0(\cdot))\mathcal{K}_h^{1/2}\left(X_1-x_1\right)\right\Vert_{\P,2} \leq M\delta_{1N}^2, \\
	& ||\mu(0,X)-\mu_0(0,X)||_{\P,\infty} \leq M\delta_{2n}, \quad ||\pi(X)-\pi_0(X)||_{\P,\infty} \leq M\delta_{2n} \\
	&  \sup_{ x_1 \in \overline{\mathcal{X}}_1 }\left\Vert(\mu(j,\cdot) - \mu_0(j,\cdot))||X_1 - x_1||_2^{1/2}\mathcal{K}_h^{1/2}\left(X_1-x_1\right)\right\Vert_{\P,2} \\
	& \times \left\Vert(\pi(\cdot) - \pi_0(\cdot))||X_1 - x_1||_2^{1/2}\mathcal{K}_h^{1/2}\left(X_1-x_1\right)\right\Vert_{\P,2} \leq M\delta_{3N}^2,
	\end{Bmatrix}, \quad j=0,1.
	\end{align*}
	Then, with probability greater than $1-\eps$, 
	\begin{align*}
	& \sup_{ x_1 \in \overline{\mathcal{X}}_1 }|I(x_1)| \\
	\leq & \sup_{(\mu(1,\cdot), \pi(\cdot)) \in \mathcal{F}_n^{(1)}, x_1 \in \overline{\mathcal{X}}_1 }\biggl|(\P_n - \P)\xi^*\biggl[\frac{D(Y-\mu(1,X))(\pi_0(X) - \pi(X))}{h^{d}\pi(X)\pi_0(X)}\biggr]\mathcal{K}_h\left(X_1-x_1\right)\biggr| \\
	& + \sup_{(\mu(1,\cdot), \pi(\cdot)) \in \mathcal{F}_n^{(1)}, x_1 \in \overline{\mathcal{X}}_1 }\biggl|h^{-d}(\P_n - \P)\xi^*\biggl[\biggl(1-\frac{D}{\pi(X)}\biggr)(\mu(1,X) - \mu_0(1,X)) \biggr]  \mathcal{K}_h\left(X_1-x_1\right)\biggr| \\
	= & I_1+ I_2.
	\end{align*}
	
	Uniformly over $x \in \mathcal{X}$,  $\pi_0(x)$ is bounded and bounded away from zero and $\mu_0(1,x)$ is bounded. Therefore, so are $(\pi(x),\mu(1,x)) \in \mathcal{F}_n^{(1)}$ as $\delta_{2N} = o(1)$. Therefore, element-wise, 
	\begin{equation}
	\begin{aligned}
	\sup_{(\mu(1,\cdot), \pi(\cdot)) \in \mathcal{F}_n^{(1)}, x_1 \in \overline{\mathcal{X}}_1 }\biggl|\xi^*\biggl[\frac{D(Y-\mu(1,X))(\pi_0(X) - \pi(X))}{h^{d}\pi(X)\pi_0(X)}\biggr]\mathcal{K}_h\left(X_1-x_1\right)\biggr| \lesssim h^{-d}\delta_{2N}|\xi^*Y|.
	\label{eq:upboundf}
	\end{aligned}
	\end{equation}
	In addition, we have
	\begin{align}
	\label{eq:sigmaf}
	& \sup_{(\mu(1,\cdot), \pi(\cdot)) \in \mathcal{F}_n^{(1)}, x_1 \in \overline{\mathcal{X}}_1 }\E \biggl| \xi^*\biggl[\frac{D(Y-\mu(1,X))(\pi_0(X) - \pi(X))}{h^{d}\pi(X)\pi_0(X)}\biggr]\mathcal{K}_h\left(X_1-x_1\right)\biggr|^2 \notag \\
	\lesssim & h^{-2d}\sup_{\pi(\cdot) \in \mathcal{F}_n^{(1)} }\E(\pi_0(X)-\pi(X))^2\mathcal{K}_h\left(X_1-x_1\right) \lesssim h^{-d}\delta_{2N}^2.
	\end{align}

	Denote
	$$\mathcal{H}_1 = \biggl\{\xi^*\biggl[\frac{D(Y-\mu(1,X))(\pi_0(X) - \pi(X))}{h^d\pi(X)\pi_0(X)}\biggr]\mathcal{K}_h\left(X_1-x_1\right):(\mu(1,\cdot), \pi(\cdot)) \in \mathcal{F}_n^{(1)}, x_1 \in \overline{\mathcal{X}}_1 \biggr\}.$$
	Combining \eqref{eq:G01} and the fact that
	$$\sup_Q \log N\biggl(\left\{K(\frac{\cdot - x_1}{h}): x_1 \in \Re\right\},||\cdot||_{Q,2},\eps\biggr) \lesssim \log(1/\eps)\vee 0, $$
	we have
	\begin{equation}
	\begin{aligned}
	\sup_Q \log N(\mathcal{H}_1,||\cdot||_{Q,2},\eps||F_1||_{Q,2}) \lesssim \delta_{4N}(\log(A_n) + \log(1/\eps)\vee 0).
	\label{eq:entropyf}
	\end{aligned}
	\end{equation}
	
	Since $\xi^*$ is either $1$ or $\eta$ which has a sub-exponential tail and $\E Y^q < \infty$, we have $||\max_i |\xi_i^*Y_i|||_2 \lesssim N^{1/q}.$ Combining this fact with \eqref{eq:upboundf}, \eqref{eq:sigmaf}, and \eqref{eq:entropyf}, \citet[Lemma C.1]{BCFH13} implies that
	\begin{equation*}
	\begin{aligned}
	\E I_1 \lesssim \sqrt{\frac{\delta_{4N} \log(A_N \vee N) \delta_{2N}^2}{Nh^{d}}} + \frac{\delta_{2N}\delta_{4N}N^{1/q} \log(A_N \vee N)}{Nh^d},
	\end{aligned}
	\end{equation*}
	and thus, by Assumption \ref{assn: 1st stage full},
	\begin{equation}
	\begin{aligned}
	I_1 = o_p((\log(N)Nh^d)^{-1/2}).
	\label{eq:IIf}
	\end{aligned}
	\end{equation}
	Similarly, for $I_2$, we have
	\begin{align*}
	\sup_{(\mu(1,\cdot), \pi(\cdot)) \in \mathcal{F}_n^{(1)}, x_1 \in \overline{\mathcal{X}}_1 }\biggl|h^{-d}\xi^*\biggl[\biggl(1-\frac{D}{\pi(X)}\biggr)(\mu(1,X) - \mu_0(1,X)) \biggr]  \mathcal{K}_h\left(X_1-x_1\right)\biggr|   \lesssim  \delta_{2N}h^{-d} |\xi^*|,
	\end{align*}
	
	\begin{equation*}
	\begin{aligned}
	&  \sup_{(\mu(1,\cdot), \pi(\cdot)) \in \mathcal{F}_n^{(1)}, x_1 \in \overline{\mathcal{X}}_1 }\E\biggl\{h^{-d}\xi^*\biggl[\biggl(1-\frac{D}{\pi(X)}\biggr)(\mu(1,X) - \mu_0(1,X)) \biggr]  \mathcal{K}_h\left(X_1-x_1\right)\biggr\}^2 \lesssim h^{-d}\delta_{2N}^2,
	\end{aligned}
	\end{equation*}
	and
	$$\sup_Q \log N(\mathcal{H}_2,||\cdot||_{Q,2},\eps||F_1||_{Q,2}) \lesssim \delta_{4N}(\log(A_n) + \log(1/\eps)\vee 0).,$$
	where
	$$\mathcal{H}_2 = \biggl\{h^{-d}\xi^*\biggl[\biggl(1-\frac{D}{\pi(X)}\biggr)(\mu(1,X) - \mu_0(1,X)) \biggr]  \mathcal{K}_h\left(X_1-x_1\right): (\mu(1,\cdot), \pi(\cdot)) \in \mathcal{F}_n^{(1)}, x_1 \in \overline{\mathcal{X}}_1  \biggr\}.$$
	Therefore, by the same argument as above,
	\begin{equation}
	\begin{aligned}
	I_2 =o_p((\log(N)Nh^d)^{-1/2}).
	\label{eq:II_2f}
	\end{aligned}
	\end{equation}
	\eqref{eq:IIf} and \eqref{eq:II_2f} imply that
	$$\sup_{x_1 \in \overline{\mathcal{X}}_1}|I(x_1)| = o_p((\log(N)Nh^d)^{-1/2}).$$	
	
	For $II(x_1)$, we have 
	\begin{align*}
	& \sup_{ x_1 \in \overline{\mathcal{X}}_1 }|II(x_1)| \\
	\lesssim & \sup_{(\mu(1,\cdot), \pi(\cdot)) \in \mathcal{F}_n^{(1)}, x_1 \in \overline{\mathcal{X}}_1}\left|\P \left[\frac{(\mu_0(1,X) - \mu(1,X))(\pi_0(X) - \pi(X))}{h^d \pi(X)}\mathcal{K}_h(X_1 -x_1)\right]\right| \\
	\lesssim & \sup_{(\mu(1,\cdot), \pi(\cdot)) \in \mathcal{F}_n^{(1)}, x_1 \in \overline{\mathcal{X}}_1}h^{-d} \left \Vert (\mu_0(1,X) - \mu(1,X)) \mathcal{K}^{1/2}_h(X_1 -x_1) \right \Vert_{\P,2} \left \Vert (\pi_0(X) - \pi(1,X)) \mathcal{K}^{1/2}_h(X_1 -x_1) \right \Vert_{\P,2} \\
	\lesssim & h^{-d}\delta_{1N}^2 = o((\log(N)Nh^d)^{-1/2}). 
	\end{align*}
	Therefore, we have established \eqref{eq:R1}. 
	
	For \eqref{eq:R2}, similarly, we only need to derive the bound for 
	\begin{align*}
	\tilde{R}_{n,2}^*(x_1) = \sum_{i \in I}\frac{\xi_i^*}{Nh^{d+2}}\left(X_{1i} - x_1\right) \left(\psi(1,W_i,\hat{\eta}(I))-\psi(1,W_i,\eta)\right) \mathcal{K}_h(X_{1i}-x_1)
	\end{align*}
	Then, the bound for 
	\begin{align*}
	\sum_{i \in I}\frac{\xi_i^*}{Nh^{d+2}}\left(X_{1i} - x_1\right) \left(\psi(0,W_i,\hat{\eta}(I))-\psi(0,W_i,\eta)\right) \mathcal{K}_h(X_{1i}-x_1)
	\end{align*}	
	will follow. We have 
	\begin{align*}
	\tilde{R}_{n,2}^*(x_1) = & (\P_N - \P)\frac{\xi_i^*}{h^{d+2}}\left(X_{1i} - x_1\right) \left(\psi(1,W_i,\hat{\eta}(I))-\psi(1,W_i,\eta)\right) \mathcal{K}_h(X_{1i}-x_1) \\
	& + \P\frac{\xi_i^*}{h^{d+2}}\left(X_{1i} - x_1\right) \left(\psi(1,W_i,\hat{\eta}(I))-\psi(1,W_i,\eta)\right) \mathcal{K}_h(X_{1i}-x_1) \\
	\equiv & III(x_1) + IV(x_1). 
	\end{align*}
	
	By Assumption \ref{assn: 1st stage full},  for any $\eps>0$, there exists a positive constant $M$, such that component-wise, with probability greater than $1-\eps$,
	\begin{align*}
	& \sup_{ x_1 \in \overline{\mathcal{X}}_1 }|III(x_1)| \\
	\leq & \sup_{(\mu(1,\cdot), \pi(\cdot)) \in \mathcal{F}_n^{(1)}, x_1 \in \overline{\mathcal{X}}_1 }\biggl|(\P_n - \P)\xi^*\biggl[\frac{D(Y-\mu(1,X))(\pi_0(X) - \pi(X))}{h^{d+2}\pi(X)\pi_0(X)}\biggr]\left(X_{1} - x_1\right)\mathcal{K}_h\left(X_1-x_1\right)\biggr| \\
	& + \sup_{(\mu(1,\cdot), \pi(\cdot)) \in \mathcal{F}_n^{(1)}, x_1 \in \overline{\mathcal{X}}_1 }\biggl|h^{-d-2}(\P_n - \P)\xi^*\biggl[\biggl(1-\frac{D}{\pi(X)}\biggr)(\mu(1,X) - \mu_0(1,X)) \biggr]    \left(X_{1} - x_1\right)  \mathcal{K}_h\left(X_1-x_1\right)\biggr| \\
	= & III_1+ III_2.
	\end{align*}
	
	Note $u \kappa(u)$ is bounded. Then, similar to \eqref{eq:upboundf}, we have, component-wise,
	\begin{equation}
	\begin{aligned}
	\sup_{(\mu(1,\cdot), \pi(\cdot)) \in \mathcal{F}_n^{(1)}, x_1 \in \overline{\mathcal{X}}_1 }\biggl|\xi^*\biggl[\frac{D(Y-\mu(1,X))(\pi_0(X) - \pi(X))}{h^{d+2}\pi(X)\pi_0(X)}\biggr]\left(X_{1} - x_1\right)\mathcal{K}_h\left(X_1-x_1\right)\biggr| \lesssim h^{-d-1}\delta_{2N}|\xi^*Y|.
	\label{eq:upboundfprime}
	\end{aligned}
	\end{equation}
	In addition, we have
	\begin{align}
	\label{eq:sigmafprime}
	& \sup_{(\mu(1,\cdot), \pi(\cdot)) \in \mathcal{F}_n^{(1)}, x_1 \in \overline{\mathcal{X}}_1 }\E \biggl\Vert \xi^*\biggl[\frac{D(Y-\mu(1,X))(\pi_0(X) - \pi(X))}{h^{d+2}\pi(X)\pi_0(X)}\biggr]\mathcal{K}_h\left(X_1-x_1\right)\left(X_1-x_1\right)\biggr\Vert_2^2 \notag \\
	\lesssim & h^{-2d-4}\sup_{\pi(\cdot) \in \mathcal{F}_n^{(1)} }\E(\pi_0(X)-\pi(X))^2\Vert X_1-x_1\Vert_2^2\mathcal{K}_h\left(X_1-x_1\right) \lesssim h^{-d-2}\delta_{2N}^2.
	\end{align}

	Denote
	$$\mathcal{H}_3 = \biggl\{\xi^*\biggl[\frac{D(Y-\mu(1,X))(\pi_0(X) - \pi(X))}{h^{d+2}\pi(X)\pi_0(X)}\biggr]\left(X_1 -x_1\right)\mathcal{K}_h\left(X_1-x_1\right):(\mu(1,\cdot), \pi(\cdot)) \in \mathcal{F}_n^{(1)}, x_1 \in \overline{\mathcal{X}}_1 \biggr\}.$$
	Combining \eqref{eq:G01} and the fact that
	$$\sup_Q \log N\biggl(\left\{(\cdot - x_1)K(\frac{\cdot - x_1}{h}): x_1 \in \Re\right\},||\cdot||_{Q,2},\eps\biggr) \lesssim \log(1/\eps)\vee 0, $$
	we have
	\begin{equation}
	\begin{aligned}
	\sup_Q \log N(\mathcal{H}_3,||\cdot||_{Q,2},\eps||F_1||_{Q,2}) \lesssim \delta_{4N}(\log(A_n) + \log(1/\eps)\vee 0).
	\label{eq:entropyfprime}
	\end{aligned}
	\end{equation}
	
	Since $\xi^*$ is either $1$ or $\eta$ which has a sub-exponential tail and $\E Y^q < \infty$, we have $||\max_i |\xi_i^*Y_i|||_2 \lesssim N^{1/q}.$ Combining this fact with \eqref{eq:upboundf}, \eqref{eq:sigmafprime}, and \eqref{eq:entropyfprime}, \citet[Lemma C.1]{BCFH13} implies that
	\begin{equation*}
	\begin{aligned}
	\E III_1 \lesssim \sqrt{\frac{\delta_{4N} \log(A_N \vee N) \delta_{2N}^2}{Nh^{d+2}}} + \frac{\delta_{2N}\delta_{4N}N^{1/q} \log(A_N \vee N)}{Nh^{d+1}},
	\end{aligned}
	\end{equation*}
	and thus, by Assumption \ref{assn: 1st stage full},
	\begin{equation}
	\begin{aligned}
	III_1 = o_p((\log(N)Nh^{d+2})^{-1/2}).
	\label{eq:IIfprime}
	\end{aligned}
	\end{equation}
	Similarly, for $III_2$, we have
	\begin{align*}
	\sup_{(\mu(1,\cdot), \pi(\cdot)) \in \mathcal{F}_n^{(1)}, x_1 \in \overline{\mathcal{X}}_1 }\biggl|h^{-d-2}\xi^*\biggl[\biggl(1-\frac{D}{\pi(X)}\biggr)(\mu(1,X) - \mu_0(1,X)) \biggr] \left(X_1-x_1\right) \mathcal{K}_h\left(X_1-x_1\right)\biggr|   \lesssim  \delta_{2N}h^{-d-1} |\xi^*|,
	\end{align*}
	
	\begin{equation*}
	\begin{aligned}
	&  \sup_{(\mu(1,\cdot), \pi(\cdot)) \in \mathcal{F}_n^{(1)}, x_1 \in \overline{\mathcal{X}}_1 }\E\biggl \Vert h^{-d-2}\xi^*\biggl[\biggl(1-\frac{D}{\pi(X)}\biggr)(\mu(1,X) - \mu_0(1,X)) \biggr]  \mathcal{K}_h\left(X_1-x_1\right)(X_1 - x_1)\biggr \Vert_2^2 \\
	\lesssim & h^{-d-2}\delta_{2N}^2,
	\end{aligned}
	\end{equation*}
	and
	$$\sup_Q \log N(\mathcal{H}_2,||\cdot||_{Q,2},\eps||F_1||_{Q,2}) \lesssim \delta_{4N}(\log(A_n) + \log(1/\eps)\vee 0).,$$
	where
	$$\mathcal{H}_4 = \biggl\{h^{-d-2}\xi^*\biggl[\biggl(1-\frac{D}{\pi(X)}\biggr)(\mu(1,X) - \mu_0(1,X)) \biggr] \left(X_1-x_1\right) \mathcal{K}_h\left(X_1-x_1\right): (\mu(1,\cdot), \pi(\cdot)) \in \mathcal{F}_n^{(1)}, x_1 \in \overline{\mathcal{X}}_1  \biggr\}.$$
	Therefore, by the same argument as above,
	\begin{equation}
	\begin{aligned}
	III_2 =o_p((\log(N)Nh^{d+2})^{-1/2}).
	\label{eq:II_2fprime}
	\end{aligned}
	\end{equation}
	\eqref{eq:IIfprime} and \eqref{eq:II_2fprime} imply that
	$$\sup_{x_1 \in \overline{\mathcal{X}}_1}|III(x_1)| = o_p((\log(N)Nh^{d+2})^{-1/2}).$$	
	
	For $IV(x_1)$, we have 
	\begin{align*}
	& \sup_{ x_1 \in \overline{\mathcal{X}}_1 }||IV(x_1)||_2 \\
	\lesssim & \sup_{(\mu(1,\cdot), \pi(\cdot)) \in \mathcal{F}_n^{(1)}, x_1 \in \overline{\mathcal{X}}_1}\left\Vert \P \left[\frac{(\mu_0(1,X) - \mu(1,X))(\pi_0(X) - \pi(X))}{h^{d+2} \pi(X)}(X_1 -x_1)\mathcal{K}_h(X_1 -x_1)\right]\right\Vert_2  \\
	\lesssim & \sup_{(\mu(1,\cdot), \pi(\cdot)) \in \mathcal{F}_n^{(1)}, x_1 \in \overline{\mathcal{X}}_1}h^{-d-2} \left \Vert (\mu_0(1,X) - \mu(1,X))||X_1 -x_1||_2^{1/2} \mathcal{K}^{1/2}_h(X_1 -x_1) \right \Vert_{\P,2} \\
	& \times \left \Vert (\pi_0(X) - \pi(1,X))||X_1 -x_1||_2^{1/2} \mathcal{K}^{1/2}_h(X_1 -x_1) \right \Vert_{\P,2} \\
	\lesssim & h^{-d-2}\delta_{3N}^2 = o((\log(N)Nh^{d+2})^{-1/2}). 
	\end{align*}
	Therefore, we have established \eqref{eq:R2}. 
\end{proof}

\begin{lemma}
	\label{lem:2}
	If assumptions in Theorems \ref{thm:cate} and \ref{thm:cate bootstrap} hold, then \eqref{eq:R1k} and \eqref{eq:R2k} hold. 
\end{lemma}
\begin{proof}
	We focus on establishing \eqref{eq:R2k}, and \eqref{eq:R1k} can be established in a similar manner. In addition, as before, we focus on bounding 
	\begin{align*}
	\sum_{i \in I_k}\frac{\xi_i^*}{nh^{d+2}}\left(X_{1i} - x_1\right) \left(\psi(1,W_i,\hat{\eta}(I_k^c))-\psi(1,W_i,\eta)\right) \mathcal{K}_h(X_{1i}-x_1).
	\end{align*}	
	We have 
	\begin{align*}
	\tilde{R}_{n,2,k}^*(x_1) = & (\P_{n,k} - \P_{I_k})\frac{\xi_i^*}{h^{d+2}}\left(X_{1i} - x_1\right) \left(\psi(1,W_i,\hat{\eta}(I_k^c))-\psi(1,W_i,\eta)\right) \mathcal{K}_h(X_{1i}-x_1) \\
	& + \P_{I_k}\frac{\xi_i^*}{h^{d+2}}\left(X_{1i} - x_1\right) \left(\psi(1,W_i,\hat{\eta}(I_k^c))-\psi(1,W_i,\eta)\right) \mathcal{K}_h(X_{1i}-x_1) \\
	\equiv & I_k(x_1) + II_k(x_1). 
	\end{align*}
	
	let
	\begin{align*}
	\mathcal{F}_{n}^{(j)} = \begin{Bmatrix} & (\pi(\cdot),\mu(j,\cdot)):  \sup_{ x_1 \in \overline{\mathcal{X}}_1 }\left\Vert(\mu(j,\cdot) - \mu_0(j,\cdot))\mathcal{K}_h^{1/2}\left(X_1-x_1\right)\right\Vert_{\P,2} \\
	& \times \left\Vert(\pi(\cdot) - \pi_0(\cdot))\mathcal{K}_h^{1/2}\left(X_1-x_1\right)\right\Vert_{\P,2} \leq M\delta_{1N}^2, \\
	& ||\mu(0,X)-\mu_0(0,X)||_{\P,\infty} \leq M\delta_{2n}, \quad ||\pi(X)-\pi_0(X)||_{\P,\infty} \leq M\delta_{2n} \\
	&  \sup_{ x_1 \in \overline{\mathcal{X}}_1 }\left\Vert(\mu(j,\cdot) - \mu_0(j,\cdot))||X_1 - x_1||_2^{1/2}\mathcal{K}_h^{1/2}\left(X_1-x_1\right)\right\Vert_{\P,2} \\
	& \times \left\Vert(\pi(\cdot) - \pi_0(\cdot))||X_1 - x_1||_2^{1/2}\mathcal{K}_h^{1/2}\left(X_1-x_1\right)\right\Vert_{\P,2} \leq M\delta_{3N}^2,
	\end{Bmatrix}, \quad j=0,1,
	\end{align*}
	and $\mathcal{A}_n(M) = \{(\widehat{\mu}(0,\cdot;I_k^c), \widehat{\pi}(\cdot;I_k^c))\in \mathcal{F}_{n}^{(0)}\} \cap\{(\widehat{\mu}(1,\cdot;I_k^c), \widehat{\pi}(\cdot;I_k^c))\in \mathcal{F}_{n}^{(1)}\} $. By Assumption \ref{ass:regularity}, for any $\eps>0$, there exists a positive constant $M$, such that $\P(\mathcal{A}_n(M)) \geq 1-\eps.$ Then, on $\mathcal{A}_n(M)$, we have
	\begin{align*}
	& \sup_{k \leq K, x_1 \in \overline{\mathcal{X}}_1}|I_k(x_1)| \\
	\leq & \sup_{k \leq K, x_1 \in \overline{\mathcal{X}}_1 }\biggl|(\P_{n,k} - \P_{I_k})\xi^*\biggl[\frac{D(Y-\widehat{\mu}(1,X;I_k^c))(\pi_0(X) - \widehat{\pi}( X;I_k^c))}{h^{d+2}\widehat{\pi}( X;I_k^c)\pi_0(X)}\biggr](X_1-x_1)\mathcal{K}_h\left(X_1-x_1\right)\biggr| \\
	& + \sup_{k \leq K,x_1 \in \overline{\mathcal{X}}_1 }\biggl|h^{-d-2}(\P_{n,k} - \P_{I_k})\xi^*\biggl[\biggl(1-\frac{D}{\pi_0(X)}\biggr)(\widehat{\mu}(1,X;I_k^c) - \mu_0(1,X)) \biggr] (X_1-x_1)     \mathcal{K}_h\left(X_1-x_1\right)\biggr| \\
	\equiv & \sup_{k \leq K, x_1 \in \overline{\mathcal{X}}_1 } I_{1k}(x_1)+  \sup_{k \leq K, x_1 \in \overline{\mathcal{X}}_1 }I_{2k}(x_1).
	\end{align*}
	Next, we aim to bound $I_{1k}(x_1)$. The upper bound for $I_{2k}(x_1)$ can be derived in the same manner, and thus, is omitted. By construction,
	$(\widehat{\mu}(1,\cdot;I_k^c),\widehat{\pi}(\cdot;I_k^c)) \indep \{W_i\}_{i \in I_k}$. This implies, conditional on $\{W_i,i\in I_k^c\}$, we can treat $(\widehat{\mu}(1,\cdot;I_k^c),\widehat{\pi}(\cdot;I_k^c))$ as fixed functions. We denote, for arbitrary fixed $(\mu(1,X),\pi(X)) \in \mathcal{F}_{n}^{(1)}$,
	$$\mathcal{H}_1(\mu(1,X),\pi(X)) = \biggl\{\xi^*\biggl[\frac{D(Y-\mu(1,X))(\pi_0(X) - \pi(X))}{h^{d+2}\pi(X)\pi_0(X)}\biggr](X_1-x_1)\mathcal{K}_h\left(X_1-x_1\right): x_1 \in \overline{\mathcal{X}}_1 \biggr\}.$$
	Then, $\mathcal{H}_{1}(\mu(1,X),\pi(X))$ has an envelope function
	\begin{align*}
	H_{1i}(\mu(1,X),\pi(X))  \lesssim |\xi^*Y|h^{-d-1}|\pi(X) - \pi_0(X)| \lesssim |\xi_i^*Y_i|h^{-d-1}\delta_{2n}
	\end{align*}
	and
	\begin{align}
	\sup_i H_{1i}(\mu(1,X),\pi(X)) \lesssim \sup_i |\xi^*_iY_i|h^{-d-1}\delta_{2n}.
	\label{eq:upbound}
	\end{align}
	
	Furthermore, for fixed functions $(\mu(1,X),\pi(X))$, we have
	\begin{equation}
	\begin{aligned}
	\sup_Q \log\left( N(\mathcal{H}_1(\mu(1,X),\pi(X)),||\cdot||_{Q,2},\eps||H_{1}(\mu(1,X),\pi(X))||_{Q,2})\right) \lesssim \log(1/\eps)\vee 0.
	\label{eq:entropy}
	\end{aligned}
	\end{equation}
	
	On $\mathcal{A}_n(M)$, we have
	\begin{equation*}
	\begin{aligned}
	& \sup_{x_1 \in \overline{\mathcal{X}}_1,(\mu(1,X),\pi(X)) \in \mathcal{F}_{n}^{(1)} }\E \biggl\Vert\xi^*\biggl[\frac{D(Y-\mu(1,X))(\pi_0(X) - \pi(X))}{h^{d+2}\pi(X)\pi_0(X)}\biggr](X_1-x_1)\mathcal{K}_h\left(X_1-x_1\right)\biggr\Vert_2^2\\
	\lesssim & h^{-2d-4}\sup_{x_1 \in \overline{\mathcal{X}}_1,(\mu(1,X),\pi(X)) \in \mathcal{F}_{n}^{(1)} }\E(\pi_0(X)-\pi(X))^2\Vert X_1-x_1\Vert_2^2\mathcal{K}_h\left(X_1-x_1\right) \\ \lesssim & h^{-d-2}\delta_{2n}^2.
	\end{aligned}
	\end{equation*}
	Therefore,
	\begin{equation}
	\begin{aligned}
	& \sup_{x_1 \in \overline{\mathcal{X}}_1}\P_{I_k} \biggl\{\xi^*\biggl[\frac{D(Y-\widehat{\mu}(1,X;I_k^c))(\pi_0(X) - \widehat{\pi}(X;I_k^c))}{h^d\widehat{\pi}(X;I_k^c)\pi_0(X)}\biggr]\mathcal{K}_h\left(X_1-x_1\right)\biggr\}^2 1\{\mathcal{A}_n(M)\}
	\lesssim \delta_{2n}^2 h^{-d}.
	\label{eq:sigma}
	\end{aligned}
	\end{equation}
	
	Then, \citet[Lemma C.1]{BCFH13} implies that
	\begin{align}
	\P_{I_k}\sup_{ x_1 \in \overline{\mathcal{X}}_1 }II_{1k}(x_1) 1\{\mathcal{A}_n(M)\} = & \E(\sup_{ x_1 \in \overline{\mathcal{X}}_1 }II_{1k}(x_1)|W_i,i\in I^c)1\{\mathcal{A}_n(M)\} \notag \\
	\lesssim & \sqrt{\frac{h^{-d-2}\delta_{2n}^2 \log(n)}{n}} + \frac{\Vert \sup_i |\xi^*_iY_i| \Vert_{P,q}\delta_{2n} \log(n)}{n h^{d+1}} \notag \\
	\lesssim & \sqrt{\frac{h^{-d-2}\delta_{2n}^2 \log(n)}{n}} + \frac{\delta_{2n} \log(n)n^{1/q}}{n h^{d+1}},\label{eq:II1}
	\end{align}
	where the second inequality holds by the fact that $\xi_i^*$ has sub-exponential tails. Hence, for an arbitrary $\eps_0>0$, as $n \rightarrow \infty$,
	\begin{align*}
	& \P(\sup_{ x_1 \in \overline{\mathcal{X}}_1 }I_{1k}(x_1) \geq \eps_0 (\log(n)nh^{d+2})^{-1/2}) \\
	\leq & \eps +  \P( \sup_{ x_1 \in \overline{\mathcal{X}}_1 }I_{1k}(x_1) 1\{\mathcal{A}_n(M)\} \geq \eps_0 (\log(n)nh^{d+2})^{-1/2}) \\
	\leq & \eps+\E\P_{I_k}(\sup_{ x_1 \in \overline{\mathcal{X}}_1 }I_{1k}(x_1)1\{\mathcal{A}_n(M)\} \geq \eps_0 (\log(n)nh^{d+2})^{-1/2}) \\
	\leq & \eps+\E\left\{\left[\frac{\P_{I_k}\sup_{ x_1 \in \overline{\mathcal{X}}_1 }I_{1k}(x_1)(\log(n)nh^d)^{1/2}}{\eps_0 }\right]1\{\mathcal{A}_n(M)\}\right\} \\
	\leq & \eps + C\left[\delta_{2n} \log(n) + \frac{\delta_{2n} \log^{3/2}(n)n^{1/q}}{(nh^{d})^{1/2}}\right]/\eps_0 \\
	\leq & 2\eps,
	\end{align*}
	where the first inequality is due to the union bound inequality, the second inequality is due to the law of iterated expectation and the fact that $\mathcal{A}_n$ belongs to the sigma field generated by $W_i,i\in I_k^c$, the third inequality is due to the Markov inequality, the fourth inequality is due to \eqref{eq:II1}, and the last inequality holds by Assumption \ref{ass:regularity}. Therefore,
	\begin{align}
	\sup_{k \leq K, x_1 \in \overline{\mathcal{X}}_1 }I_{1k}(x_1) = o_p((\log(n)nh^{d+2})^{-1/2}).
	\label{eq:II}
	\end{align}
	Similarly, we can show that
	\begin{align}
	\sup_{k \leq K, x_1 \in \overline{\mathcal{X}}_1 }I_{2k}(x_1) = o_p((\log(n)nh^{d+2})^{-1/2}).\label{eq:II_2}
	\end{align}
	\eqref{eq:II} and \eqref{eq:II_2} imply that
	$\sup_{k \leq K, x_1 \in \overline{\mathcal{X}}_1}|I_k(x_1)| = o_p((\log(n)nh^{d+2})^{-1/2}).$

	For $II_k(x_1)$, On $\mathcal{A}_n(M)$, we have
	\begin{align}
	& \sup_{k\leq K, x_1 \in \overline{\mathcal{X}}_1 }||II_k(x_1)||_2 \notag \\
	\leq & \sup_{k\leq K x_1 \in \overline{\mathcal{X}}_1 } \left\Vert\P_{I_k}\xi^* h^{-d-2}\left[\psi(1,W,\hat{\eta}(I_k^c)) - \psi(1,W,\eta_0) \right](X_1-x_1)\mathcal{K}_h\left(X_1-x_1\right)\right\Vert_2 \nonumber\\
	\leq & \sup_{k\leq K, x_1 \in \overline{\mathcal{X}}_1,\eta \in \mathcal{A}_n(M) } \left\Vert \P_{I_k}\xi^* h^{-d-2}\left[\psi(1,W,\eta) - \psi(1,W,\eta_0) \right](X_1-x_1)\mathcal{K}_h\left(X_1-x_1\right)\right\Vert_2 \nonumber \\
	\leq & \sup_{x_1 \in \overline{\mathcal{X}}_1,\eta \in \mathcal{A}_n(M) }\biggl\Vert \E\biggl[\frac{(\mu(1,X)-\mu_0(1,X))(\pi(X) - \pi_0(X))}{h^{d+2}\pi(X)}\biggr]\mathcal{K}_h(X_1-x_1)\left(X_1-x_1\right)\biggr\Vert_2 \nonumber\\
	\leq & \sup_{x_1 \in \overline{\mathcal{X}}_1,\eta \in \mathcal{A}_n(M) }h^{-d-2}\left\Vert (\hat\mu(j,X; I_k^c)-\mu_0(j,X))||X_1 - x_1||_2^{1/2}\mathcal{K}^{1/2}_h\left(X_1 - x_1\right)\right\Vert_{\P_{I_k},2}  \nonumber \\
	& \times \left\Vert (\hat\pi(X; I_k^c)-\pi_0(X))||X_1 - x_1||_2^{1/2}\mcalK_h^{1/2}\left(X_1-x_1\right)\right\Vert_{\P_{I_k},2}\nonumber \\
	\lesssim & h^{-d-2}\delta_{3n}^2.
	\label{eq:3}
	\end{align}
	Because $\eps$ is arbitrary, we have
	$$\sup_{k\leq K, x_1 \in \overline{\mathcal{X}}_1}||II_k(x_1)||_2 = O_p(h^{-d-2}\delta_{3n}^2) = o_p( (\log(n)nh^{d+2})^{-1/2}).$$
	This leads to \eqref{eq:R2k}. 
\end{proof}

\section{Monte Carlo simulations}\label{sec:sim}
In this section, we investigate the finite sample properties of the proposed high-dimensional CATE estimators using Monte Carlo experiments. The goal is to evaluate estimation accuracy and the validity of the proposed inference procedures. In accordance with the theoretical setup of the paper, we do not assume the knowledge of $\pi_0$ or $\mu_0$, but rather use variable selection techniques (specifically, lasso) to estimate these functions. We generate an initial pool of covariates whose size $p$ is comparable with, or even exceeds, the sample size $N$. Nevertheless, only a small subset of variables has significant effects. We consider two designs: (1) the data generating process (DGP) is strictly sparse (i.e., the number of covariates that actually enter the DGP is fixed and small), and (2) the DGP is approximately sparse in the sense that all $p$ covariates are relevant to some extent but only a few are truly important. Throughout this section we repeat the simulations 5000 times. For the cross-fitting estimator, we set $K=4$ as suggested by \cite{CC18}.



\subsection{Data-generating process}\label{ssec:dgp}

\noindent {\bf{DGP 1 (Strict sparsity):}}

This design is based on LOW. We specify the following linear model for the potential outcomes:
\begin{equation}\label{MC sim: Y(1)}
Y(1) = 10+\sum_{k=1}^{p}\beta_k X_k+\epsilon,
\end{equation}
and we set $Y(0)=0$. The covariates $\mathbf{X}=(X_1,X_2,\dots, X_p)'$ are drawn from the $N(0,I_p)$ distribution, where $I_p$ is the $p$-dimensional identity matrix. The error $\epsilon$ follows standard normal distribution and is independent of $\mathbf{X}$. We consider $p=100, 200, 500$, and sample sizes $N=500, 1000$. We set $\beta_i=1$, for $i=1,2,\dots, p_1$, and 0 for $i>p_1$. Therefore, only the first $p_1$ covariates are actually present in the conditional mean function. In our simulations we set $p_1=4$. 

The treatment status is determined by
\begin{equation}\label{MC sim: D}
D = \mathbf{l} \left \{ \Lambda \left ( \sum_{k=1}^{p}\gamma_k X_k \right ) > U \right \},
\end{equation}
where $U\sim \text{unif }(0,1)$, independent of $(\mathbf{X}, \epsilon)$, and $\Lambda(\cdot)$ is the logistic link function. The observed outcome is then $Y=DY(1)$. We set $\gamma_i=0.5$, for $i=1,2,\dots, p_1$, and 0 for $i>p_1$. We use the same value of $p_1$ as in the conditional mean function. 
The estimation target is $CATE(x_1)=E[Y(1)|X_1=x_1]$, given by 
\begin{equation}
CATE(x_1) = 10 + x_1.
\end{equation}

\noindent {\bf{DGP 2 (Approximate sparsity):}}

The strict sparsity assumption in DGP~1 may be too strong in applications. Given the functional forms specified above, approximate sparsity means that there are many small but nonzero coefficients $\beta_k$ and $\gamma_k$.  In DGP 2 we mimic this situation using a dwindling coefficients setup as in \cite{BCFH13}. Similarly to DGP 1, we define 
\begin{equation}\label{eq20}
\begin{split}
&Y(1)= \sum_{k=1}^{p}\beta_k X_k +\epsilon,\\
&D=\mathbf{l} \left \{ \Lambda \left ( \sum_{k=1}^{p}\gamma_k X_k \right ) > U \right \}, 
\end{split}
\end{equation}
where $\epsilon \sim N(0, 1)$, $U\sim \text{unif }(0,1)$, and  $\epsilon$ and $U$ are independent. We also set $Y(0)=0$. Nevertheless, the distribution of $X$ and the coefficients are generated differently. Specifically, we draw $X_1$ from the $N(0,1)$ distribution independently of the other covariates. The rest of the covariates are generated as $X_{i} \sim N(0, \Sigma)$ with $\Sigma_{kj} = (0.5)^{|j-k|}$, for $i=2, 3, \dots, p$. We consider $p=100, 200$ and $500$ and sample sizes $N = 500, 1000$. 

The model coefficients are specified in the following way. Let $\gamma_k = c_{d}\theta_{0,k}$, $\beta_k=c_{y}\theta_{0,k} $, where $\theta_{0}$ is a $p \times 1$ vector such that $\theta_{0,k} = (1/k)^{2}$ for $k = 1, ..., p$,  $c_{d}$ and $c_{y}$ are scalars that control the strength of the relationship between the controls, the outcome, and the treatment variable. We use several different combinations of $c_{d}$ and $c_{y}$ , setting $c_{d} = \sqrt{\frac{(\pi^{2}/3)R^{2}_{d}}{(1-R^{2}_{d})\theta'_{0}\Sigma\theta_{0}}}$ and $c_{y} = \sqrt{\frac{R^2_{y}}{(1-R^{2}_{d})\theta'_{0}\Sigma\theta_{0}}}$ for $R^{2}_{d} = R^{2}_{y} \in \{0.1, 0.5\}$.\footnote{The formulas for $c_{d}$ and $c_{y}$ ensure that the regressions under (\ref{eq20}) attain the given (pseudo) R$^2$ values. A higher value of R$^2$ means that the regressors are more informative about $D$ and $Y$.} 
The estimation target is again $CATE(x_1)=E[Y(1)|X_1=x_1]$, given by 
\begin{equation}
CATE(x_1) = \beta_1 x_1.
\end{equation}

The computation procedure is the following: we first estimate $\mu_0$ and $\pi_0$ using the
post lasso method of \cite{BCH14a}. After this step, $\hat
\beta$ and $\hat \gamma$ are sparse as some (most) coefficients are
estimated as exact zeros. Thus, this step also serves as the model selection
procedure, which is the advantage of a lasso type approach. In the second step, we substitute the estimated values of $\mu_0(j,X_i)$ and $\pi_0(X_i)$ into the score function $\psi$ and run kernel regressions of $\psi(W_i,\hat\eta)$ on $X_{1i}$ at various evaluation points $x_1$. In particular, we use a grid of 201 equally spaced points over the interval $\mathcal{X}_1=[-1,1]$, and estimate CATE at these points. We use the Gaussian kernel throughout.

Before we proceed, we briefly discuss bandwidth choice in practice. To obtaining our theoretical results, we require undersmoothing to eliminate bias asymptotically.  When $d=1$ as in the simulations, we suggest setting $h_N=\hat h \times N ^{1/5} \times N^{-2/7}$, where $\hat h=1.06\cdot \hat{\sigma}_{x_1} N^{-1/5}$ and $\hat{\sigma}_{x_1}$ is the estimated standard deviation of $X_1$. The formula for $\hat h$ corresponds to the rule-of-thumb bandwidth with a Gaussian kernel suggested by \cite{Silverman1986}.\footnote{When $d=2$ or $3$, we suggest setting for $j=1,\ldots,d$, $h_{jN}=\hat h_j \times N ^{1/(4+d)} \times N^{-2/(4+3d)}$ and $\hat h_j=1.06\cdot \hat{\sigma}_{x_{1j}} N^{-1/(4+d)}$ and $\hat{\sigma}_{x_{1j}}$ is the estimator of the standard deviation of the $j$-th element of $X_1$.} The bandwidth selection is done with the entire sample even for the cross-fitting method. Also the same selection method is employed in the empirical application.


We examine the coverage probability of the (2-sided) uniform confidence band for ${\text{CATE}(x_1)}$ over the grid described above. The nominal coverage probabilities ($1-\alpha$) are the standard 99\%, 95\% and 90\%. We compute the empirical coverage (EMP), the mean critical value (Mcri), and the standard deviation of critical value (Sdcri).\footnote{``Critical value'' refers to the statistic $\widehat{C}^{\text{2-sided}}_\alpha$, whose computation is discussed in Section \ref{sec: unif inference}.} We also compute the coverage probabilities of the confidence band based on the critical values computed by the Gumbel approximation (using the full sample) as in LOW. In addition, we report various statistics describing the properties of the estimates for $x_1 \in \{-1,-0.5,0,0.5,1\}$. These are the bias (BIAS), standard deviation (SD), the average estimated standard error of $\widehat{\text{CATE}}(x_1)$ (ASE), and the root mean squared error (RMSE). 

\subsection{Simulation Results}

\begin{table}[!ht]
	\centering
	\caption{DGP 1, high-dimensional CATE estimation}
	\label{T1}
	\renewcommand\arraystretch{0.6}
	\begin{tabular}{ccc|ccc|cccc}
		\hline
		& & & \multicolumn{3}{c|}{cross-fitting} & \multicolumn{3}{c}{full sample} & \\ \hline
		p & N & Conf. level & EMP & Mcri & Sdcri & EMP & Mcri & Sdcri & Gumbel \\ \hline
		\multirow{6}{*}{p=100} & \multirow{3}{*}{N=500} & 99\% & 0.984 & 3.291 & 0.093 & 0.980 & 3.292 & 0.104 & 1.000 \\
		& & 95\% & 0.947 & 2.755 & 0.097 & 0.944 & 2.753 & 0.106 & 0.998 \\
		& & 90\% & 0.880 & 2.475 & 0.113 & 0.874 & 2.476 & 0.120 & 0.996 \\ \cline{2-10}
		& \multirow{3}{*}{N=1000} & 99\% & 0.988 & 3.294 & 0.082 & 0.984 & 3.296 & 0.084 & 1.000 \\
		& & 95\% & 0.948 & 2.759 & 0.097 & 0.947 & 2.756 & 0.102 & 0.998 \\
		& & 90\% & 0.903 & 2.478 & 0.106 & 0.907 & 2.480 & 0.114 & 0.992 \\ \hline
		\multirow{6}{*}{p=200} & \multirow{3}{*}{N=500} & 99\% & 0.981 & 3.331 & 0.092 & 0.977 & 3.335 & 0.095 & 0.998 \\
		& & 95\% & 0.939 & 2.762 & 0.106 & 0.935 & 2.764 & 0.108 & 0.998 \\
		& & 90\% & 0.875 & 2.481 & 0.115 & 0.867 & 2.485 & 0.118 & 0.996 \\ \cline{2-10}
		& \multirow{3}{*}{N=1000} & 99\% & 0.991 & 3.286 & 0.077 & 0.994 & 3.289 & 0.081 & 1.000 \\
		& & 95\% & 0.953 & 2.747 & 0.090 & 0.956 & 2.749 & 0.093 & 0.998 \\
		& & 90\% & 0.895 & 2.475 & 0.101 & 0.890 & 2.474 & 0.103 & 0.994 \\ \hline
		\multirow{6}{*}{p=500} & \multirow{3}{*}{N=500} & 99\% & 0.979 & 3.415 & 0.095 &0.975 & 3.417 & 0.098 & 0.998 \\
		& & 95\% & 0.944 & 2.762 & 0.110 &0.939 & 2.766 & 0.114 & 0.998 \\
		& & 90\% & 0.886 & 2.481 & 0.124 &0.882 & 2.483 & 0.126 & 0.996 \\ \cline{2-10}
		& \multirow{3}{*}{N=1000} & 99\% & 0.986 & 3.296 & 0.077 & 0.983 & 3.294 & 0.079 & 1.000 \\
		& & 95\% & 0.945 & 2.752 & 0.091 & 0.942 & 2.748 & 0.096 & 1.000 \\
		& & 90\% & 0.893 & 2.476 & 0.101 & 0.889 & 2.479 & 0.103 & 0.998 \\ \hline
	\end{tabular}
	\begin{tablenotes}
		\setlength{\baselineskip}{7pt}
		\scriptsize
		\item \textit{Notes:} The nominal coverage probabilities that we consider are 99\%,
		95\% and 90\%. We compute the empirical coverage (``EMP"), the
		mean critical value (``Mcri"), and the standard deviation of critical value (``Sdcri").
		We also compute the coverage probabilities of
		the confidence band based on the critical values computed by the
		Gumbel approximation (``Gumbel").
	\end{tablenotes}
\end{table}

\begin{table}[!ht]
	\centering
	\caption{DGP1, high-dimensional CATE estimation}
	\renewcommand\arraystretch{0.6}
	\label{T2}
	\begin{tabular}{ccc|cccccccc}
		\hline
		& & & \multicolumn{8}{c}{HD-DR} \\ \cline{4-11}
		& & & \multicolumn{4}{c|}{cross-fitting} & \multicolumn{4}{c}{full sample} \\ \hline
		& & At x= & BIAS & SD & ASE & \multicolumn{1}{c|}{RMSE} & BIAS & SD & ASE & RMSE \\ \hline
		\multirow{10}{*}{p=100} & \multirow{5}{*}{N=500} & -1 & 0.018 & 0.258 & 0.251 & \multicolumn{1}{c|}{0.260} & 0.020 & 0.269 & 0.259 & 0.271 \\
		& & -0.5 & 0.008 & 0.214 & 0.205 & \multicolumn{1}{c|}{0.215} & 0.009 & 0.223 & 0.209 & 0.223 \\
		& & 0 & 0.007 & 0.201 & 0.182 & \multicolumn{1}{c|}{0.202} & 0.011 & 0.210 & 0.187 & 0.212 \\
		& & 0.5 & 0.010 & 0.207 & 0.184 & \multicolumn{1}{c|}{0.208} & 0.011 & 0.215 & 0.192 & 0.216 \\
		& & 1 & -0.019 & 0.237 & 0.215 & \multicolumn{1}{c|}{0.239} & -0.022 & 0.247 & 0.223 & 0.249 \\ \cline{2-11}
		& \multirow{5}{*}{N=1000} & -1 & 0.009 & 0.186 & 0.175 & \multicolumn{1}{c|}{0.188} & 0.011 & 0.192 & 0.181 & 0.195 \\
		& & -0.5 & 0.010 & 0.152 & 0.138 & \multicolumn{1}{c|}{0.153} & 0.012 & 0.157 & 0.145 & 0.159 \\
		& & 0 & -0.006 & 0.146 & 0.125 & \multicolumn{1}{c|}{0.146} & -0.008 & 0.150 & 0.131 & 0.151 \\
		& & 0.5 & 0.005 & 0.148 & 0.132 & \multicolumn{1}{c|}{0.148} & 0.010 & 0.157 & 0.136 & 0.158 \\
		& & 1 & -0.011 & 0.158 & 0.149 & \multicolumn{1}{c|}{0.160} & -0.013 & 0.169 & 0.161 & 0.171 \\ \hline
		\multirow{10}{*}{p=200} & \multirow{5}{*}{N=500} & -1 & 0.020 & 0.271 & 0.258 & \multicolumn{1}{c|}{0.273} & 0.022 & 0.294 & 0.262 & 0.296 \\
		& & -0.5 & -0.012 & 0.219 & 0.207 & \multicolumn{1}{c|}{0.220} & -0.011 & 0.226 & 0.214 & 0.228 \\
		& & 0 & 0.009 & 0.201 & 0.185 & \multicolumn{1}{c|}{0.202} & 0.011 & 0.210 & 0.193 & 0.211 \\
		& & 0.5 & 0.011 & 0.212 & 0.188 & \multicolumn{1}{c|}{0.214} & 0.013 & 0.217 & 0.195 & 0.219 \\
		& & 1 & -0.018 & 0.220 & 0.207 & \multicolumn{1}{c|}{0.221} & -0.021 & 0.229 & 0.216 & 0.231 \\ \cline{2-11}
		& \multirow{5}{*}{N=1000} & -1 & 0.011 & 0.183 & 0.165 & \multicolumn{1}{c|}{0.185} & 0.013 & 0.194 & 0.172 & 0.197 \\
		& & -0.5 & 0.008 & 0.148 & 0.128 & \multicolumn{1}{c|}{0.149} & 0.010 & 0.156 & 0.136 & 0.157 \\
		& & 0 & 0.004 & 0.130 & 0.118 & \multicolumn{1}{c|}{0.130} & 0.005 & 0.134 & 0.122 & 0.134 \\
		& & 0.5 & -0.003 & 0.131 & 0.115 & \multicolumn{1}{c|}{0.131} & -0.005 & 0.139 & 0.127 & 0.139 \\
		& & 1 & 0.013 & 0.159 & 0.131 & \multicolumn{1}{c|}{0.161} & 0.017 & 0.164 & 0.137 & 0.167 \\ \hline
		\multirow{10}{*}{p=500} & \multirow{5}{*}{N=500} & -1 & 0.028 & 0.260 & 0.243 & \multicolumn{1}{c|}{0.262} & 0.035 & 0.281 & 0.262 & 0.282 \\
		& & -0.5 & 0.006 & 0.233 & 0.209 & \multicolumn{1}{c|}{0.234} & 0.008 & 0.239 & 0.218 & 0.240 \\
		& & 0 & 0.004 & 0.196 & 0.184 & \multicolumn{1}{c|}{0.196} & 0.007 & 0.208 & 0.189 & 0.208 \\
		& & 0.5 & -0.012 & 0.207 & 0.185 & \multicolumn{1}{c|}{0.207} & -0.016 & 0.219 & 0.191 & 0.219 \\
		& & 1 & 0.021 & 0.229 & 0.208 & \multicolumn{1}{c|}{0.230} & 0.025 & 0.238 & 0.211 & 0.239 \\ \cline{2-11}
		& \multirow{5}{*}{N=1000} & -1 & -0.014 & 0.184 & 0.165 & \multicolumn{1}{c|}{0.186} & -0.019 & 0.191 & 0.174 & 0.192 \\
		& & -0.5 & 0.004 & 0.142 & 0.127 & \multicolumn{1}{c|}{0.143} & 0.006 & 0.151 & 0.139 & 0.152 \\
		& & 0 & 0.003 & 0.140 & 0.123 & \multicolumn{1}{c|}{0.141} & 0.004 & 0.146 & 0.129 & 0.147 \\
		& & 0.5 & -0.008 & 0.134 & 0.114 & \multicolumn{1}{c|}{0.135} & -0.012 & 0.142 & 0.117 & 0.143 \\
		& & 1 & 0.019 & 0.167 & 0.134 & \multicolumn{1}{c|}{0.169} & 0.021 & 0.172 & 0.141 & 0.175 \\ \hline
	\end{tabular}
	\begin{tablenotes}
		\setlength{\baselineskip}{7pt}
		\scriptsize
		\item \textit{Notes:}
		In both tables,  we estimate CATEF($x_{1}$) for $x_{1} \in \{-1,-0.5,0,0.5,1\}$
		and compute the mean bias (``BIAS"), standard deviation (``SD"), the average of standard error for CATEF($x_{1}$) (``ASE"), and the root mean squared error (``RMSE").
	\end{tablenotes}
\end{table}

Tables \ref{T1}-\ref{T2} show the simulation results for DGP 1 (strict sparsity).
In particular, Table \ref{T1} shows that the proposed uniform confidence band has very good finite sample coverage properties in this setting. For example, for $p=200$ and $N=500$, the uniform confidence band for the cross-fitting method covers $\tau_0$ over the $[-1,1]$ interval with  98.1\%, 93.9\% and 87.5\% coverage rates given the nominal probability of 99\%, 95\% and 90\%, respectively. When the sample size increases to $N=1000$, the corresponding empirical probabilities further improve to 99.1\%, 95.3\% and 89.5\%, respectively. The coverage rates are similarly accurate in all other cases, including for the full-sample estimator. Thus, our approach is capable of providing sound inference for the unknown $\tau_0$. 

The importance of conducting uniform inference properly is apparent from the average size of the critical values. For example, the pointwise critical value for $1-\alpha = 95\%$ is 1.96, while the critical value for the corresponding uniform confidence band is around 2.75. Hence, using pointwise confidence bands to draw inferences about the global properties of the function $\tau_0$ can be misleading. The Gumbel approximation, on the other hand, is shown to be conservative. 

Table \ref{T2} displays various statistics describing finite sample estimation accuracy. Both variants of the estimator exhibit small bias, and as the sample size increases, RMSE improves.\footnote{Our results show that the lasso method has good performance in selecting the important variables in the $\mu_0$ and $\pi_0$ functions. These results are available upon request. We also did simulations for kernel estimator instead of local linear estimators, and the kernel estimator results are presented in an earlier version of the paper available upon request.} The estimated standard errors show a small downward bias. Comparing the two variants, the differences between the RMSE values are small, but it is always the cross-fitting approach that comes out slightly better in each setting.  

\begin{table}[!ht]
	\centering
	\caption{ DGP 2, $R^2= 0.1$, HD CATE estimation}
	\label{T3}
	\renewcommand\arraystretch{0.6}
	\begin{tabular}{ccc|ccc|cccc}
		\hline
		& & & \multicolumn{3}{c|}{cross-fitting} & \multicolumn{3}{c}{full sample} & \\ \hline
		p & n & Conf.level & EMP & Mcri & Sdcri & EMP & Mcri & Sdcri & Gumbel \\ \hline
		\multirow{6}{*}{p=100} & \multirow{3}{*}{N=500} & 99\% & 0.984 & 3.290 & 0.089 & 0.982 & 3.293 & 0.090 & 1.000 \\
		& & 95\% & 0.942 & 2.754 & 0.103 & 0.936 & 2.751 & 0.106 & 1.000 \\
		& & 90\% & 0.886 & 2.473 & 0.117 & 0.879 & 2.475 & 0.117 & 0.996 \\ \cline{2-10}
		& \multirow{3}{*}{N=1000} & 99\% & 0.991 & 3.296 & 0.091 & 0.987 & 3.295 & 0.091 & 1.000 \\
		& & 95\% & 0.940 & 2.756 & 0.107 & 0.937 & 2.756 & 0.109 & 1.000 \\
		& & 90\% & 0.892 & 2.479 & 0.116 & 0.889 & 2.481 & 0.118 & 0.998 \\ \hline
		\multirow{6}{*}{p=200} & \multirow{3}{*}{N=500} & 99\% & 0.981 & 3.335 & 0.084 & 0.979 & 3.299 & 0.086 & 1.000 \\
		& & 95\% & 0.937 & 2.762 & 0.102 & 0.929 & 2.759 & 0.104 & 1.000 \\
		& & 90\% & 0.873 & 2.484 & 0.113 & 0.869 & 2.483 & 0.114 & 0.992 \\ \cline{2-10}
		& \multirow{3}{*}{N=1000} & 99\% & 0.987 & 3.291 & 0.077 & 0.985 & 3.292 & 0.082 & 1.000 \\
		& & 95\% &0.942  & 2.756 & 0.092 & 0.939 & 2.753 & 0.097 & 0.998 \\
		& & 90\% & 0.891 & 2.478 & 0.103 & 0.889 & 2.479 & 0.105 & 0.996 \\ \hline
		\multirow{6}{*}{p=500} & \multirow{3}{*}{N=500} & 99\% & 0.979 & 3.341 & 0.085 & 0.975 & 3.335 & 0.088 & 1.000 \\
		& & 95\% & 0.939 & 2.761 & 0.103 & 0.934 & 2.762 & 0.106 & 1.000 \\
		& & 90\% & 0.879 & 2.485 & 0.113 & 0.874 & 2.486 & 0.115 & 0.992 \\ \cline{2-10}
		& \multirow{3}{*}{N=1000} & 99\% & 0.986 & 3.296 & 0.096 & 0.983 & 3.297 & 0.096 & 1.000 \\
		& & 95\% & 0.941 & 2.756 & 0.113 & 0.936 & 2.758 & 0.116 & 1.000 \\
		& & 90\% & 0.886 & 2.486 & 0.124 & 0.879 & 2.485 & 0.127 & 0.996 \\ \hline
	\end{tabular}
	\begin{tablenotes}
		\scriptsize
		\setlength{\baselineskip}{7pt}
		\item \textit{Notes:} See notes to Table \ref{T1}.
	\end{tablenotes}
\end{table}
\begin{table}[!ht]
	\centering
	\renewcommand\arraystretch{0.6}
	\caption{DGP 2, $R^2= 0.1$, HD CATE estimation}
	\label{T4}
	\begin{tabular}{ccc|cccccccc}
		\hline
		& & & \multicolumn{8}{c}{HD-DR} \\ \cline{4-11}
		& & & \multicolumn{4}{c|}{cross-fitting} & \multicolumn{4}{c}{full sample} \\ \hline
		& & At x= & BIAS & SD & ASE & \multicolumn{1}{c|}{RMSE} & BIAS & SD & ASE & RMSE \\ \hline
		\multirow{10}{*}{p=100} & \multirow{5}{*}{N=500} & -1 & 0.025 & 0.186 & 0.169 & \multicolumn{1}{c|}{0.188} & 0.037 & 0.192 & 0.175 & 0.199 \\
		& & -0.5 & 0.014 & 0.143 & 0.128 & \multicolumn{1}{c|}{0.146} & 0.017 & 0.161 & 0.135 & 0.163 \\
		& & 0 & -0.005 & 0.109 & 0.102 & \multicolumn{1}{c|}{0.109} & -0.008 & 0.114 & 0.108 & 0.115 \\
		& & 0.5 & -0.012 & 0.112 & 0.109 & \multicolumn{1}{c|}{0.113} & -0.016 & 0.115 & 0.114 & 0.116 \\
		& & 1 & 0.016 & 0.123 & 0.107 & \multicolumn{1}{c|}{0.125} & 0.019 & 0.135 & 0.117 & 0.136 \\ \cline{2-11}
		& \multirow{5}{*}{N=1000} & -1 & 0.022 & 0.122 & 0.109 & \multicolumn{1}{c|}{0.124} & 0.026 & 0.132 & 0.113 & 0.134 \\
		& & -0.5 & 0.018 & 0.095 & 0.090 & \multicolumn{1}{c|}{0.097} & 0.021 & 0.103 & 0.094 & 0.107 \\
		& & 0 & 0.006 & 0.090 & 0.075 & \multicolumn{1}{c|}{0.091} & 0.007 & 0.093 & 0.081 & 0.095 \\
		& & 0.5 & -0.014 & 0.108 & 0.074 & \multicolumn{1}{c|}{0.109} & -0.018 & 0.112 & 0.080 & 0.114 \\
		& & 1 & 0.009 & 0.116 & 0.084 & \multicolumn{1}{c|}{0.118} & 0.012 & 0.121 & 0.091 & 0.126 \\ \hline
		\multirow{10}{*}{p=200} & \multirow{5}{*}{N=500} & -1 & 0.036 & 0.146 & 0.120 & \multicolumn{1}{c|}{0.149} & 0.041 & 0.154 & 0.122 & 0.157 \\
		& & -0.5 & 0.018 & 0.122 & 0.091 & \multicolumn{1}{c|}{0.125} & 0.021 & 0.130 & 0.095 & 0.132 \\
		& & 0 & 0.008 & 0.116 & 0.099 & \multicolumn{1}{c|}{0.118} & 0.009 & 0.119 & 0.080 & 0.121 \\
		& & 0.5 & -0.015 & 0.114 & 0.093 & \multicolumn{1}{c|}{0.116} & -0.017 & 0.116 & 0.082 & 0.118 \\
		& & 1 & 0.037 & 0.126 & 0.095 & \multicolumn{1}{c|}{0.129} & 0.043 & 0.130 & 0.091 & 0.132 \\ \cline{2-11}
		& \multirow{5}{*}{N=1000} & -1 & -0.030 & 0.123 & 0.108 & \multicolumn{1}{c|}{0.126} & -0.036 & 0.127 & 0.113 & 0.129 \\
		& & -0.5 & 0.013 & 0.104 & 0.089 & \multicolumn{1}{c|}{0.105} & 0.014 & 0.107 & 0.093 & 0.108 \\
		& & 0 & 0.004 & 0.090 & 0.076 & \multicolumn{1}{c|}{0.090} & 0.006 & 0.085 & 0.081 & 0.085 \\
		& & 0.5 & -0.012 & 0.105 & 0.084 & \multicolumn{1}{c|}{0.106} & -0.014 & 0.109 & 0.081 & 0.110 \\
		& & 1 & -0.024 & 0.113 & 0.088 & \multicolumn{1}{c|}{0.115} & -0.033 & 0.117 & 0.090 & 0.118 \\ \hline
		\multirow{10}{*}{p=500} & \multirow{5}{*}{N=500} & -1 & -0.041 & 0.195 & 0.177 & \multicolumn{1}{c|}{0.199} & -0.044 & 0.206 & 0.188 & 0.209 \\
		& & -0.5 & -0.023 & 0.145 & 0.123 & \multicolumn{1}{c|}{0.147} & -0.025 & 0.150 & 0.129 & 0.153 \\
		& & 0 & -0.009 & 0.119 & 0.108 & \multicolumn{1}{c|}{0.120} & -0.011 & 0.128 & 0.114 & 0.129 \\
		& & 0.5 & 0.020 & 0.126 & 0.115 & \multicolumn{1}{c|}{0.128} & 0.026 & 0.137 & 0.122 & 0.139 \\
		& & 1 & 0.035 & 0.131 & 0.102 & \multicolumn{1}{c|}{0.135} & 0.040 & 0.139 & 0.117 & 0.144 \\ \cline{2-11}
		& \multirow{5}{*}{N=1000} & -1 & 0.028 & 0.211 & 0.190 & \multicolumn{1}{c|}{0.215} & 0.030 & 0.221 & 0.202 & 0.224 \\
		& & -0.5 & 0.018 & 0.178 & 0.143 & \multicolumn{1}{c|}{0.180} & 0.020 & 0.184 & 0.151 & 0.187 \\
		& & 0 & -0.007 & 0.118 & 0.122 & \multicolumn{1}{c|}{0.120} & -0.008 & 0.124 & 0.126 & 0.125 \\
		& & 0.5 & -0.016 & 0.116 & 0.115 & \multicolumn{1}{c|}{0.119} & -0.021 & 0.121 & 0.120 & 0.125 \\
		& & 1 & 0.030 & 0.126 & 0.130 & \multicolumn{1}{c|}{0.129} & 0.038 & 0.131 & 0.134 & 0.136 \\ \hline
	\end{tabular}
	\begin{tablenotes}
		\setlength{\baselineskip}{7pt}
		\item \textit{Note.} See note to Table \ref{T2}.
	\end{tablenotes}
\end{table}

\begin{table}[!ht]
	\centering
	\caption{ DGP 2, $R^2= 0.5$, HD CATE estimation}
	\renewcommand\arraystretch{0.6}
	\label{T5}
	\begin{tabular}{ccc|ccc|cccc}
		\hline
		& & & \multicolumn{3}{c|}{cross-fitting} & \multicolumn{3}{c}{full sample} & \\ \hline
		p & N & Conf.level & EMP & Mcri & Sdcri & EMP & Mcri & Sdcri & Gumbel \\ \hline
		\multirow{6}{*}{p=100} & \multirow{3}{*}{N=500} & 99\% & 0.983 & 3.289 & 0.088 & 0.979 & 3.289 & 0.091 & 1.000 \\
		& & 95\% & 0.930 & 2.746 & 0.103 & 0.924 & 2.747 & 0.107 & 0.998 \\
		& & 90\% & 0.876 & 2.471 & 0.112 & 0.872 & 2.473 & 0.117 & 0.998 \\ \cline{2-10}
		& \multirow{3}{*}{N=1000} & 99\% & 0.987 & 3.315 & 0.105 & 0.983 & 3.314 & 0.107 & 1.000 \\
		& & 95\% & 0.941 & 2.782 & 0.122 & 0.938 & 2.781 & 0.129 & 1.000 \\
		& & 90\% & 0.880 & 2.505 & 0.134 & 0.876 & 2.504 & 0.140 & 0.984 \\ \hline
		\multirow{6}{*}{p=200} & \multirow{3}{*}{N=500} & 99\% & 0.970 & 3.297 & 0.095 & 0.964 & 3.292 & 0.097 & 0.998 \\
		& & 95\% & 0.926 & 2.754 & 0.111 & 0.920 & 2.753 & 0.112 & 0.994 \\
		& & 90\% & 0.869 & 2.482 & 0.120 & 0.861 & 2.478 & 0.122 & 0.990 \\ \cline{2-10}
		& \multirow{3}{*}{N=1000} & 99\% & 0.982 & 3.331 & 0.100 & 0.979 & 3.330 & 0.103 & 1.000 \\
		& & 95\% & 0.939 & 2.767 & 0.118 & 0.935 & 2.768 & 0.123 & 0.996 \\
		& & 90\% & 0.875 & 2.490 & 0.132 & 0.869 & 2.492 & 0.135 & 0.984 \\ \hline
		\multirow{6}{*}{p=500} & \multirow{3}{*}{N=500} & 99\% & 0.968 & 3.298 & 0.099 & 0.964 & 3.297 & 0.101 & 0.998 \\
		& & 95\% & 0.923 & 2.759 & 0.116 & 0.921 & 2.756 & 0.118 & 0.996 \\
		& & 90\% & 0.865 & 2.477 & 0.128 & 0.860 & 2.479 & 0.130 & 0.992 \\ \cline{2-10}
		& \multirow{3}{*}{N=1000} & 99\% & 0.980 & 3.304 & 0.102 & 0.972 & 3.305 & 0.103 & 1.000 \\
		& & 95\% & 0.936 & 2.771 & 0.117 & 0.929 & 2.770 & 0.121 & 0.998 \\
		& & 90\% & 0.872 & 2.496 & 0.131 & 0.866 & 2.497 & 0.133 & 0.994 \\ \hline
	\end{tabular}
	\begin{tablenotes}
		\setlength{\baselineskip}{7pt}
		\item \textit{Note.} See note to Table \ref{T1}.
	\end{tablenotes}
	
\end{table}

\begin{table}[!ht]
	\centering
	\caption{ DGP 2, $R^2= 0.5$, HD CATE estimation}
	\label{T6}
	\renewcommand\arraystretch{0.6}
	\begin{tabular}{ccc|cccccccc}
		\hline
		& & & \multicolumn{8}{c}{HD-DR} \\ \cline{4-11}
		& & & \multicolumn{4}{c|}{cross-fitting} & \multicolumn{4}{c}{full sample} \\ \hline
		& & At x= & BIAS & SD & ASE & \multicolumn{1}{c|}{RMSE} & BIAS & SD & ASE & RMSE \\ \hline
		\multirow{10}{*}{p=100} & \multirow{5}{*}{N=500} & -1 & 0.035 & 0.323 & 0.309 & \multicolumn{1}{c|}{0.326} & 0.040 & 0.329 & 0.316 & 0.331 \\
		& & -0.5 & 0.022 & 0.177 & 0.160 & \multicolumn{1}{c|}{0.179} & 0.025 & 0.183 & 0.164 & 0.185 \\
		& & 0 & -0.005 & 0.120 & 0.107 & \multicolumn{1}{c|}{0.121} & -0.007 & 0.124 & 0.112 & 0.126 \\
		& & 0.5 & 0.036 & 0.118 & 0.102 & \multicolumn{1}{c|}{0.119} & 0.039 & 0.122 & 0.108 & 0.124 \\
		& & 1 & 0.048 & 0.146 & 0.136 & \multicolumn{1}{c|}{0.148} & 0.051 & 0.151 & 0.136 & 0.153 \\ \cline{2-11}
		& \multirow{5}{*}{N=1000} & -1 & -0.029 & 0.230 & 0.199 & \multicolumn{1}{c|}{0.233} & -0.032 & 0.239 & 0.205 & 0.242 \\
		& & -0.5 & 0.019 & 0.137 & 0.117 & \multicolumn{1}{c|}{0.139} & 0.023 & 0.142 & 0.125 & 0.144 \\
		& & 0 & 0.005 & 0.089 & 0.075 & \multicolumn{1}{c|}{0.089} & 0.006 & 0.091 & 0.080 & 0.092 \\
		& & 0.5 & -0.021 & 0.093 & 0.077 & \multicolumn{1}{c|}{0.094} & -0.028 & 0.104 & 0.082 & 0.106 \\
		& & 1 & -0.023 & 0.111 & 0.089 & \multicolumn{1}{c|}{0.113} & -0.027 & 0.115 & 0.092 & 0.118 \\ \hline
		\multirow{10}{*}{p=200} & \multirow{5}{*}{N=500} & -1 & 0.047 & 0.326 & 0.262 & \multicolumn{1}{c|}{0.329} & 0.051 & 0.338 & 0.271 & 0.341 \\
		& & -0.5 & 0.032 & 0.191 & 0.164 & \multicolumn{1}{c|}{0.194} & 0.034 & 0.199 & 0.178 & 0.202 \\
		& & 0 & -0.009 & 0.117 & 0.104 & \multicolumn{1}{c|}{0.118} & -0.011 & 0.123 & 0.112 & 0.125 \\
		& & 0.5 & 0.060 & 0.125 & 0.122 & \multicolumn{1}{c|}{0.127} & 0.065 & 0.133 & 0.130 & 0.134 \\
		& & 1 & 0.066 & 0.134 & 0.119 & \multicolumn{1}{c|}{0.137} & 0.071 & 0.139 & 0.124 & 0.141 \\ \cline{2-11}
		& \multirow{5}{*}{N=1000} & -1 & 0.027 & 0.222 & 0.195 & \multicolumn{1}{c|}{0.225} & 0.030 & 0.234 & 0.201 & 0.237 \\
		& & -0.5 & 0.023 & 0.140 & 0.121 & \multicolumn{1}{c|}{0.142} & 0.026 & 0.147 & 0.127 & 0.150 \\
		& & 0 & -0.005 & 0.090 & 0.082 & \multicolumn{1}{c|}{0.091} & -0.007 & 0.096 & 0.090 & 0.098 \\
		& & 0.5 & -0.031 & 0.098 & 0.079 & \multicolumn{1}{c|}{0.099} & -0.035 & 0.101 & 0.082 & 0.103 \\
		& & 1 & -0.035 & 0.126 & 0.095 & \multicolumn{1}{c|}{0.129} & -0.039 & 0.132 & 0.099 & 0.135 \\ \hline
		\multirow{10}{*}{p=500} & \multirow{5}{*}{N=500} & -1 & -0.053 & 0.329 & 0.281 & \multicolumn{1}{c|}{0.332} & -0.057 & 0.341 & 0.288 & 0.344 \\
		& & -0.5 & -0.037 & 0.201 & 0.181 & \multicolumn{1}{c|}{0.204} & -0.040 & 0.210 & 0.185 & 0.213 \\
		& & 0 & -0.009 & 0.115 & 0.111 & \multicolumn{1}{c|}{0.116}  & -0.013 & 0.123 & 0.113 & 0.125 \\
		& & 0.5 & 0.048 & 0.138 & 0.124 & \multicolumn{1}{c|}{0.141}& 0.052 & 0.142 & 0.130 & 0.144 \\
		& & 1 & 0.085 & 0.146 & 0.112 & \multicolumn{1}{c|}{0.149} & 0.096 & 0.154 & 0.119 & 0.157 \\ \cline{2-11}
		& \multirow{5}{*}{N=1000} & -1 & -0.036 & 0.232 & 0.198 & \multicolumn{1}{c|}{0.235} & -0.040 & 0.240 & 0.206 & 0.243 \\
		& & -0.5 & -0.032 & 0.149 & 0.122 & \multicolumn{1}{c|}{0.151} & -0.037 & 0.160 & 0.131 & 0.162 \\
		& & 0 & 0.011 & 0.094 & 0.089 & \multicolumn{1}{c|}{0.095} & 0.014 & 0.099 & 0.093 & 0.100 \\
		& & 0.5 & 0.037 & 0.103 & 0.091 & \multicolumn{1}{c|}{0.104} & 0.042 & 0.106 & 0.086 & 0.108 \\
		& & 1 & 0.034 & 0.132 & 0.096 & \multicolumn{1}{c|}{0.135} & 0.036 & 0.137 & 0.098 & 0.140 \\ \hline
	\end{tabular}
	\begin{tablenotes}
		\setlength{\baselineskip}{7pt}
		\item \textit{Note.} See note to Table \ref{T2}.
	\end{tablenotes}
\end{table}

Tables \ref{T3}-\ref{T6} show the simulation results for DGP 2 (approximate sparsity),
for $R^2=$0.1 and 0.5 which controls the overall explanatory power of the covariates. The overall performance of the proposed estimators is satisfactory both in terms of empirical coverage and estimation accuracy. For example, in Table \ref{T3} ($R_d^2=R_y^2=0.1$), for $p=200$, $N=500$, the empirical coverage rate of the uniform confidence band for the cross-fitting estimator is 98.1\%, 93.7\% and 87.3\% given the respective nominal coverage rates 99\%, 95\% and 90\%. As sample size increases to $n=1000$, the corresponding figures improve to 98.7\%, 94.2\% and 89.1\%, respectively. In most cases, the difference between the nominal and the empirical coverage rate is somewhat larger than under strict sparsity, but inference is still very precise. 

Turning to Table \ref{T4} (estimation accuracy for $R^2=0.1$), one observes some of the same patterns as in the case of DGP~1. Specifically, the RMSE decreases with the sample size, and is slightly but consistently smaller for the cross-fitting estimator than for the full-sample estimator. The  estimated standard error is also slightly downward biased. Compared with DGP~1, the estimators have somewhat larger bias. 

Finally, Tables \ref{T5} and \ref{T6} show the results for DGP~2 with $R^2=0.5$. Generally speaking, less sparsity lead to slightly worse finite sample performance both in terms of coverage rates and estimation accuracy. While inference is still reasonably accurate, the RMSE of both estimators increase substantially, in some cases by as much as 50 percent or more. This is likely due to the fact that variable selection, i.e., obtaining a parsimonious approximation to the DGP, is more difficult when the DGP is less sparse. In case of $R^2=0.5$, the model coefficients are larger and shrink to zero slower compared to the case of $R^2=0.1$. Thus, it is hard to single out a handful of variables as the most important ones, and there can be substantial small sample variation in the set of regressors selected by the lasso. (The $R^2$ associated with DGP~1 is greater than 0.5, which shows that it is not the high $R^2$ itself that is the root of the problem but the increased difficulty of model selection in finite samples.) This finding is consistent with the theory presented in Section \ref{sec:asy}.  



\subsection{Simulation Results II: Comparison with LOW}\label{sec:sim2}
In this section, we compare the finite-sample performance of our method with the one proposed by LOW. Specifically, we consider the following three designs: 1) the exact simulation DGP setting of LOW. In this case we wish to show the robustness of our method in relatively low dimensions and non-sparse models. 2)  the exact simulation DGP setting of LOW, but with sparsity features,  and 3) our DGPs as shown in Section \ref{ssec:dgp}, but we now consider cases with p = 10, 30, as in LOW\footnote{To make a fair comparison with LOW, we use the value $p=10, 30$. Our paper proposes a machine learning first stage estimator, which can handle high-dimensional covariates. For instance, our method can handle cases in which $p$ is comparable or even greater than the sample size $n$, but it would not be feasible for LOW.}. We first show the simulation results for design 1. The DGP of LOW is
\begin{equation}\label{eqlow}
\begin{split}
&Y(1)= 10+ \sum_{k=1}^{p} \frac{1}{\sqrt p} X_k +\epsilon, \quad Y(0)=0, \\
&D=\mathbf{l} \left \{ \Lambda \left ( \sum_{k=p/2}^{p} \frac{1}{\sqrt{ p/2}}  X_k \right ) > U \right \}, 
\end{split}
\end{equation}
where the X's, $\epsilon$ and $U$ have the same probabilistic behaviors as our DGP described in Section \ref{ssec:dgp}.  And the $CATE(x_1)=10+x_1/\sqrt p$. Notice in this setting the model is non-sparse: for $p=10$, the coefficients for regression model and propensity score are 0.316 and 0.447, for $p=30$, the coefficients for regression model and propensity score are 0.183 and 0.258, respectively.

LOW reports the simulation results of four scenarios: 1) true propensity score and regression models; 2) true propensity score model but false regression model; 3) false propensity score model but true regression model, and 4) false propensity score and regression models. In this simulation study, we consider all the variables in the LOW model \eqref{eqlow} with no a priori knowledge of which variables are in the true model. And we compare our results with the aforementioned scenario 1 of LOW: true propensity score and regression models, with $\mu=\alpha_0+\sum_{k=1}^p\alpha_k X_k$ and  $\pi=\Lambda\left ( \beta_0+\sum_{k=1}^p\beta_k X_k \right )$. Notice this would be the best case scenario for LOW. Tables \ref{T7} shows the simulation results. From the table we see that HDCATE performs quite closely to LOW. This is due to the fact that the lasso method selects most influential variables. 

\begin{table}[H]
	\centering
	\caption{Simulation results: LOW's DGP, non-sparse model}
	\label{T7}
	\renewcommand\arraystretch{0.6}
	\begin{tabular}{ccc|ccc|ccc|cccc}
		\hline
		& & & \multicolumn{3}{c|}{cross-fitting} & \multicolumn{3}{c|}{full sample}& \multicolumn{3}{c}{LOW-Best} & \\ \hline
		p & N & Conf. level & EMP & Mcri & Sdcri & EMP & Mcri & Sdcri & EMP & Mcri & Sdcri & Gumbel \\ \hline
		\multirow{6}{*}{p=10} & \multirow{3}{*}{N=500} & 99\% & 0.980 & 3.291 & 0.091 & 0.975 & 3.293 & 0.102 & 0.986 & 3.290 & 0.091 & 0.999 \\
		& & 95\% & 0.935 & 2.752 & 0.108 & 0.930 & 2.753 & 0.111 & 0.939 & 2.750 & 0.107 & 0.999 \\
		& & 90\% & 0.885 & 2.477 & 0.120 & 0.883 & 2.479 & 0.123 & 0.887 & 2.474 & 0.118 & 0.995 \\ \cline{2-13}
		& \multirow{3}{*}{N=2000} & 99\% & 0.983 & 3.296 & 0.087 & 0.977 & 3.298 & 0.088 & 0.988 & 3.299 & 0.090 & 1.000 \\ 
		& & 95\% & 0.939 & 2.760 & 0.100 & 0.933 & 2.761 & 0.103 & 0.941 & 2.761 & 0.106 & 1.000 \\ 
		& & 90\% & 0.891 & 2.480 & 0.109 & 0.889 & 2.483 & 0.112 & 0.880 & 2.486 & 0.117 & 0.997 \\  \hline
		\multirow{6}{*}{p=30} & \multirow{3}{*}{N=500} & 99\% & 0.977 & 3.252 & 0.098 & 0.972 & 3.254 & 0.099 & 0.993 & 3.249 & 0.096 & 1.000 \\ 
		& & 95\% & 0.939 & 2.760 & 0.115 & 0.935 & 2.763 & 0.117 & 0.960 & 2.754  & 0.113 & 0.999 \\ 
		& & 90\% & 0.880 & 2.476 & 0.118 & 0.877 & 2.479 & 0.121 & 0.915  & 2.478  & 0.124  & 0.997 \\  \cline{2-13}
		& \multirow{3}{*}{N=2000} & 99\% & 0.984 & 3.286 & 0.077 & 0.980 & 3.289 & 0.081 & 0.990 &  3.296 & 0.085  & 1.000  \\ 
		& & 95\% & 0.945 & 2.747 & 0.090 & 0.942 & 2.749 & 0.093 & 0.950 & 2.757 & 0.100 & 1.000 \\ 
		& & 90\% & 0.891 & 2.475 & 0.101 & 0.885 & 2.477 & 0.103 & 0.893 & 2.481 & 0.110 & 0.998 \\ \hline
		
	\end{tabular}
	\begin{tablenotes}
		\setlength{\baselineskip}{7pt}
		\scriptsize
		\item \textit{Notes:} The nominal coverage probabilities that we consider are 99\%,
		95\% and 90\%. We report the simulation results of HDCATE for both cross-fitting and full-sample estimators. LOW-Best is the ``true propensity score model, true regression model" of LOW. We compute the empirical coverage (``EMP"), the
		mean critical value (``Mcri"), and the standard deviation of critical value (``Sdcri").
		We also compute the coverage probabilities of
		the confidence band based on the critical values computed by the
		Gumbel approximation (``Gumbel").
	\end{tablenotes}
\end{table}

For design 2, we use the same DGP as in LOW, but now introduce sparsity in a random setting, so that the truly important variables are not known a priori. Specifically, we set the probability of each $X_k$, $k=1, \dots , p$ to be the true variable being 0.6 and 0.5 for the regression model and the propensity model, respectively. We report the simulation results for our method and LOW (using the same model specifications in Table \ref{T7}) in Table \ref{T8}. It is shown that our method performs quite robust in this random sparsity case. When $p=10$, LOW does not perform well due to model misspecifications. Overall, our method provides a more robust solution to the inference of CATE even in relatively low dimensions with unknown sparsity structure in LOW model setting.

\begin{table}[!ht]
	\centering
	\caption{Simulation results: LOW's DGP, random sparse model}
	\label{T8}
	\renewcommand\arraystretch{0.6}
	\begin{tabular}{ccc|ccc|ccc|cccc}
		\hline
		& & & \multicolumn{3}{c|}{cross-fitting} & \multicolumn{3}{c|}{full sample}& \multicolumn{3}{c}{LOW} & \\ \hline
		p & N & Conf. level & EMP & Mcri & Sdcri & EMP & Mcri & Sdcri & EMP & Mcri & Sdcri & Gumbel \\ \hline
		\multirow{6}{*}{p=10} & \multirow{3}{*}{N=500} & 99\% & 0.973 &  3.289  &0.091 & 0.970 & 3.291 & 0.098 & 0.734  & 3.288   &  0.084 &  0.998   \\
		& & 95\% & 0.936 &  2.748  & 0.107 & 0.930 & 2.751 & 0.113 &  0.504 & 2.748 & 0.100 &0.996    \\
		& & 90\% & 0.876 &  2.472  & 0.117 & 0.871 & 2.474 & 0.119 &   0.352 & 2.471 &  0.109 &  0.982   \\ \cline{2-13}
		& \multirow{3}{*}{N=2000} & 99\% & 0.975 & 3.291 & 0.094 & 0.971 & 3.295 & 0.098 &  0.136 & 3.312  & 0.096 & 0.992    \\ 
		& & 95\% & 0.936 & 2.751 & 0.111 & 0.932 & 2.758 & 0.115 &  0.036 &  2.775 &0.113 &  0.904  \\ 
		& & 90\% & 0.882 & 2.475 & 0.122 & 0.877 & 2.479 & 0.125 &  0.016  &  2.502 & 0.124 &  0.786  \\  \hline
		\multirow{6}{*}{p=30} & \multirow{3}{*}{N=500} & 99\% & 0.968 & 3.293 & 0.102 & 0.964 & 3.298 & 0.104 & 0.988   & 3.301   & 0.113  &0.998 \\ 
		& & 95\% & 0.930 & 2.760 & 0.115 & 0.926 & 2.755 & 0.118 & 0.960 & 2.762  & 0.132   & 0.998  \\ 
		& & 90\% & 0.871 & 2.476 & 0.118 & 0.867 & 2.480 & 0.124 & 0.938  & 2.488 & 0.145   &0.996   \\  \cline{2-13}
		& \multirow{3}{*}{N=2000} & 99\% & 0.972 & 3.286 & 0.082 & 0.969 & 3.289 & 0.085 & 0.985 &  3.298 & 0.087 & 1.000  \\ 
		& & 95\% & 0.934 & 2.751 & 0.092 & 0.929 & 2.753 & 0.097 &  0.956&  2.759 & 0.103 & 1.000 \\ 
		& & 90\% & 0.878 & 2.480 & 0.103 & 0.873 & 2.483 & 0.107 &  0.884 &  2.484  & 0.113 &  0.998 \\ \hline
		
	\end{tabular}
	\begin{tablenotes}
		\setlength{\baselineskip}{7pt}
		\scriptsize
		\item \textit{Notes:} The nominal coverage probabilities that we consider are 99\%,
		95\% and 90\%. We report the simulation results of HDCATE for both cross-fitting and full-sample estimators. LOW is the double robust method of LOW using model specification of Table \ref{T7}. We compute the empirical coverage (``EMP"), the
		mean critical value (``Mcri"), and the standard deviation of critical value (``Sdcri").
		We also compute the coverage probabilities of
		the confidence band based on the critical values computed by the
		Gumbel approximation (``Gumbel").
	\end{tablenotes}
\end{table}

\begin{table}[!ht]
	\centering
	\caption{Simulation results: DGP 1}
	\label{T9}
	\renewcommand\arraystretch{0.6}
	\begin{tabular}{ccc|ccc|ccc|cccc}
		\hline
		& & & \multicolumn{3}{c|}{cross-fitting} & \multicolumn{3}{c|}{full sample}& \multicolumn{3}{c}{LOW} & \\ \hline
		p & N & Conf. level & EMP & Mcri & Sdcri & EMP & Mcri & Sdcri & EMP & Mcri & Sdcri & Gumbel \\ \hline
		\multirow{6}{*}{p=10} & \multirow{3}{*}{N=500} & 99\% & 0.982 & 3.293 & 0.081 & 0.978 & 3.285 & 0.086 & 0.284 &  3.283 &  0.078  &0.996  \\
		& & 95\% & 0.939 & 2.741 & 0.101 & 0.935 & 2.746 & 0.099 & 0.122  & 2.741 & 0.091 & 0.980  \\
		& & 90\% & 0.881 & 2.457 & 0.108 & 0.878 & 2.463 & 0.102 &  0.060 & 2.465&  0.100  & 0.956  \\ \cline{2-13}
		& \multirow{3}{*}{N=2000} & 99\% & 0.989 & 3.285 & 0.084 & 0.985 & 3.291 & 0.081 & 0.014 & 3.298  & 0.082  & 0.974  \\ 
		& & 95\% & 0.945 & 2.763 & 0.102 & 0.942 & 2.755 & 0.106 & 0.010  & 2.760 & 0.096 &0.830 \\ 
		& & 90\% & 0.893 & 2.482 & 0.104 & 0.889 & 2.488 & 0.109 & 0.003 & 2.485 & 0.106  &0.714 \\  \hline
		\multirow{6}{*}{p=30} & \multirow{3}{*}{N=500} & 99\% & 0.979 & 3.272 & 0.088 & 0.976 & 3.278 & 0.090 &  0.194  &3.290 & 0.085 &  1.000  \\ 
		& & 95\% & 0.936 & 2.752 & 0.109 & 0.930 & 2.760 & 0.112 &0.086 &  2.750  &0.101  &  0.988 \\ 
		& & 90\% & 0.870 & 2.470 & 0.109 & 0.867 & 2.479 & 0.116 & 0.044  & 2.474  & 0.111  &0.954   \\  \cline{2-13}
		& \multirow{3}{*}{N=2000} & 99\% & 0.985 & 3.291 & 0.080 & 0.980 & 3.286 & 0.083 & 0.002 &  3.299&  0.085  & 0.948   \\ 
		& & 95\% & 0.943 & 2.752 & 0.098 & 0.940 & 2.743 & 0.095 &0.000 & 2.760 & 0.100 & 0.766 \\ 
		& & 90\% & 0.892 & 2.479 & 0.106 & 0.887 & 2.480 & 0.109 &  0.000 &  2.485  &0.110 & 0.668  \\ \hline
		
	\end{tabular}
	\begin{tablenotes}
		\setlength{\baselineskip}{7pt}
		\scriptsize
		\item \textit{Notes:} The nominal coverage probabilities that we consider are 99\%,
		95\% and 90\%. We report the simulation results of HDCATE for both cross-fitting and full-sample estimators. LOW is the double robust method of LOW using model specification of Table \ref{T7}. We compute the empirical coverage (``EMP"), the
		mean critical value (``Mcri"), and the standard deviation of critical value (``Sdcri").
		We also compute the coverage probabilities of
		the confidence band based on the critical values computed by the
		Gumbel approximation (``Gumbel").
	\end{tablenotes}
\end{table}

\begin{table}[!ht]
	\centering
	\caption{Simulation results: DGP 2, $R^2=0.1$}
	\label{T10}
	\renewcommand\arraystretch{0.6}
	\begin{tabular}{ccc|ccc|ccc|cccc}
		\hline
		& & & \multicolumn{3}{c|}{cross-fitting} & \multicolumn{3}{c|}{full sample}& \multicolumn{3}{c}{LOW} & \\ \hline
		p & N & Conf. level & EMP & Mcri & Sdcri & EMP & Mcri & Sdcri & EMP & Mcri & Sdcri & Gumbel \\ \hline
		\multirow{6}{*}{p=10} & \multirow{3}{*}{N=500} & 99\% & 0.985 & 3.295 & 0.093 & 0.982 & 3.297 & 0.104 & 0.974 & 3.291& 0.087 & 0.996 \\
		& & 95\% & 0.941 & 2.748 & 0.103 & 0.940 & 2.750 & 0.108 & 0.926 & 2.751 & 0.102 & 0.996  \\
		& & 90\% & 0.885 & 2.471 & 0.116 & 0.883 & 2.473 & 0.119 & 0.872  & 2.475  & 0.112  & 0.996  \\ \cline{2-13}
		& \multirow{3}{*}{N=2000} & 99\% & 0.987 & 3.290 & 0.093 & 0.983 & 3.294 & 0.095 & 0.984  &3.305  & 0.090  & 1.000  \\ 
		& & 95\% & 0.943 & 2.756 & 0.104 & 0.941 & 2.763 & 0.108 & 0.912 & 2.768 & 0.106 &0.998  \\ 
		& & 90\% & 0.891 & 2.488 & 0.112 & 0.887 & 2.485 & 0.115 & 0.856 &  2.494  &0.117 &0.990   \\  \hline
		\multirow{6}{*}{p=30} & \multirow{3}{*}{N=500} & 99\% & 0.983 & 3.273 & 0.094 & 0.980 & 3.279 & 0.092 & 0.962  & 3.288  &0.077 & 1.000 \\ 
		& & 95\% & 0.942 & 2.757 & 0.111 & 0.939 & 2.761 & 0.113 & 0.890  & 2.747  & 0.091 &1.000  \\ 
		& & 90\% & 0.880 & 2.479 & 0.107 & 0.877 & 2.475 & 0.110 &0.782  & 2.471  & 0.101 & 0.994   \\  \cline{2-13}
		& \multirow{3}{*}{N=2000} & 99\% & 0.988 & 3.292 & 0.082 & 0.984 & 3.287 & 0.086 &  0.754 &  3.303 & 0.087  & 1.000  \\ 
		& & 95\% & 0.945 & 2.751 & 0.096 & 0.941 & 2.755 & 0.098 & 0.494  & 2.766  & 0.103 & 0.998  \\ 
		& & 90\% & 0.894 & 2.480 & 0.107 & 0.891 & 2.485 & 0.109 &  0.352  & 2.492 & 0.113 & 0.984  \\ \hline
		
	\end{tabular}
	\begin{tablenotes}
		\setlength{\baselineskip}{7pt}
		\scriptsize
		\item \textit{Notes:} The nominal coverage probabilities that we consider are 99\%,
		95\% and 90\%. We report the simulation results of HDCATE for both cross-fitting and full-sample estimators. LOW is the double robust method of LOW using model specification of Table \ref{T7}. We compute the empirical coverage (``EMP"), the
		mean critical value (``Mcri"), and the standard deviation of critical value (``Sdcri").
		We also compute the coverage probabilities of
		the confidence band based on the critical values computed by the
		Gumbel approximation (``Gumbel").
	\end{tablenotes}
\end{table}

Tables \ref{T9} and \ref{T10} show the simulation results for DGPs 1 and 2 in Section \ref{ssec:dgp}, respectively. Our method is robust to different types of sparsity (strict and approximate) and performs well in the relatively low-dimensional models with $p=10,30$. LOW has serious size distortions in both cases.

\bibliographystyle{chicago}
\bibliography{cate_cf}

\end{document}